\documentclass[aps,pra,reprint,longbibliography,superscriptaddress,showkeys,10pt]{revtex4-1}
\usepackage{lipsum}
\usepackage{hyperref} % For hyperlinks in the PDF
\hypersetup{colorlinks=true,linkcolor=blue,urlcolor=blue,citecolor=blue}
\usepackage[textsize=scriptsize, backgroundcolor=yellow, textwidth=1.9\columnsep]{todonotes}
\usepackage{mathtools,amsmath,amsfonts,amssymb,amsbsy}  % Math packages
\usepackage{graphicx}
\usepackage{multirow}
\usepackage[noend]{algpseudocode}
\usepackage{newfloat}
\usepackage{setspace}
\usepackage{enumitem}

\usepackage{cleveref}

\usepackage{qtree}

\usepackage{xpatch}
\xpretocmd{\eqref}{Eq.~}{}{}

\DeclareFloatingEnvironment[
	fileext=loa,
	listname=List of Algorithms,
	name=ALG.,
	placement=tbhp,
]{algorithm}

\usepackage{braket}
\newcommand{\ketbra}[1]{\ket{#1}\bra{#1}}

\newcommand{\bbra}[1]{\langle \! \bra{#1}}
\newcommand{\kket}[1]{\ket{#1} \! \rangle}
\newcommand{\bbraket}[1]{\langle \! \braket{#1} \! \rangle}
\newcommand{\kketbra}[1]{\kket{#1}\bbra{#1}}

\usepackage{amsthm}
\usepackage{aliascnt}

\newtheoremstyle{plain}
{10pt}   % ABOVESPACE
{10pt}   % BELOWSPACE
{\itshape}  % BODYFONT
{0pt}       % INDENT (empty value is the same as 0pt)
{\bfseries} % HEADFONT
{.}         % HEADPUNCT
{5pt plus 1pt minus 1pt} % HEADSPACE
{}

\newtheoremstyle{definition}
{10pt}   % ABOVESPACE
{10pt}   % BELOWSPACE
{}  % BODYFONT
{0pt}       % INDENT (empty value is the same as 0pt)
{\bfseries} % HEADFONT
{.}         % HEADPUNCT
{5pt plus 1pt minus 1pt} % HEADSPACE
{}  

\makeatletter
\renewenvironment{proof}[1][\proofname]{\par
	\vspace{-5pt}% remove the space after the theorem
	\pushQED{\qed}%
	\normalfont
	\topsep0pt \partopsep0pt % no space before
	\trivlist
	\item[\hskip\labelsep
	\itshape
	#1\@addpunct{.}]\ignorespaces
}{%
	\popQED\endtrivlist\@endpefalse
	\addvspace{5pt plus 0pt} % some space after
}
\makeatother
\makeatletter
\def\@endtheorem{\endtrivlist}
\makeatother

\theoremstyle{plain}
\newtheorem{thm}{Theorem}

\newtheorem{lemma}{Lemma}

\newaliascnt{proposition}{lemma}
\newtheorem{proposition}[proposition]{Proposition}

\aliascntresetthe{proposition} 
\newtheorem*{corollary}{Corollary}

\theoremstyle{definition}

\theoremstyle{definition}
\newtheorem{exmp}{Example}
\newenvironment{example}{\pushQED{\qed}\exmp}{\popQED\endexmp}

\usepackage{tabularx}
\newcolumntype{L}[1]{>{\hsize=#1\hsize\raggedright\arraybackslash}X}%
\newcolumntype{R}[1]{>{\hsize=#1\hsize\raggedleft\arraybackslash}X}%
\newcolumntype{C}[1]{>{\hsize=#1\hsize\centering\arraybackslash}X}%

\usepackage{booktabs}
\newcommand{\ra}[1]{\renewcommand{\arraystretch}{#1}}  % for adjusting row spacing
\AtBeginDocument{
	\heavyrulewidth=.1em
	\lightrulewidth=.05em
	\cmidrulewidth=.03em
	\belowrulesep=.65ex
	\belowbottomsep=0pt
	\aboverulesep=.4ex
	\abovetopsep=0pt
	\cmidrulesep=\doublerulesep
	\cmidrulekern=.5em
	\defaultaddspace=.5em
}

\graphicspath{{./figs/}}

\newcommand{\Bc}{\ensuremath{\mathcal{B}}}     	% B
\newcommand{\Ec}{\ensuremath{\mathcal{E}}}		% E
\newcommand{\Gc}{\ensuremath{\mathcal{G}}}		% G 
\newcommand{\Hc}{\ensuremath{\mathcal{H}}}		% H
\newcommand{\Ic}{\ensuremath{\mathcal{I}}}		% I
\newcommand{\Mc}{\ensuremath{\mathcal{M}}}		% N
\newcommand{\Nc}{\ensuremath{\mathcal{N}}}		% N
\newcommand{\Oc}{\ensuremath{O}}		% O -- Regular O according to PRA style
\newcommand{\Uc}{\ensuremath{\mathcal{U}}}		% U
\newcommand{\Zc}{\ensuremath{\mathcal{Z}}}		% Z

\newcommand{\Gcbf}{\ensuremath{\pmb{\Gc}}}		        
\newcommand{\Gavgbf}[1]{\ensuremath{\pmb{\Gc}_\mathrm{avg}^{(#1)}}}
\newcommand{\Mcbf}{\ensuremath{\pmb{\Mc}}}
\newcommand{\Ncbf}{\ensuremath{\pmb{\Nc}}}
\newcommand{\Lambdabf}{\ensuremath{\pmb{\Lambda}}}

\newcommand{\C}{\ensuremath{\mathbb{C}}}
\newcommand{\Ebb}{\ensuremath{\mathbb{E}}}
\newcommand{\N}{\ensuremath{\mathbb{N}}}
\newcommand{\Mbb}{\ensuremath{\mathbb{M}}}
\newcommand{\Pbb}{\ensuremath{\mathbb{P}}}
\newcommand{\R}{\ensuremath{\mathbb{R}}}
\newcommand{\Vbb}{\ensuremath{\mathbb{V}}}

\newcommand{\jbf}{\ensuremath{\mathbf{j}}}
\newcommand{\ideal}{\ensuremath{\mathrm{id}}}
\newcommand{\spam}{\ensuremath{\mathrm{err}}}

\newcommand{\diff}{\ensuremath{\; \mathrm{d}}}
\newcommand{\urm}{\ensuremath{{\! \mathrm{u}}} }
\newcommand{\half}{\ensuremath{\frac{1}{2}}}
\newcommand{\Gavg}[1]{\ensuremath{\Gc_\mathrm{avg}^{(#1)}}}
\newcommand{\tp}[1]{{\otimes #1}}
\newcommand{\adj}{\ensuremath{\mathrm{[adj]}}}

\DeclareMathOperator{\Tr}{Tr}

\DeclareMathOperator{\Rge}{Rge}

\DeclareMathOperator{\Span}{Span}

\newcommand{\Cliff}[1]{\ensuremath{{\mathsf{C}(#1)}}}	% Cliff to C sf
\newcommand{\GLc}{\ensuremath{\mathsf{GL}}}	            % GL Caligraphic to sf
\newcommand{\Gbb}{\ensuremath{\mathsf{G}}}              % G Blackboard to sf (used for groups)
\newcommand{\Unitary}[1]{\ensuremath{{\mathsf{U}(#1)}}} % Unitary(d) to U(d) sf
\newcommand{\Pc}{\ensuremath{\mathsf{P}}}				% P Caligraphic to sf (pauli matrices)
\newcommand{\Lc}{\ensuremath{\mathsf{L}}}				% L	Caligraphic to sf (linear operators)

\allowdisplaybreaks

\begin{document}
%----------------------------------------------------------------------------------------
%	TITLE SECTION
%----------------------------------------------------------------------------------------
\title{Efficient Unitarity Randomized Benchmarking of Few-qubit Clifford Gates}
\date{\today}

\author{Bas Dirkse}
\email{b.dirkse@tudelft.nl}
\affiliation{QuTech, Delft University of Technology, Lorentzweg 1, 2628 CJ Delft, The Netherlands}
\affiliation{QuSoft, CWI and University of Amsterdam, Science Park 123 1098 XG Amsterdam, The Netherlands}
\author{Jonas Helsen}
\affiliation{QuTech, Delft University of Technology, Lorentzweg 1, 2628 CJ Delft, The Netherlands}
\author{Stephanie Wehner}
\affiliation{QuTech, Delft University of Technology, Lorentzweg 1, 2628 CJ Delft, The Netherlands}

\begin{abstract}
	Unitarity randomized benchmarking (URB) is an experimental procedure for estimating the coherence of implemented quantum gates independently of state preparation and measurement errors. These estimates of the coherence are measured by the unitarity. A central problem in this experiment is relating the number of data points to rigorous confidence intervals. 
	In this work we provide a bound on the required number of data points for Clifford URB as a function of confidence and experimental parameters. This bound has favorable scaling in the regime of near-unitary noise and is asymptotically independent of the length of the gate sequences used. We also show that, in contrast to standard randomized benchmarking, a nontrivial number of data points is always required to overcome the randomness introduced by state preparation and measurement errors even in the limit of perfect gates. Our bound is sufficiently sharp to benchmark small-dimensional systems in realistic parameter regimes using a modest number of data points. For example, we show that the unitarity of single-qubit Clifford gates can be rigorously estimated using few hundred data points under the assumption of gate-independent noise. This is a reduction of orders of magnitude compared to previously known bounds.
\end{abstract}

%\keywords{Randomized Benchmarking; Unitarity; Noise characterization}

\maketitle

%----------------------------------------------------------------------------------------
%	SECTION 1: Introduction
%----------------------------------------------------------------------------------------

\section{Introduction}
In order to further advance the efforts in building large-scale quantum computers, it is essential to characterize the errors of elementary quantum gates in practical implementations. Randomized benchmarking (RB) \cite{Emerson2005, Knill2008, Magesan2011, Magesan2012a} has in the past years become the standard for assessing the quality of quantum gates \cite{Knill2008,Chow2010,Olmschenk2010,Gaebler2012,Barends2014,Muhonen2015,Xia2015}. This is because RB has a simple and efficiently scalable implementation that characterizes gates errors independently of any state preparation and measurement (SPAM) errors. Since the introduction of randomized benchmarking, several variants have been developed \cite{Magesan2012,Wallman2015a,Wallman2015b,Wallman2016,Combes2017}. One of these variants is unitarity randomized benchmarking (URB) \cite{Wallman2015a,Feng2016}.

This paper is concerned with the URB protocol proposed in \cite{Wallman2015a}. It provides a method to characterize the coherence of errors in implemented quantum gates that is robust against SPAM errors. This characterization of coherence is quantified by the unitarity, a quantity that is independent of the average gate fidelity measured by standard RB. Being able to estimate the unitarity experimentally provides an extra source of information when optimizing experimental implementations of quantum gates \cite{Feng2016}. In particular, the unitarity can help to discriminate whether the dominant error process is coherent (i.e., overrotation or calibration errors) or incoherent (i.e., depolarizing or dephasing noise). This information is useful since these two different types of noise are generally reduced in different ways \cite{Feng2016, Sheldon2016}.
Additionally, knowing the unitarity of a gate or gate set can be used to get sharper bounds on the credible interval of an interleaved randomized benchmarking experiment \cite{Carignan-Dugas2016} and also get improved bounds on the diamond norm error \cite{Sanders2016,Kueng2016,Wallman2015}, which is the relevant metric in the setting of fault-tolerant quantum computing. 

The URB protocol is similar to the standard RB protocol and they share many characteristics, like SPAM independent estimation of its figure of merit. It aims only to provide a partial characterization of the gate set (by estimating the unitarity), instead of characterizing the noise completely, which is what channel or gate set tomography for instance aim to do. Since full tomography with rigorous confidence intervals is very resource-intensive \cite{Thinh2018}, in situations where partial noise characterization suffices, more lightweight solutions like RB and URB may be the choice of preference. 

In RB-type protocols, the noise-characterizing figure of merit is obtained from the exponential decay rate of the average survival probability with the length of the sequence of gates. For fixed sequence length, the average survival probability is estimated by averaging over a number of randomly sampled gate sequences. An important problem for RB-type procedures is then determining a number of random gate sequences that is practical yet yields a confident estimate of the figure of merit. This problem was realized in the first concrete proposal of RB \cite{Magesan2012a}. Subsequent work focused on resolving this problem in two different, complementary ways. First, statistical tools were applied to allow for confident estimation of the RB decay rate with fewer random gate sequences \cite{Epstein2014,Granade2015,Hincks2018}. Second, the underlying distribution from which the RB protocol samples data was analyzed. In particular a sharp bound on the variance of this distribution was derived, which also allows for more resource-efficient estimation of the RB decay rate from measurement data \cite{Wallman2014,Helsen2017}. However, no such analysis exists for the related URB protocol. 

Here we analyze the statistics of unitarity randomized benchmarking. The aim of this work is to contribute a solution to the following central question: How many random sequences of gates are required in the URB protocol to get a confident estimate of the unitarity from the obtained measurement data? We proceed along the lines of \cite{Wallman2014,Helsen2017} by providing a sharp bound on the variance of the underlying distribution from which the URB protocol samples. This additional knowledge of the URB sampling distribution allows for more resource-efficient estimation of the unitarity from experimental data. Concretely we demonstrate how our variance bound can be used to bound the required number of random sequences as a function of desired confidence parameters. 

In this work, we derive a bound on the variance of the distribution induced by the random sampling of gate sequences in a modified version of the Clifford URB protocol. This modification is based on the adapted RB protocol of \cite{Helsen2017}. It requires no experimental overhead while leading to a sharper variance bound (and hence fewer required gate sequences) as well as a simpler fit model for extracting the unitarity. In addition, our statistical analysis reveals the optimal input state and output measurement for minimizing the variance and maximizing the signal strength. 
We then apply this variance bound using standard concentration inequalities to relate the number of random sequences to desired confidence intervals. 
Our result is sufficiently sharp to perform the modified URB protocol on few-qubit systems with a modest number of sequences in realistic parameter regimes. It is an improvement of several orders of magnitude in the number of sequences required for fixed confidence, compared to a concentration inequality that does not use the variance (as was first done for RB in \cite{Magesan2012a}). 
We show that the variance, and thus number of required gate sequences, scales favorably  in the regime of large unitarity, which is the relevant regime for high quality gates. We also show that, in contrast to standard RB \cite{Helsen2017}, a nontrivial number of sequences is always required to overcome the randomness introduced by state preparation and measurement errors even in the limit of perfect gates.

This paper is organized as follows. In the remainder of this section we review the concept of unitarity and the URB protocol to estimate the unitarity of a gate set. We introduce a modification of the protocol based on \cite{Helsen2017} for the purpose of improved statistics. Furthermore we explicitly distinguish the two different implementations of the URB protocol and emphasize their benefits and drawbacks.  In \autoref{sec:results} we present our main result (\eqref{eq:hoeffding2} and \eqref{eq:var_bound}) and illustrate how to apply it using a simulated example. In \autoref{sec:discussion} we examine the behavior of our bound in various parameter regimes and discuss the different features of our bound. A brief overview of the proof techniques used to derive our main result is presented in \autoref{sec:methods}. All technical details of the proof  have been delegated to the appendices. In \autoref{sec:Outlook} we summarize the main conclusions of our work and provide suggestions for future research.

\subsection{Unitarity}
Let us begin with defining the figure of merit that URB estimates. For a quantum channel $\Ec$ (here a quantum channel will refer to a completely positive and trace-preserving (CPTP) superoperator), the unitarity  is defined as \cite{Wallman2015a}
	\begin{equation}
u(\Ec) = \frac{d}{d - 1} \int \diff \psi \Tr \left[ \left(\Ec\Big(\ketbra{\psi} - \frac{I}{d}\Big)\right)^2 \right],
\end{equation}
where the integration is with respect to the uniform Haar measure on the state space $\Hc$. The prefactor is chosen such that $0 \leq u \leq 1$. An equivalent definition of the unitarity can be given as \cite[Proposition 1]{Wallman2015a}
\begin{equation}
\label{eq:unitary_definition2}
u(\Ec) = \frac{1}{d^2 - 1} \sum_{\sigma,\tau \in \Pc^*} \Tr[\tau \Ec(\sigma)]^2,
\end{equation}
where the summation is over the set of all nonidentity, normalized Pauli matrices $\Pc^*$. The normalization is with respect to the Hilbert-Schmidt norm $\| \sigma \|_2 = \sqrt{\Tr[\sigma^\dagger \sigma]}$. 
This alternative definition of the unitarity is often more pleasant to work with. In \autoref{ex:unitarity} the unitarity of a depolarizing channel is calculated.

The unitarity has some properties that one would intuitively expect a good measure of the coherence of gates to have \cite[Proposition 7]{Wallman2015a}. First, $u = 1$ if and only if $\Ec$ is a unitary quantum channel. Second, the unitarity is invariant under unitary transformation. That is, if $\Uc,\mathcal{V}$ are unitary quantum channels, then $u(\Ec) = u(\Uc \Ec \mathcal{V})$. The unitarity is independent of but related to the average gate fidelity. In fact, the unitarity provides an upper bound on the average gate fidelity \cite[Proposition 8]{Wallman2015a},
\begin{equation}
\label{eq:agf}
\left(\frac{d F_\mathrm{avg} - 1}{d-1}\right)^2 \leq u.
\end{equation}
Here $F_\mathrm{avg}$ is the average gate fidelity between the implemented gate and the ideal target gate. This relation expresses the fact that a perfect gate ($F_\mathrm{avg} = 1$) must be unitary ($u=1$). However, the converse does not hold. Indeed, a unitary gate ($u=1$) can have arbitrary average gate fidelity by considering purely unitary noise (i.e., overrotation). The inequality \eqref{eq:agf} is tight, since it holds with equality for a depolarizing channel. 

\begin{example}
	\label{ex:unitarity}
	Let $\Ec$ be a depolarizing quantum channel with depolarizing parameter $p$
	$$\Ec: A \mapsto p A + \frac{1-p}{d} \Tr [A] I.$$
	Then the unitarity $u$ of $\Ec$ is computed using \eqref{eq:unitary_definition2} as
	\begin{equation*}
	u = \frac{1}{d^2 - 1} \sum_{\sigma,\tau \in \Pc^*} \Tr[p \tau^\dagger  p \sigma]^2 = p^2,
	\end{equation*}
	since $\Tr[\tau^\dagger \sigma] = \delta_{\sigma,\tau}$. Note that $F_\mathrm{avg}(\Ec) = p + \frac{1-p}{d}$, so that the inequality \eqref{eq:agf} is saturated by the depolarizing channel.
\end{example}

\subsection{The URB protocol}
\label{subsec:protocol}

\begin{algorithm*}
	\framebox[0.95\linewidth]{
		\linespread{1.15}
		\begin{minipage}{0.9\linewidth}
			\hspace{5pt}
			\normalsize
			\algrenewcommand\algorithmicloop{\textbf{repeat }}
			\algrenewcommand\algorithmicindent{1.0em}
			\setlength{\abovedisplayskip}{0pt}
			\setlength{\belowdisplayskip}{0pt}
			\begin{algorithmic}[1]
				%				\Statex{\textbf{data:} Let $\Gbb \subset U(d)$ be a finite subset of the unitary group that is also a unitary 2-design on a $d$-dimensional Hilbert space $\mathcal{H}$ . Let the noise model be $\tilde{\Gc} = \tilde{\Gc} \Lambda$ for all $\Gc \in \Gbb$, where $\Lambda$ is a CPTP map on $\Lc(\Hc)$.}
				%				\Statex{\textbf{input:} Choose $\Mbb \subset \N$ and a $N_m \in \N$ for each $m \in \Mbb$. Pick states and an observable $\rho, \hat{\rho}, E \in \Lc(\Hc \otimes \Hc)$ for the two-copy implementation or pick states and an observable $\rho_\Hc, \hat{\rho}_\Hc, E_\Hc \in \Lc(\Hc)$ for the single-copy implementation.  }
				%				\Statex{\textbf{output:} An estimate of the unitarity of the noise map $u(\Lambda)$.}
				\Statex{Fix a gate set $\Gbb$, choose a set of sequence lengths $\Mbb$ to use and determine the number of random sequences $N_m$ per sequence length $m \in \Mbb$.}
				\Procedure{URB}{$\Gbb, \Mbb, \{N_m\}$}
				\ForAll{sequence lengths $m \in \Mbb$}
				\Loop{$N_m$ times}
				\State{Sample $m$ random gates $\Gc_{j_1}, ..., \Gc_{j_m}$ independently and uniformly at random from $\Gbb$;}
				\State{Compose the sequence $\Gc_\jbf =  \Gc_{j_m} \cdots \Gc_{j_2}\Gc_{j_1}$;}
				\If{Two-copy implementation}
				\State{Prepare states $\rho \approx \frac{I + S}{d(d+1)}$ and $\hat{\rho} \approx \frac{I - S}{d(d-1)}$, apply $\Gc_\jbf^{\otimes 2}$ to each state and measure $E \approx S$ a large number of times (where $S$ denotes the Swap gate);}
				\State{From this data, estimate the average sequence purity as
					$$q_\jbf^{(2)} = (\Tr[E  \Gc_\jbf^{\otimes 2} (\rho)] - \Tr[E \Gc_\jbf^{\otimes 2} (\hat{\rho})]) =  \Tr[E\Gc_\jbf^{\otimes 2}(\bar{\rho})];$$}
				\EndIf
				\If{Single-copy implementation}
				\ForAll{nonidentity Pauli's $P,Q \neq I$}
				\State{Prepare states $\rho_\Hc^{(P)} \approx \frac{I + P}{d}$ and $\hat{\rho}_\Hc^{(P)} \approx \frac{I - P}{d}$, apply $\Gc_\jbf$ to each state and measure $E_\Hc^{(Q)} \approx Q$ a large number of times; }
				\EndFor
				\State{From this data, estimate the average sequence purity as
					$$q_\jbf^{(1)} = \frac{1}{d^2-1}\sum_{P,Q \neq I} \left(\Tr[E_\Hc^{(Q)}  \Gc_\jbf({\rho}_\Hc^{(P)})] - \Tr[E_\Hc^{(Q)}  \Gc_\jbf(\hat{\rho}_\Hc^{(P)})]\right)^2;$$}
				
				\EndIf
				\EndLoop
				\State{Compute the empirical average over the sampled sequences $\bar{q}_{m} = \frac{1}{N_m} \sum_{\jbf} q_\jbf$;}
				\EndFor
				\State{Fit $\bar{q}_{m} = B u^{m-1}$, where $B$ is a constant absorbing SPAM errors and $u$ is the unitarity of the noise map.}
				\EndProcedure
			\end{algorithmic}
			\vspace*{-10pt}
			\caption{Outline of the modified unitarity randomized benchmarking protocol.}
			\label{Alg:URB_protocol}
		\end{minipage}
	}
\end{algorithm*}

This section gives an overview of the URB protocol of \cite{Wallman2015a} and gives a small modification based on \cite{Helsen2017}. The protocol is described for any gate set $\Gbb$ that is a unitary 2-design \cite{Gross2007}. Note that even though the protocol works for all these gate sets, our result of the confidence analysis is only applicable to the Clifford group. In \autoref{Alg:URB_protocol} we present an outline of the URB protocol, where we distinguish two different implementations (discussed later in this section).

The URB protocol works similar to the standard RB protocol. First one draws a uniformly distributed random sequence of gates (with length $m$) from the gate set $\Gbb$. Denote such a sequence 
\begin{equation}
\Gc_\jbf = \Gc_{j_m} \cdots \Gc_{j_{2}} \Gc_{j_{1}},
\end{equation}
where each $j_s$ denotes the randomly drawn gate from $\Gbb$ at position $s$. The subscript $\jbf$ denotes the multi-index $(j_1, j_2,...,j_m)$ and therefore indexes the entire sequence. Such a randomly sampled sequence $\Gc_\jbf$ is then applied to a state $\rho$ , after which a two-outcome measurement is performed (in this work the operator $E$ denotes the Hermitian observable associated with a two-outcome measurement $\{M, I-M\}$ with outcomes $\pm1$). However, there are two differences here with respect to the RB protocol. First, there is no global inverse applied at the end of each sequence and second, the expectation value of the measurement outcome is squared. So the URB random variable of interest then becomes $q_\jbf = \Tr[E \Gc_\jbf({\rho})]^2$. Throughout this work, we shall call the URB random variable $q_\jbf$ the sequence purity (in standard RB, the random variable of interest is typically referred to as the survival probability).  The rest of the procedure is then similar: estimate the mean of the sequence purity $q_\jbf$ using $N$ random sequences of fixed length, repeat for various sequence lengths and fit to the model  
\begin{equation}
\Ebb[q_\jbf] = Bu^{m-1} + A
\end{equation}
to obtain the unitarity. 

Here we analyze a slightly modified version of the protocol of \cite{Wallman2015a}, based on ideas of \cite{Helsen2017,Granade2015,Knill2008}. Every sequence of randomly sampled gates $\Gc_\jbf$ is applied to two different input states $\rho$ and $\hat{\rho}$, and half of the difference of their expectation values is taken before squaring. By linearity of quantum mechanics, this is equivalent to performing URB with the traceless input operator
\begin{equation}
\label{eq:rhobar}
\bar{\rho} := \half(\rho - \hat{\rho}).
\end{equation}
The factor $\half$ is strictly not necessary but is added for better statistical comparison.
The key idea behind this is that one effectively works with a traceless input operator $\bar{\rho}$. There are two main benefits of this modification. First, it improves the fitting procedure, because the modified fit model for the mean of the sequence purity becomes (see \eqref{eq:fitmodel} in \autoref{subsec:fit_model_variance_expression})
\begin{equation}
\label{eq:fit}
\Ebb[q_\jbf] = B u^{m-1},
\end{equation}
where the constant $B$ only depends on the input operator $\bar{\rho}$ and the measurement observable $E$. 
This is a linear fitting problem in $u$ by taking the logarithm and can therefore be performed more easily. Second, this modification narrows the distribution of the sequence purity $q_\jbf$, improving the confidence in our point estimate $\bar{q}_m = \frac{1}{N} \sum q_\jbf$ of the exact $\Ebb[q_\jbf]$. 
In the next section we discuss the implementation of the protocol in more detail and emphasize that there are two possible methods to estimate $q_\jbf$.

\subsubsection{The two different implementations}
In this section we discuss two different possible implementations of the URB protocol (as briefly discussed in \cite{Wallman2015a}), which are illustrated in \autoref{fig:urb}. The choice of implementation depends on whether the experimenter has access to two identical copies of the system or not. The implementations differ in the way the sequence purity $q_\jbf$ is computed and what the ideal input operator $\bar{\rho}$ and measurement $E$ are. By ideal operators, we mean the operators that maximize the signal strength (the proportionality factor $B$ in the fit model \eqref{eq:fit}) from which the unitarity is estimated. We will then show that the two implementations are closely related.

\begin{figure}
	\centering
	\includegraphics[width=0.95\linewidth]{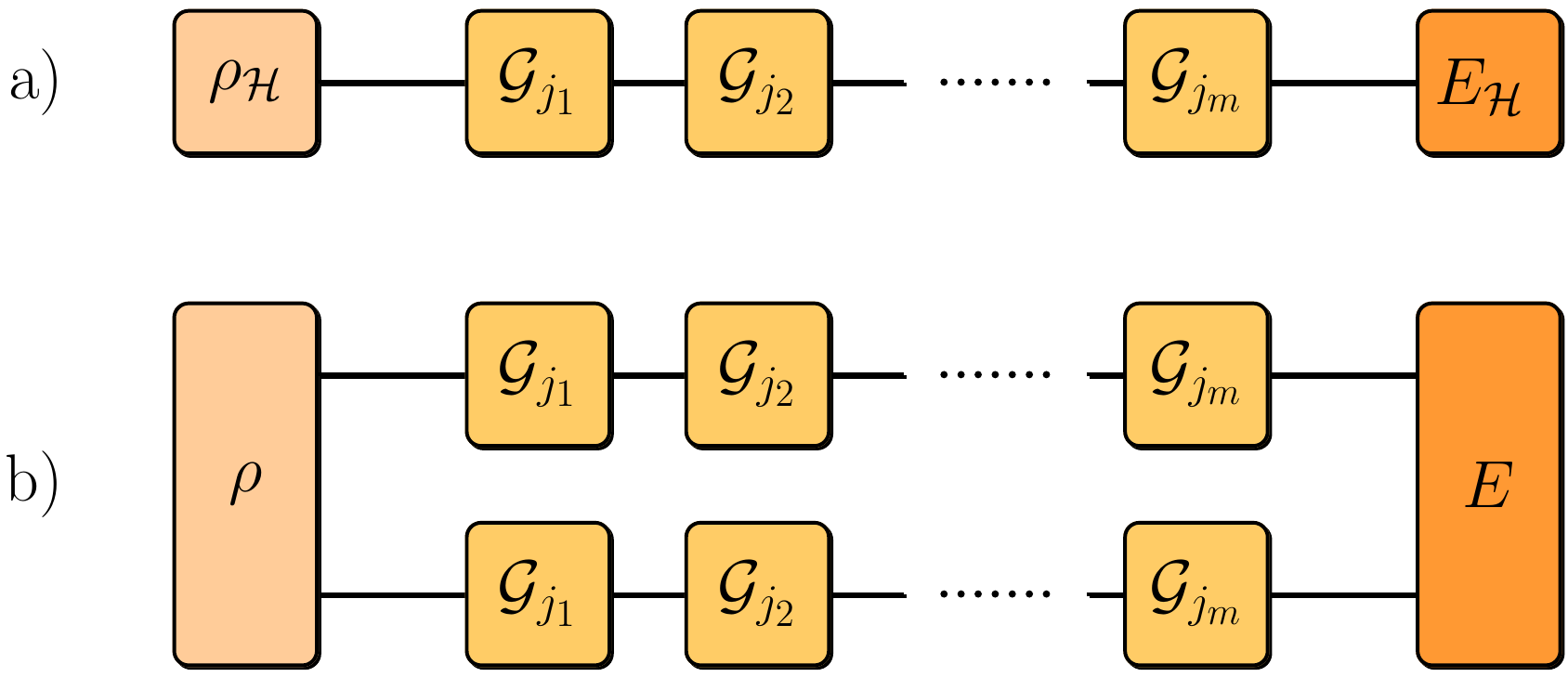}
	\caption{Schematic difference between the single-copy implementation (a) and the two-copy implementation (b) of the unitarity randomized benchmarking protocol. Each line represents a system on the base Hilbert space $\Hc$. In the single-copy implementation, the expected value of the measurement $\Tr[E_\Hc  \Gc_\jbf(\bar{\rho}_\Hc)]$ needs to be squared to obtain $q_\jbf$, whereas in the two-copy implementation $q_\jbf = \Tr[E  \Gc_\jbf^{\otimes 2}(\bar{\rho})]$ yields the direct outcome. 	}
	\label{fig:urb}
\end{figure}

Let us start by discussing the two-copy implementation (\autoref{fig:urb}.b). As the name suggests, this requires two copies of the system $\Hc$ under investigation. The use of two copies follows from the mathematical equivalence
\begin{equation}
\label{eq:7}
q_\jbf = \Tr[E  \Gc_\jbf(\bar{\rho})]^2 = \Tr[E^{\otimes 2} \Gc_\jbf^{\otimes 2}\left(\bar{\rho}^{\otimes 2}\right)].
\end{equation}
If the experimenter has access to two identical copies of the system $\Hc$, the input and measurement operator can be entangled across the two copies of the system. The sequence $\Gc_\jbf$ is then applied to each half of the system $\Hc \otimes \Hc$. This yields the sequence purity of the two-copy implementation as
\begin{equation}
\label{eq:q_jbf^2}
q_\jbf^{(2)} = \Tr[E \Gc_\jbf^{\otimes 2}\left(\bar{\rho}\right)],
\end{equation}
where $\bar{\rho}, E \in \Lc(\Hc \otimes \Hc)$ are now operators on the two copies of the system. 
Since $E$ is a two-valued measurement with outcomes ($\pm1$) and $\bar{\rho}$ is half the difference between two physical states, it is not hard to show that the sequence purity lies in the interval $q_\jbf^{(2)} \in [-1,1]$. In \autoref{subsec:interval_length} we show that this interval can be narrowed under mild assumptions. In the two-copy implementation it is implicitly assumed that the experimenter can operate identically on each subsystem without any cross-talk between the two subsystems.  Moreover, the experimenter should be able to prepare and measure over the two copies of the system.
Experimentally the input and measurement operators $\bar{\rho}, E \in \Lc(\Hc \otimes \Hc)$ should be as close to the ideal operators as possible. The ideal operators are given by (see Appendix~\ref{app_sub:statistical_results_proof} for more details and proof)
\begin{align}
\label{eq:ideal_operators}
\rho_\ideal &= \frac{I + S}{d(d+1)}, & \hat{\rho}_\ideal &= \frac{I - S}{d(d-1)}, & E_\ideal  &= S,
\end{align}
where $I$ is the identity and $S$ is the Swap operator on $\Hc\otimes \Hc$, and $d$ is the dimension of $\Hc$. The state $\rho_\ideal$  ($\hat{\rho}_\ideal$) is the maximally mixed state on the symmetric (anti-symmetric) subspace of $\Hc \otimes \Hc$. Note that the maximally mixed state on a subspace can be prepared by uniformly sampling pure states from an orthonormal basis of this subspace. The operator $E_\ideal$ is the Hermitian observable associated with a two-valued measurement that discriminates between symmetric (outcome $1$) and anti-symmetric states (outcome $-1$). 

In the single-copy implementation, the experimenter must obtain an estimate of the sequence purity $q_\jbf$ using only a single copy of the system $\Hc$. From \eqref{eq:7}, it can be seen that $q_\jbf = \Tr[E_\Hc \Gc_\jbf(\bar{\rho}_\Hc)]^2$ is the sequence purity given the operators $\bar{\rho}_\Hc, E_\Hc \in \Lc(\Hc)$. Here the subscript $\Hc$ is to emphasize that the operators are on a single copy of $\Hc$. Throughout this paper we will just write $\bar{\rho}$ and $E$ for operators on $\Hc \otimes \Hc$ and indicate operators on a single copy explicitly by adding a subscript $\Hc$.  
There are two disadvantages in defining the single-copy sequence purity using one pair of input and measurement operators $\bar{\rho}_\Hc, E_\Hc \in \Lc(\Hc)$. First, the proportionality factor $B$ in \eqref{eq:fit} is upper bounded by $\frac{1}{d^2-1}$, where $d$ is the dimension of $\Hc$ \cite{Wallman2015a}. This means that the signal strength decreases exponentially with the system size. Second, the variance of the sequence purity is large. This leads to large uncertainty in the estimated average sequence purity $\bar{q}_m$. These disadvantages can be resolved by using multiple different pairs of input and measurement operators \cite{Wallman2015a}. The ideal set of operators is chosen in such a way that summing the expectation values squared for each pair of operators leads to effectively simulating the ideal operators of \eqref{eq:ideal_operators}. Let us make this more precise. Define the single-copy sequence purity as
\begin{equation}
\label{eq:q_jbf^1}
q_\jbf^{(1)} = \frac{1}{d^2-1}\sum_{P,Q \neq I} \Tr[E_\Hc^{(Q)}  \Gc_\jbf(\bar{\rho}_\Hc^{(P)})]^2,
\end{equation}
where the sum is over all nonidentity multiqubit Pauli operators $P,Q$. Each $\bar{\rho}_\Hc^{(P)}$ and $E_\Hc^{(Q)}$ are different input and measurement operator settings indexed by the nonidentity Pauli operators $P$ and $Q$ respectively. For each pair $P, Q$, the expectation value $\Tr[E_\Hc^{(Q)}  \Gc_\jbf(\bar{\rho}_\Hc^{(P)})]$ is to be estimated experimentally. This expectation can be shown to lie in the interval $[-1,1]$ by definition of $E$ and $\bar{\rho}$, so that the expectation value squared lies in the unit interval. Therefore the single-copy sequence purity can in principle lie anywhere in the interval $q_\jbf^{(1)} \in [0, d^2-1]$, since each summand lies in the unit interval and the summation runs over $(d^2-1)^2$ terms. However in \autoref{subsec:interval_length} we show that this interval can be narrowed significantly under mild assumptions. Since the sum runs twice over all nonidentity Pauli operators, estimating the sequence purity $q_\jbf^{(1)}$ requires  $(d^2-1)^2$ different settings. This is a number that grows exponentially in the number of qubits comprising the system. We also emphasize that simply squaring and summing up estimates of $\Tr[E_\Hc^{(Q)}  \Gc_\jbf(\bar{\rho}_\Hc^{(P)})]$ to obtain an estimate of $q_\jbf^{(1)}$ yields a positively biased estimator for $q_\jbf^{(1)}$. This may lead to overestimating the unitarity. See \autoref{subsubsec:unbiased_estimator} for more details on how to correctly estimate $q_\jbf^{(1)}$. The states ${\rho}_\Hc^{(P)}, \hat{\rho}_\Hc^{(P)}$ and measurement $E_\Hc^{(Q)}$ should be implemented as closely as possible to the ideal operators
\begin{align}
\label{eq:ideal_single_copy}
\rho_{\Hc,\ideal}^{(P)} &= \frac{I + P}{d}, & \hat{\rho}_{\Hc,\ideal}^{(P)} &= \frac{I - P}{d}, & E_{\Hc,\ideal}^{(Q)} &= Q.
\end{align}
The ideal state $\rho_{\Hc,\ideal}^{(P)}$ ($\hat{\rho}_{\Hc,\ideal}^{(P)}$) is the maximally mixed state on the positive (negative) eigenspace of the Pauli operator $P$, and the measurement $E_{\Hc,\ideal}^{(Q)}$ is the two-valued measurement that discriminates between the positive (outcome $1$) and negative (outcome $-1$) eigenspace of the Pauli operator $Q$. 

Next we show that the single-copy can be interpreted as a special case of the two-copy implementation (this is not surprising in view of \eqref{eq:7}). To do so, we show that in the single-copy implementation, one effectively works with two-copy operators of the form
\begin{equation}
\label{eq:state_12}
\begin{aligned}
\bar{\rho}_\mathrm{eff} &= \frac{d}{d^2 - 1} \sum_{P \neq I} \bar{\rho}_\Hc^{(P)} \otimes \bar{\rho}_\Hc^{(P)}, \\
\bar{E}_\mathrm{eff} &= \frac{1}{d} \sum_{Q\neq I} \bar{E}_\Hc^{(Q)} \otimes \bar{E}_\Hc^{(Q)}.
\end{aligned}
\end{equation}
Here $\bar{E}$ ($\bar{E}_\Hc^{(Q)}$) is the traceless part of the observable $E$ (${E}_\Hc^{(Q)}$) , defined as
\begin{equation}
\label{eq:traceless_measurement}
\bar{E} := E - \Tr[E] \frac{I}{d^2}, \qquad \bar{E} := E_\Hc - \Tr[E_\Hc] \frac{I_\Hc}{d}.
\end{equation}
The key point is that replacing the observable $E$ with $\bar{E}$ makes no difference, since $\Tr[E \Gc_\jbf^\tp{2}(\bar{\rho})] = \Tr[\bar{E} \Gc_\jbf^\tp{2}(\bar{\rho})]$. This follows directly from \eqref{eq:traceless_measurement}, since $\Tr[I \Gc_\jbf^\tp{2}(\bar{\rho})] = 0$ by the tracelessness of $\bar{\rho}$ and the trace-preserving property of $\Gc_\jbf^\tp{2}$. Analogously, in the single-copy implementation, the traceless measurement $\bar{E}_\Hc^{(Q)}$ can be used instead of the observable ${E}_\Hc^{(Q)}$. Throughout the paper, a bar over the measurement operator will mean the traceless component as defined by \eqref{eq:traceless_measurement}.

The key idea of \eqref{eq:state_12} is that $\bar{\rho}_\mathrm{eff}$ and $\bar{E}_\mathrm{eff}$ are constructed such  that computing $q_\jbf^{(1)}$ with \eqref{eq:q_jbf^1} is mathematically equivalent to computing $q_\jbf^{(2)}$ with \eqref{eq:q_jbf^2} using the effective operators \eqref{eq:state_12},
\begin{multline}
q_\jbf^{(1)} = \frac{1}{d^2-1}\sum_{P,Q \neq I} \Tr[\bar{E}_\Hc^{(Q)}  \Gc_\jbf(\bar{\rho}_\Hc^{(P)})]^2 \\
=\Tr\left[ \bar{E}_\mathrm{eff}   \Gc_\jbf^\tp{2}(\bar{\rho}_\mathrm{eff})\right] 
=  q_\jbf^{(2)}.
\end{multline}
In particular the ideal effective operators $\bar{\rho}_{\mathrm{eff},\ideal}$ and $\bar{E}_{\mathrm{eff},\ideal}$ (defined by \eqref{eq:state_12} for the ideal single-copy  operators \eqref{eq:ideal_single_copy}) are equal to the ideal two-copy operators \eqref{eq:ideal_operators},
\begin{equation}
\bar{\rho}_{\mathrm{eff},\ideal} = \bar{\rho}_{\ideal} \qquad \mbox{and} \qquad \bar{E}_{\mathrm{eff},\ideal} = \bar{E}_{\ideal}.
\end{equation}
This follows from the fact that \cite{Wallman2015a}
\begin{equation}
S = \frac{1}{d} \sum_P P \otimes P.
\end{equation}
Note that the sum is here over all Pauli matrices including the identity. 
As a result of this, the rest of the paper will exclusively deal with the two-copy operators $\bar{\rho}$,  $E \in \Lc(\Hc \otimes \Hc)$. The results can be interpreted for the single-copy protocol by considering the effective operators \eqref{eq:state_12}.

The two-copy implementation of the protocol as previously discussed, can only be implemented if the experimenter has access to two different, but identical copies of the system under examination. These two systems must be simultaneously accessible for entangled state preparation and measurements, but the unitary control on each subsystem needs to be fully disjoint (i.e., without crosstalk) and identical (meaning noise must be identical on each subsystem). These assumptions are hard if not impossible to fulfill in any experimental system. We emphasize however that the two-copy implementation is introduced as a mathematical tool for the analysis of the URB protocol and its equivalence to the more realistic single-copy protocol was shown. 

This concludes our review of the URB protocol, including the proposed modification of traceless input operators and emphasizing the two different implementations (which we have named the single- and two-copy implementation, respectively). Next, we will present our main result. We will show how a concentration inequality can be used to relate the required resources (the number of sequences $N$) to parameters that quantify the confidence in the estimate of the average sequence purity $\bar{q}_m$. To do so, we will present a sharp bound $\sigma^2$ on the variance of the sequence purity $\Vbb[q_\jbf^{(K)}]$ and present a bound $L$ on the length of the interval in which the sequence purity $q_\jbf^{(K)}$ lies. These bounds are independent of $K$ (the choice between single or two-copy implementation). Therefore, if no implementation-specific details are discussed, the sequence purity is just denoted $q_\jbf$.

%----------------------------------------------------------------------------------------
%	SECTION 2: Results
%----------------------------------------------------------------------------------------
\section{Summary of results}
\label{sec:results}
In this section the main contribution of the paper is summarized. The main result is a sharp bound on the number of sequences $N$ required to obtain the average sequence purity $\bar{q}_m$ given fixed sequence length $m$ with a certain \textit{a priori} determined confidence. In \autoref{subsec:hoeffding} we review a result from statistics to quantify the relation between the number of sequences $N$ and the confidence. This relation requires some knowledge on the distribution of the sequence purity $q_\jbf$. A bound on the variance and a bound on the interval length of the sequence purity are needed. In \autoref{subsec:var_bound} we present a bound on the variance of the URB sequence purity ${q}_\jbf$ for benchmarking the Clifford gate set. This is the main contribution of this work. In \autoref{subsec:interval_length} we present a bound on the length of the interval in which $q_\jbf$ must lie. Finally in \autoref{subsec:examples} we give some examples on how to use our results.

\subsection{Relation between the confidence parameters and the number of sequences}
\label{subsec:hoeffding}
Using concentration inequalities from statistics, the confidence in the estimate $\bar{q}_m$ can be expressed as the probability that it deviates at most $\epsilon$ from the exact mean $\Ebb[q_\jbf]$. If this probability $\Pbb[|\bar{q}_m - \Ebb[q_\jbf] | \geq \epsilon] \leq \delta$ is to be bounded by $\delta$, then the number of required data points $N$ is related to the confidence parameters $\epsilon, \delta$ by \cite{Hoeffding1963}
\begin{equation}
\label{eq:hoeffding2}
 2 \left(\left( \frac{L}{L -\epsilon}\right)^{\frac{L^2 -\epsilon L}{{\sigma}^2 + L^2}} \left(\frac{{\sigma}^2}{{\sigma}^2 + \epsilon L}\right)^{\frac{{\sigma}^2+\epsilon L}{{\sigma}^2+L^2}}\right)^N \leq \delta.
\end{equation}
In this expression $\sigma^2$ is a bound on the variance $\Vbb[q_\jbf]$ and $L$ is a bound on the length of the interval in which $q_\jbf$ lies. Given $\sigma^2$ and $L$, there are two ways to apply this inequality. It can either be solved (numerically) for $\epsilon$, given fixed $N$ and $\delta$, or it can be solved for $N$ given $\epsilon, \delta$. In any case, it provides a direct relation between the number of required sequences $N$ and the confidence parameters $\epsilon, \delta$, given $L$ and $\sigma^2$. So in order to apply \eqref{eq:hoeffding2}, the bounds $L$ and $\sigma^2$ are needed.

In the next section we will present a sharp bound $\sigma^2$ on the variance of the sequence purity $\Vbb[q_\jbf]$. This bound is the key ingredient in using \eqref{eq:hoeffding2} and it is the main contribution of this paper.
 
\subsection{Bound on the variance of the sequence purity}
\label{subsec:var_bound}
In this section we present a bound $\sigma^2$ on the variance of the sequence purity $\Vbb[q_\jbf]$ that is valid under the following assumptions:
\begin{enumerate}[leftmargin=1.5\parindent,rightmargin=0.5\parindent,itemsep=2pt]
	\item The gate set under investigation is the $d$-dimensional Clifford group, denoted $\Cliff{d}$. Here $d = 2^q$ for a $q$-qubit system. This assumption is necessary for deriving a variance bound. Even though the expected value $\Ebb[q_\jbf]$ of the URB sequence purity is independent of the chosen gate set (as long as it is a unitary 2-design), the variance is not. The Clifford group was chosen as the default gate set.
	
	\item Gate errors are independent of the gate. This is known as the gate-independent error model. In this model, the implemented noisy gate is $\tilde{\Gc} = \Gc \Lambda$, where $\Gc \in \Cliff{d}$ is the ideal Clifford gate and $\Lambda$ is an arbitrary quantum channel describing the noise. Crucially, $\Lambda$ does not depend on the specific gate $\Gc \in \Cliff{d}$. This is assumption is necessary for deriving the fit model for URB \cite{Wallman2015a}. Consequently our variance bound also employs this assumption. The URB protocol has not been analyzed in a gate dependent noise setting.
	
	\item The noise map $\Lambda$ is assumed to be unital if $q \geq 2$ (or equivalently if $d \geq 4$). A quantum channel $\Lambda$ is unital if the maximally mixed state is a fixed point of the map, $\Lambda(I) = I$. If the system under investigation $\Hc$ is a single-qubit system ($d = 2$), than this assumption is not necessary. Our result thus holds for any single-qubit quantum channel $\Lambda$. This assumption enters in our derivation of the variance bound. It is not a fundamental assumption but rather a condition under which we were able to derive a useful, sharp bound.
\end{enumerate}
At this point, we emphasize that $\Vbb[q_\jbf]$ is the between-sequence variance, i.e., the variance of $q_\jbf$ due to the randomly sampled sequence indexed by $\jbf$. In particular this means that given a sequence $\jbf$, we assume that $q_\jbf$ can be determined with arbitrary precision. In reality $q_\jbf$ can only be estimated due to the probabilistic nature of quantum mechanics by taking many single-shot measurements of the same sequence $\jbf$. In \autoref{subsec:estimation} we relax this assumption by splitting the total variance into the sum of the between-sequence variance (the variance due to randomly sampled $\jbf$) and the within-sequence variance (the variance due to uncertainty in $q_\jbf$ for fixed $\jbf$). 

Under the assumptions stated above, the following bound on the variance $\Vbb[q_\jbf]$ is derived (see \autoref{thm:var_bound} in Appendix~\ref{app:variance_bound})

\begin{equation}
\label{eq:var_bound}
\begin{split}
&\Vbb[q_\jbf^{(K)}] \leq  \sigma^2 \\ &\quad=  \frac{1-u^{2(m-1)}}{1-u^2} (1 - u)^2 \Big[ c_1(d) 
 +  c_2(d) \| \bar{E}_\spam \|_\infty^2 \\ &\qquad+ c_3(d) \| \bar{\rho}_\spam \|_1^2 \Big]    +    \| \bar{\rho}_\spam \|_1^2 \| \bar{E}_\spam \|_\infty^2 ,   
\end{split}
\end{equation}
which is independent of the used implementation (single or two-copy, corresponding to $K = 1,2$). Here $u$ is the unitarity of $\Lambda$, $m$ is the sequence length,  $\| \bar{E}_\spam \|_\infty^2$, $\| \bar{\rho}_\spam \|_1^2$ are quantities depending on the quality of state preparation and measurement and $c_i$ are constants that solely depend on the dimension $d$. The values of $c_i$ for small $d$ are tabulated in \autoref{tab:const_dims}. For precise definitions of these quantities, see \autoref{thm:var_bound} in Appendix~\ref{app:variance_bound}. The error operators have the following definitions:
\begin{equation}
\label{eq:alpha}
\begin{aligned}
\bar{\rho}_\spam &= \bar{\rho} - \frac{\Tr[\bar{\rho}_\ideal  \bar{\rho}]}{\|\bar{\rho}_\ideal\|_2^2} \bar{\rho}_\ideal = \bar{\rho} - (d^2-1)\Tr[\bar{\rho}_\ideal  \bar{\rho}] \bar{\rho}_\ideal, \\
\bar{E}_\spam &= \bar{E} - \frac{\Tr[\bar{E}_\ideal  \bar{E} ]}{\| \bar{E}_\ideal \|_2^2} \bar{E}_\ideal = \bar{E} - \frac{\Tr[\bar{E}_\ideal  \bar{E} ]}{d^2-1} \bar{E}_\ideal,
\end{aligned}
\end{equation}
where the ideal operators $\bar{\rho}_\ideal, \bar{E}_\ideal$ are defined in \eqref{eq:ideal_operators} and a bar over the measurement operator indicates its traceless component $\bar{E} = E - \frac{\Tr[E]}{d^2} I$ (as defined in \eqref{eq:traceless_measurement}). Recall that $\bar{\rho}$ was defined as the difference between two states (\eqref{eq:rhobar}). The error operators are defined in such a way that they are orthogonal to the ideal operators with respect to the Hilbert-Schmidt inner product
\begin{equation}
\Tr[\bar{\rho}_\spam \bar{\rho}_\ideal] = \Tr[\bar{E}_\spam \bar{E}_\ideal] = 0.
\end{equation}
The norms on the error operators are the trace norm and operator norm respectively, defined for all $A \in \Lc(\Hc\otimes \Hc)$ as 
\begin{equation}
\begin{split}
\| A \|_1 &=      \Tr[ \sqrt{A^\dagger A}] = \sum_i s_i(A),\\
\| A \|_\infty &= \sup_{0 \neq x \in \Hc^\tp{2}}  \frac{\|Ax\|_2}{\|x\|_2} = \max_i\{ s_i(A)\},
\end{split}
\end{equation}
with $s_i(A)$ the $i$-th singular value of $A$ and $\| x \|_2$ the euclidean norm on $\Hc^\tp{2}$. Note that in the single-copy case the quantities $\| \bar{\rho}_\spam \|_1^2, \| \bar{E}_\spam \|_\infty^2$ as defined in \eqref{eq:alpha} are to be estimated  using $\bar{\rho}_\mathrm{eff}$ and $\bar{E}_\mathrm{eff}$ as defined in \eqref{eq:state_12}.

\begin{table}
	\caption{Evaluation of the constants $c_i(d)$ for various small-dimensional systems. The last row indicates the asymptotic behavior.}
	\ra{1.5}
	\centering
	\begin{tabularx}{\linewidth}{ C{1} C{1} C{1} C{1}}
		\toprule
		$d$  & $c_1(d)$ & $c_2(d)$ & $c_3(d)$ \\ 
		\midrule
		$2$ & $\frac{11}{12}$   & $\frac{13}{9}$  & $\frac{5}{2}$  \\ 

		$4$  & $\frac{179}{60}$ & $54.675$ & $48.053$  \\ 

		$8$  & $1.6322$ & $81.445$ & $119.31$  \\ 

		$16$  & $1.1443$ & $110.64$ & $296.88$  \\ 
	
		$32$  & $1.0354$  & $173.80$  & $891.69$  \\ 

		$\rightarrow \infty$ & $\Oc(1)$  & $\Oc(d)$  & $\Oc(d^2)$  \\ 
		\bottomrule 
	\end{tabularx}
	\label{tab:const_dims}
\end{table}

The variance bound of \eqref{eq:var_bound} has some appealing qualitative features. The first feature is that the first term is proportional to $(1-u)^2$. This means that the first term goes to zero quadratically as the unitarity $u$ of the error map $\Lambda$ approaches 1. The fact that the second term is constant with respect to both $u$ and $m$ is unavoidable, as will be discussed in \autoref{subsec:spam_dependence}. The second appealing feature is the fact that the bound is asymptotically independent of the sequence length $m$. Thus the variance bound is useful in any regime of $m$. In \autoref{sec:discussion} the dependence of the variance bound and the resulting number of sequences on various parameters is discussed in greater detail. 

In the next section we present a bound $L$ in the length of the interval in which the sequence purity $q_\jbf$ lies. This is the final ingredient needed in order to apply \eqref{eq:hoeffding2}. 

\subsection{Bound on the interval of the sequence purity}
\label{subsec:interval_length}
In this section we present the improved bound $L$ on the length of the interval in which the sequence purity $q_\jbf^{(K)}$ lies. Even though the actual interval depends on $K$, the length of these intervals is the same. Thus the bound $L$ on the interval length of the sequence purity is independent of the implementation indexed by $K$. The improved bound is derived under the mild assumption that
the experimental control is sufficiently good such that $\Tr[\bar{\rho}_\ideal  \bar{\rho}] \geq 0$ and $\Tr[\bar{E}_\ideal  \bar{E}] \geq 0$ (analogous assumption holds for the single-copy input and measurement operators). These conditions are satisfied only if the conditions
\begin{align}
\label{eq:26}
\Tr[ \rho_\ideal \rho] &\geq \Tr[ \hat{\rho}_\ideal \rho], & \Tr[ \hat{\rho}_\ideal \hat{\rho}] \geq \Tr[ \rho_\ideal \hat{\rho}], \\
\label{eq:27}
\Tr[E \bar{\rho}_\ideal] &\geq 0
\end{align}
are satisfied. \eqref{eq:26} can be interpreted as requiring that the implemented states $\rho$, $\hat{\rho}$ have more overlap with their corresponding ideal state than with the noncorresponding ideal states. \eqref{eq:27} is equivalent to $\Tr[\bar{E} \bar{E}_\ideal] \geq 0$ since $\bar{E}_\ideal = (d^2-1) \bar{\rho}_\ideal$ and $\Tr[\bar{\rho}_\ideal \bar{E}] = \Tr[E \bar{\rho}_\ideal]$. \eqref{eq:27} has the interpretation that the measurement $\{M, I-M\}$ associated with the observable $E = 2M - I$ assigns the correct outcome ($+1$ for $\rho_\ideal$ and $-1$ for $\hat{\rho}_\ideal$) with at least probability $\half$, or alternatively, that the measurement can correctly discriminate the maximally mixed state on the symmetric subspace ($\rho_\ideal$) from the maximally mixed state on the anti-symmetric subspace ($\hat{\rho}_\ideal$). These are very reasonable assumptions for any practical quantum information device.

In \autoref{lem:interval_length} of Appendix~\ref{app_sub:statistical_results_proof} we show that under the stated assumption, the sequence purity lies in the interval
\begin{align}
q_\jbf^{(1)} &\in [0, 1 +  \| \bar{\rho}_\spam \|_1 +  \| \bar{E}_\spam \|_\infty + \| \bar{\rho}_\spam \|_1\| \bar{E}_\spam \|_\infty], \\
q_\jbf^{(2)} &\in [- \| \bar{\rho}_\spam \|_1 - \| \bar{E}_\spam \|_\infty - \| \bar{\rho}_\spam \|_1\| \bar{E}_\spam \|_\infty, 1].
\end{align}
Therefore it follows that
\begin{equation}
\label{eq:interval_length}
L = 1 + \| \bar{\rho}_\spam \|_1 +  \| \bar{E}_\spam \|_\infty + \| \bar{\rho}_\spam \|_1\| \bar{E}_\spam \|_\infty
\end{equation}
for both implementations. The idea of the proof of \autoref{lem:interval_length} is to decompose the input and measurement operators $\bar{\rho}$ and $\bar{E}$ into their ideal and error components according to \eqref{eq:alpha}. This gives rise to four terms. The ideal term $\Tr[E_\ideal \Gc_\jbf^\tp{2} (\bar{\rho}_{\ideal})]$ can be bounded in the interval $[0,1]$. The other terms are then bounded in magnitude using H\"older's inequality, which contributes the last three terms in \eqref{eq:interval_length}. 

\subsection{Examples}
\label{subsec:examples}
Perhaps the best way to gain insight in the use of \eqref{eq:hoeffding2}, \eqref{eq:var_bound} and \eqref{eq:interval_length} is by example. In \autoref{ex:2} we calculate the required number of sequences for a fixed choice of all relevant parameters. In \autoref{ex:3} we simulate a URB experiment using fixed number of sequences and compute the confidence interval around each estimate $\bar{q}_m$. We compare the results of these examples with a previously known bound (first used in \cite{Magesan2012a}). This bound does not use the variance, but just uses the boundedness of the sequence purity $q_\jbf$.  It claims that $\Pbb[| \bar{q}_m - \Ebb[q_\jbf]| \geq \epsilon] \leq \delta$, whenever \cite{Hoeffding1963}
\begin{equation}
\label{eq:Hoeffding1}
2 e^{-2N\frac{\epsilon^2}{L^2}} \leq \delta.
\end{equation}
The number of sequences $N$ is merely a function of the confidence parameters $\epsilon$,  $\delta$ and the interval length $L$. In particular it does not depend on the variance of $q_\jbf$.

\begin{example}
	\label{ex:2}
	Suppose that a URB experiment is performed on the single-qubit Clifford group ($d=2$). The choice of implementation (single-copy or two-copy) is irrelevant for this example since both the variance bound \eqref{eq:var_bound} and the interval length bound \eqref{eq:interval_length} are independent of the choice of implementation. The only difference in practice is how to estimate the SPAM parameters $\| \bar{\rho}_\spam \|_1^2, \| E_\spam \|_\infty^2$. Furthermore suppose that an priori estimate of the unitarity is $u = 0.98$ and an estimate for the SPAM parameters is $\| \bar{\rho}_\spam \|_1^2 = \| E_\spam \|_\infty^2 = 0.02$. Then, after choosing appropriate sequence lengths to use in the experiment, an upper bound on the variance as a function of the sequence length can be computed using \eqref{eq:var_bound}. The interval length can be bounded using \eqref{eq:interval_length}. Using $\| \bar{\rho}_\spam \|_1^2 = \| E_\spam \|_\infty^2 = 0.02$, this yields $L = 1.02 + 0.2\sqrt{2} \approx 1.303$. 
	Finally, choosing an interval $\epsilon$ and confidence $\delta$, \eqref{eq:hoeffding2} gives the required number of sequences $N$ (at fixed length $m$). Concretely, setting $\epsilon = 0.02$, $\delta = 0.01$ and all other parameters as discussed, the number of sequences required for sequences of length $m = 10$, is $N = 242$. For sequence length $m = 30$, the required number is $N = 366$, whereas $m = 100$ requires $N = 452$. The long sequence length limit (when $u^{2(m-1)} \ll 1$), yields $N = 457$. 
	
	Let us compare these numbers with the previously known bound \eqref{eq:Hoeffding1} that does not use the variance of $q_\jbf$. Given our choices of $\epsilon = 0.02$, $\delta = 0.01$ and $\| \bar{\rho}_\spam \|_1^2 = \| E_\spam \|_\infty^2 = 0.02$  (from which $L = 1.02 + 0.2\sqrt{2} \approx 1.303$ is computed using \eqref{eq:interval_length}), the bound \eqref{eq:Hoeffding1} yields $N = 11242$ required sequences. We emphasize that this number is independent of $u$ or $m$. In this scenario, our bound gives approximately two orders of magnitude improvement.
\end{example}

\begin{figure}
	\centering
	\includegraphics[width=\linewidth]{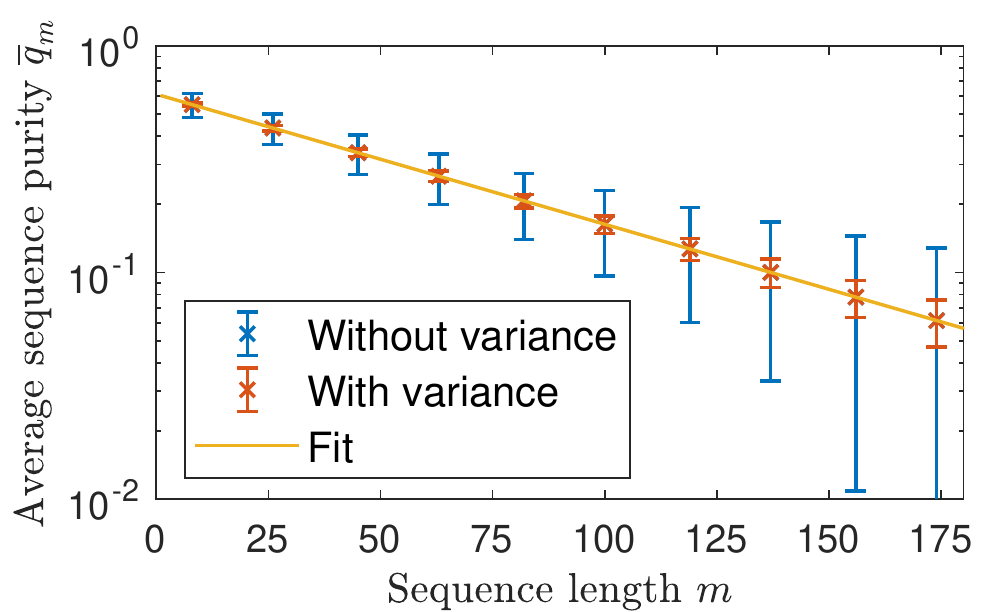}
	\caption{Comparison of the $99\%$ confidence intervals around the average sequence purity $\bar{q}_m$ calculated with and without our variance bound at several different sequence lengths. The plot is based on a simulated URB experiment of the single-qubit Clifford group with $N = 250$ samples per sequence length $m$. The empirical average sequence purity $\bar{q}_m$ (marked with a cross) is plotted versus the sequence length $m$ on a semilogarithmic scale. The larger (blue) bars indicate the $99\%$ confidence interval without variance (\eqref{eq:Hoeffding1}) and the smaller (red) bars indicate the $99\%$ confidence interval of \eqref{eq:hoeffding2} based on our sharp variance bound \eqref{eq:var_bound}. Here we used an \textit{a priori} estimates of the unitarity and SPAM parameters of $u = 0.98$ and $\| \bar{\rho}_\spam \|_1^2 = \| E_\spam \|_\infty^2 = 0.02$ respectively. Then \eqref{eq:interval_length} yields $L = 1.02 + 0.2\sqrt{2}$.
	For completeness, a least-squares fit according to the model $\bar{q}_m = B u^{m-1}$ (see \eqref{eq:fit}) is shown in the yellow solid line. This yields $u \approx 0.987$. }	  
	\label{fig:simulation}
\end{figure}

\begin{example}
	\label{ex:3}
	In \autoref{fig:simulation} we compare the $99\%$ confidence intervals $\epsilon$ (for fixed $N = 250$ and $\delta = 0.01$) around the empirical average sequence purity $\bar{q}_m$ calculated with and without our variance bound at several different sequence lengths. The empirical average sequence purity $\bar{q}_m$ data is based on a simulated single-qubit Clifford URB experiment.  
	The length of the confidence interval $\epsilon$ without variance (larger blue bars) is computed from \eqref{eq:Hoeffding1}. Then the choice of $N = 250$ and $\delta = 0.01$ yields $\epsilon = 0.134$. On the other hand, the length of the confidence interval $\epsilon$ with variance (smaller red bars in the plot) is computed from \eqref{eq:hoeffding2} by solving the equation for $\epsilon$, using our sharp variance bound \eqref{eq:var_bound}. In the evaluation of \eqref{eq:var_bound}, the \textit{a priori} estimates $u = 0.98$ and $\| \bar{\rho}_\spam \|_1^2 = \| E_\spam \|_\infty^2 = 0.02$ were used. Then \eqref{eq:interval_length} yields $L = 1.02 + 0.2\sqrt{2}$. Using our sharp variance bound, the values of the confidence interval vary between $\epsilon = 0.019$ (for $m = 8$) and $\epsilon = 0.029$ (for $m = 174$). This is approximately an order of magnitude larger than the confidence interval without variance $\epsilon = 0.134$.
	
	In this simulated experiment the Clifford gates are implemented with a fixed error channel $\Lambda$ that is generated by taking a convex combination of the identity channel (with high weight) and a random CPTP map (sampled using QETLAB \cite{QETLAB}). Similarly, the noisy input states and measurement operator are simulated by taking a convex combination of the ideal operators and randomly generated operators (generated using QETLAB). 
	For this particular realization of an error map $\Lambda$, the data points seem to be even more accurate than our confidence interval might suggest based on their proximity to the fit. This is due to the fact that this particular error channel is well-behaved. We emphasize that our bound is valid for any unital or single-qubit error map. In particular this means that our bound is valid for the worst case realizations of $\Lambda$. It is unclear what error map $\Lambda$ maximizes the variance of the sequence purity.
	
	We emphasize that the point of this simulated example is not to prescribe a direct method for extracting the confidence in the unitarity, as this generally depends on the fitting model and the way the uncertainty in the average sequence purity are propagated into the uncertainty of the unitarity. Moreover, more advanced statistical tools may be used to extract the unitarity from the obtained (in this case simulated) data, like \cite{Epstein2014,Hincks2018}. The goal of this example is to illustrate the significant gain in confidence of the average sequence purity when the simple concentration inequalities of Hoeffding are applied \cite{Hoeffding1963}. 
	The point is that the additional knowledge of a variance bound on the underlying distribution of the sequence purity $q_\jbf$ can be used by statistical tools to extract the unitarity with improved confidence. 
\end{example}

In the next section we explore the behavior of our bound in various parameter regimes.

\begin{figure}
	\centering
	\includegraphics[width=\linewidth]{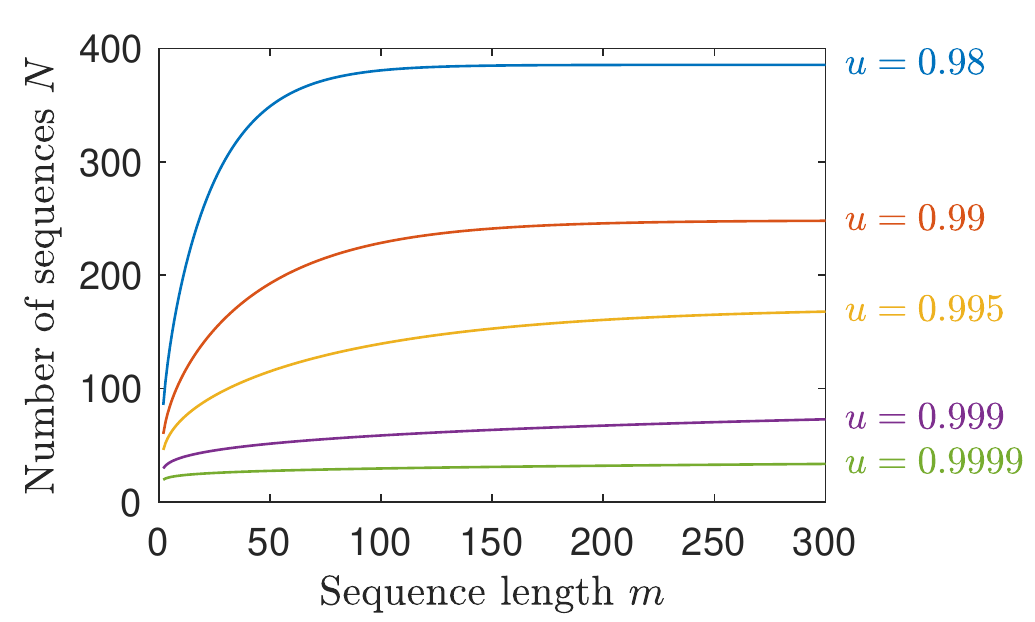}
	\caption{Number of sequences $N$ versus the sequence length $m$ for various values of the unitarity $u$ when benchmarking the single-qubit Clifford group ($d=2$). Confidence parameters are $\epsilon = 0.02$ and $\delta = 0.01$. The SPAM parameters are  $\| \bar{\rho}_\spam \|_1^2 = \| E_\spam \|_\infty^2 = 0$. By \eqref{eq:interval_length} then $L = 1$ is used. The number of sequences is asymptotically independent of the sequence length. This is consistent with our variance bound \eqref{eq:var_bound}.  }
	\label{fig:N_vs_m}
\end{figure}

%----------------------------------------------------------------------------------------
%	SECTION 3: Discussion
%----------------------------------------------------------------------------------------
\section{Discussion}
\label{sec:discussion}
This section is devoted to discussing the variance bound and the interval length of the sequence purity in more detail. In particular we discuss the variance bound in several different parameter regimes in more detail and aim to provide a better understanding of the parameters that ultimately determine the statistical confidence of the measurements. In \autoref{subsec:dependence_on_u_m} we discuss the dependence of the variance bound \eqref{eq:var_bound} on the unitarity $u$ and the sequence length $m$. In \autoref{subsec:spam_dependence} we discuss the dependence on the SPAM parameters $\| \bar{\rho}_\spam \|_1^2 $ and $\| \bar{E}_\infty \|_1^2 $. Here we also show by example that the variance of the sequence purity does not go to zero in the presence of SPAM errors. In \autoref{subsec:dim_dependence} the dependence of the variance bound on the system size is discussed.

\subsection{Dependence on unitarity and sequence length}
\label{subsec:dependence_on_u_m}

First, we discuss the dependence of the number of required sequences $N$ on the sequence length $m$. In \autoref{fig:N_vs_m} this dependence is plotted for various values of $u$ in the absence of SPAM errors (that is, $\| \bar{\rho}_\spam \|_1^2 = \| E_\spam \|_\infty^2 = 0$). The confidence parameters were fixed at $\delta = 0.01$ and $\epsilon = 0.02$. It can be seen from the figure that $N$ approaches a constant as $m$ increases. This is consistent with our variance bound \eqref{eq:var_bound}, where the factor depending on $m$ is 
\begin{equation}
\frac{1 - u^{2(m-1)}}{1 - u^2} (1-u)^2.
\end{equation}
This approaches a constant in the limit of large sequence lengths. This limit is approximately achieved when $u^{2(m-1)} \ll 1$. The exact limit is given by
\begin{equation}
\lim\limits_{m \rightarrow \infty} \frac{1 - u^{2(m-1)}}{1 - u^2} (1-u)^2 = \frac{1-u}{1+u}.
\end{equation}
In the presence of SPAM errors, the asymptotic constant is larger than in its absence, but the behavior is similar. Since the variance approaches a constant, so does the required number of sequences for fixed values of the confidence parameters. From here on out, the `large sequence limit' means the regime of $m$ where $u^{2(m-1)} \ll 1$ so that the variance bound (and thus the number of sequences) is approximately independent of $m$.

\begin{figure}
	\centering
	\includegraphics[width=1\linewidth]{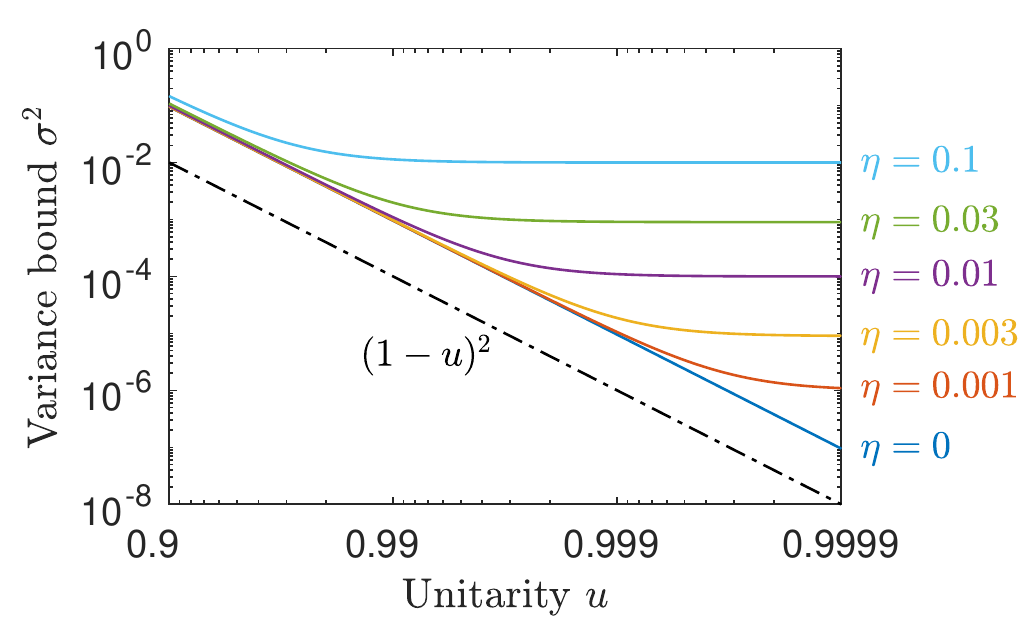}
	\caption{Semilogarithmic plot of the variance bound $\sigma^2$ as a function of the unitarity $u$ for various magnitudes of SPAM errors in the large sequence limit for single-qubit Clifford URB ($d=2$). The black dash-dotted line is a reference line plotting $\sigma^2 = (1-u)^2$. The differently colored solid lines indicate the various magnitudes of SPAM errors, where $\| \bar{\rho}_\spam \|_1^2 = \| \bar{E}_\spam \|_\infty^2 = \eta$. There are two regimes. For small SPAM errors and small $u$, the variance scales as $(1-u)^2$, whereas for nonzero SPAM errors and large $u$, the variance approaches a constant. }
	\label{fig:var_vs_u}
\end{figure}

Second we discuss the dependence of the variance bound on the unitarity $u$. In \autoref{fig:var_vs_u} the variance bound $\sigma^2$ is plotted as a function of the unitarity $u$ for various values of SPAM errors in the long sequence length limit. This figure shows two regimes. In the regime of low unitarity and small SPAM error, the variance is proportional to $(1-u)^2$. This is consistent with \eqref{eq:var_bound}, where the variance is dominated by the first term in this regime. However, for nonzero SPAM error and large unitarity, this behavior transitions into a constant variance. In this regime, the variance is dominated by the second, constant term (independent of $u$) in \eqref{eq:var_bound}.

The number of required sequences $N$ shows qualitatively similar behavior, but there are differences. This is due to the fact that $N$ is a nonlinear function of $\sigma^2$. In the regime of constant variance, the number of sequences is also constant. In the regime where the variance bound is proportional to $(1-u)^2$, the number of sequences also decreases as $N$ increases, but the rate depends also on the choice of $\epsilon$.

\subsection{Dependence on SPAM parameters}
\label{subsec:spam_dependence}
\begin{figure}
	\centering
	\includegraphics[width=1\linewidth]{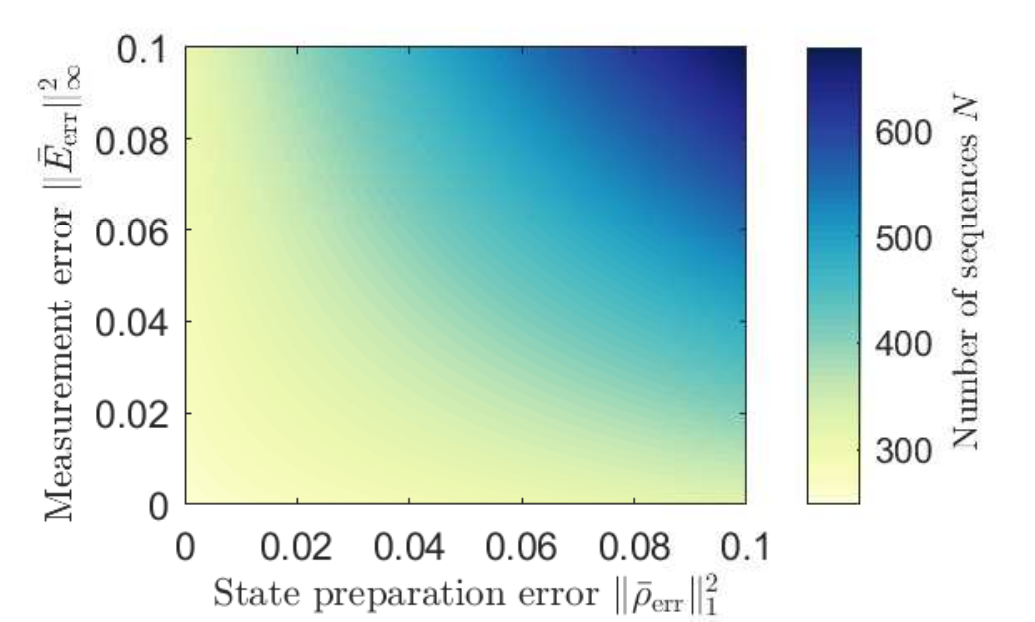}
	\caption{Color plot of the number of sequences $N$ as a function of the SPAM parameters $\| \bar{\rho}_\spam \|_1^2$ and $\| E_\spam \|_\infty^2$ in the large sequence length limit for single-qubit Clifford URB ($d=2$). The parameters $u = 0.99$ and $\epsilon = 0.02$, $\delta = 0.01$ were used. This plot illustrates the sensitivity of our result to SPAM errors. In particular, the number of sequences increases most significantly when both state preparation and measurement errors are large. }
	\label{fig:spam_colorplot}
\end{figure}
In \autoref{fig:spam_colorplot} we show a color plot of the number of sequences $N$ as a function of the SPAM parameters $\| \bar{\rho}_\spam \|_1^2$ and $\| E_\spam \|_\infty^2$ for fixed unitarity $u$ in the limit of large sequences. The plot illustrates the qualitative dependence of $N$ on the magnitude of these SPAM parameters. There are two ways that the SPAM parameters contribute to the number of required sequences $N$. First, the variance bound $\sigma^2$ depends on the SPAM parameters $\| \bar{\rho}_\spam \|_1^2$ and $\| E_\spam \|_\infty^2$ (see \eqref{eq:var_bound}). Second, the interval length bound $L$ depends on the square root of these parameters, $\| \bar{\rho}_\spam \|_1$ and $\| E_\spam \|_\infty$ (see \eqref{eq:interval_length}. Both these bounds increase as the SPAM parameters increase. From the concentration inequality \eqref{eq:hoeffding2}, it follows that the required number of sequences $N$ for fixed confidence parameters grows with increasing variance and interval length.
Both these effects have qualitatively similar behavior. This translate into the illustrated dependence of the number of sequences $N$ on the SPAM parameters in \autoref{fig:spam_colorplot}. In particular, the number of sequences most strongly depends on the product between the two, showing a larger required number in the area where the product $\| \bar{\rho}_\spam \|_1^2 \| E_\spam \|_\infty^2$ is largest.

The variance bound of \eqref{eq:var_bound} has a constant term $\| \bar{\rho}_\spam \|_1^2 \| E_\spam \|_\infty^2$, independent of the unitarity $u$ and sequence length $m$. In particular this means that the variance bound is nonzero in the presence of SPAM error for all sequence lengths $m$ even in the limit of ideal gates $\Lambda \rightarrow \Ic$. This behavior is also seen in \autoref{fig:var_vs_u}. We argue that this is fundamental to the URB protocol, by showing that the actual variance of the sequence purity $\Vbb[q_\jbf]$ also has this behavior even when ideal gates are considered. This is done in \autoref{ex:nonzero_variance}. In this example we construct noisy operators $\bar{\rho}$ and $\bar{E}$ such that the average sequence purity $q_\jbf$ is not constant over all possible ideal gate sequences $\Gc_\jbf$ (i.e., sequences with $\Lambda = \Ic$). Thus there exists an error channel (namely $\Lambda = \Ic$) and noisy operators (namely those constructed in \autoref{ex:nonzero_variance}) such that the variance, and thus the required number of sequences, is nonzero. This behavior is in contrast with standard RB, where all RB gate sequences compose to the identity when $\Lambda = \Ic$  (in the RB protocol, a global inverse gate is applied after each sequence). Therefore in standard RB, the survival probability does not depend on the sequence in the absence of gate errors and hence the variance is zero.

\begin{example}
	\label{ex:nonzero_variance}
	Consider a URB experiment where the gate set under investigation is the single-qubit Clifford group  $\Cliff{2}$. Suppose that the gates are implemented perfectly, i.e, $\Lambda = \Ic$. Furthermore assume that the state and measurement operators are given by
	\begin{equation}
	\rho, \hat{\rho} = \frac{I\otimes I \pm X \otimes X}{4}, \quad\quad\mbox{and}\quad\quad E = {X \otimes X},
	\end{equation}
	where $I$ is the identity and $X$ is the Pauli-$X$ matrix on the single-qubit Hilbert space $\Hc \simeq \C^2$. Since $\Lambda = \Ic$, the sequence $\Gc_\jbf$ of $m$ independently and uniformly distributed Clifford gates reduces to a single Clifford gate $\Gc_{i}$ uniformly drawn from $\Cliff{2}$. The group $\Cliff{2}$ has 24 elements, eight of which map $X \mapsto \pm X$. Whether such a map sends $X$ to $+X$ or $-X$ is irrelevant, since if $\Gc$ maps $X \mapsto \pm X$ then $\Gc^\tp{2}$ maps $X^\tp{2} \mapsto X^\tp{2}$ in either case. The other 16 Clifford gates send $X \mapsto \pm Y$ or $X \mapsto \pm Z$, where again the sign is irrelevant. Thus, given that $\bar{\rho} = \frac{X\otimes X}{4}$, a fraction $\frac{8}{24}$ of all sequences $\Gc_\jbf$ will satisfy $\Gc_\jbf^\tp{2}(\bar{\rho}) = \frac{X\otimes X}{4}$ while the others will send $\bar{\rho}$ either to $\frac{Y \otimes Y}{4}$ or $\frac{Z\otimes Z}{4}$.	
	Since $\Tr[E  (\frac{X\otimes X}{4})] = 1$ and $\Tr[E (\frac{Y\otimes Y}{4})] = \Tr[E  (\frac{Z\otimes Z}{4})] = 0$, the following probability distribution on $q_\jbf^{(2)}$ is obtained:
	\begin{equation}
	\Pbb\left[q_\jbf^{(2)} = 1\right] = \frac{1}{3} \quad\quad\mbox{and}\quad\quad \Pbb\left[q_\jbf^{(2)} = 0\right] = \frac{2}{3}.
	\end{equation}
	Clearly then $\Ebb[q_\jbf^{(2)}] = \frac{1}{3}$ and $\Vbb[q_\jbf^{(2)}] = \frac{2}{9} > 0$. This example shows that the variance $\Vbb[q_\jbf]$ of the sequence purity can not go to zero as the unitarity $u \rightarrow 1$.
\end{example}

Given noisy implementations $\bar{\rho}$ and $E$ in the two-copy implementation, the SPAM parameters $\| \rho_\spam \|_1^2$ and $\| \bar{E}_\spam \|_\infty^2$ defined in \eqref{eq:alpha} can in principle be estimated by relating them to  the ideal states and measurements of \eqref{eq:ideal_operators}. In practice, this requires (partial) knowledge of the noisy operators $\bar{\rho}$ and $E$. If a full (tomographic) description of $\rho,\hat{\rho}, E$ is available, then $\| \bar{\rho}_\spam \|_1^2$ and $\| \bar{E}_\spam \|_\infty^2$ can be calculated from the definition \eqref{eq:alpha}. However, if only partial knowledge is available (e.g., a lower bound on state preparation fidelity), then the SPAM quantities need to be bounded.  For example $\| \bar{\rho}_\spam \|_1^2$ can be upper bounded if the fidelity between $\rho$ ($\hat{\rho}$) and $\rho_\ideal$ ($\hat{\rho}_\ideal$) is known, by application of the Fuchs-Van de Graaff inequality \cite{Fuchs1999}. 
In the single-copy implementation, slightly more work is needed. The SPAM parameters are then defined with respect to $\bar{\rho}_\mathrm{eff}$ and $\bar{E}_\mathrm{eff}$ (\eqref{eq:state_12}). However, only (partial) knowledge of the physical operators $\rho_\Hc$ and $E_\Hc$ are available. Noise on these physical operators needs to be translated to noise on the effective operators $\bar{\rho}_\mathrm{eff}$ and $\bar{E}_\mathrm{eff}$.

\subsection{Dimension-dependent constants}
\label{subsec:dim_dependence}
\begin{figure}
	\centering
	\includegraphics[width=1\linewidth]{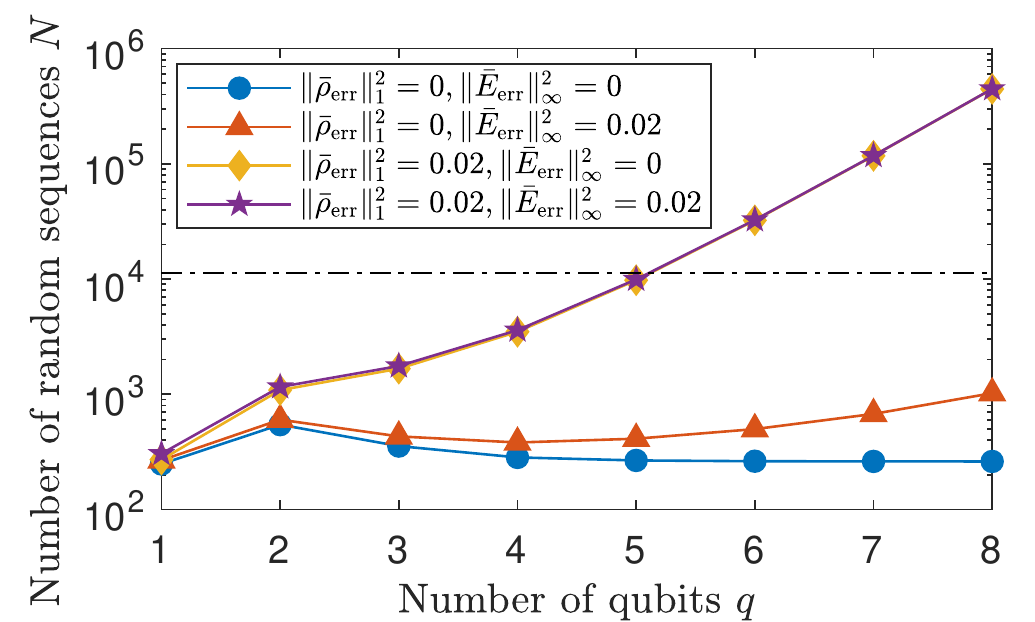}
	\caption{Number of sequences $N$ as a function of the number of qubits $q$ comprising the system for different values of the SPAM parameters. A fixed unitarity $u = 0.99$ and the large sequence length limit are used. The interval bound $L$ is computed using \eqref{eq:interval_length} as a function of the SPAM quantities (see legend). The confidence parameters $\epsilon = 0.02$, $\delta = 0.01$ were used. The dashed line indicates the first-order bound (\eqref{eq:Hoeffding1}) corresponding to $\| \bar{\rho}_\spam \|_1^2 = \| \bar{E}_\spam \|_\infty^2 = 0.02$. For the given confidence and SPAM parameters, our bound gives an improvement of the required number of sequences up to five-qubit systems.  }
	\label{fig:N_vs_q}
\end{figure}

In this section, the dependence of the variance bound \eqref{eq:var_bound} and consequently the number of sequences on the system size  is examined. An undesirable feature of the variance bound is the asymptotic growth of the constants $c_2(d)$ and $c_3(d)$ with the dimension $d = 2^q$ of the $q$-qubit system. This means that for large systems, the bound becomes loose and ultimately vacuous. This is illustrated in \autoref{fig:N_vs_q}, where the number of sequences $N$ is plotted as a function of the system size $q$ on a semilogarithmic scale (for fixed unitarity $u$ and large sequence length $m$). The number of sequences is plotted in the absence of SPAM error, with state preparation or measurement error only and with both errors simultaneously. This is done to distinguish the different contributions of the constants $c_1$, $c_2$ and $c_3$ in \eqref{eq:var_bound}. In the absence of SPAM error, only $c_1$ is relevant. This constant takes its maximum at $q=2$ and asymptotically goes to 1. However with measurement error, the number of sequences needed grows exponentially with the system size. With state preparation error, this expectational growth is even faster. This is consistent with the asymptotic limits of the constants $c_2 = \Oc(d)$ and $c_3 = \Oc(d^2)$, since $d = 2^q$. In particular, this figure shows that our variance bound is prohibitively loose for $q \geq 6$ (assuming $u=0.99$ and large $m$), since the first order bound \eqref{eq:Hoeffding1} yields a smaller number of sequences $N$ as indicated by the black dash-dotted line in the figure.

We believe that the unbounded growth of our variance bound with the system size is an artifact of the proof rather than a fundamental property. The sequence purity $q_\jbf$ is a bounded, discrete random variable, where the bound $L$ does not depend on the dimension $d$. Therefore the exact variance $\Vbb[q_\jbf]$ can not asymptotically grow with the system dimension $d$. The bound of \eqref{eq:var_bound} is, however, sharp enough for practical use in few-qubit systems.

%----------------------------------------------------------------------------------------
%	SECTION 4: Exact statement and proof overview
%----------------------------------------------------------------------------------------
\section{Methods}
\label{sec:methods}
This section gives an high level overview of the methods used for deriving our main result \eqref{eq:hoeffding2} and \eqref{eq:var_bound}. In \autoref{subsec:estimation} we focus on the statistical aspect of our result related to \eqref{eq:hoeffding2}. We also relate the between-sequence variance $\Vbb[q_\jbf]$ (the quantity which we bounded in this work) to the within-sequence variance that arises due to the fact that $q_\jbf$ can be estimated by only collecting a finite sample of single-shot measurements for a given sequence. In \autoref{subsec:fit_model_variance_expression} we discuss the derivation of the fit model (as derived in \cite{Wallman2015a}) and derive an expression for the variance $\Vbb[q_\jbf]$. In \autoref{subsec:variance_bound_proof_sketch} we give an outline of the proof of our variance bound \eqref{eq:var_bound}.

\subsection{Estimation theory}
\label{subsec:estimation}
Ultimately, the URB protocol leads to the complex statistical estimation problem of determining $u$ and the confidence thereof, given a large set of realizations of the sequence purity $q_\jbf$ (for multiple sequence lengths $m$). There are several ways one can go about this problem (see e.g., \cite{Hincks2018} for a Bayesian inference approach). In this paper we take a frequentist approach and determine a confidence interval for the point estimates $\bar{q}_m$ of $\Ebb[q_\jbf]$. These confidence intervals (for different values of $m$) can then be taken into account when fitting the point estimates $\bar{q}_m = B u^{m-1}$ to the fit model. The main contribution of this work is improving the confidence interval of $\bar{q}_m$ by bounding the variance of the sequence purity $q_\jbf$. This variance bound provides strictly more information on the distribution of $q_\jbf$ than what was known before \cite{Wallman2015a} and could therefore also be of value when using other estimation techniques to extract the unitarity $u$ from the set of measurement outcomes.

The intuitive idea is that estimating the mean of a bounded distribution of random variables requires fewer samples when the distribution is narrowly peaked around the mean. Since the variance is a measure of the spread of the distribution, it is intuitive that having knowledge of the variance improves the confidence in the estimate of the mean. This idea is made precise in statistics by concentration inequalities. Here we use a concentration inequality due to Hoeffding \cite{Hoeffding1963}. Given a collection of $N$ independent and identically distributed (i.i.d.) random variables $X_i$, sampled from a distribution on a length $L$ interval with mean $\mu$ and variance $\sigma^2$, the following statement holds for all $0 \leq \epsilon \leq L$
\begin{equation}
\Pbb\left[|\bar{X} - \mu| \geq \epsilon \right] \leq 2 \left[\left[ \frac{L}{L -\epsilon}\right]^{\frac{L^2 -\epsilon L}{{\sigma}^2 + L^2}} \left[\frac{{\sigma}^2}{{\sigma}^2 + \epsilon L}\right]^{\frac{{\sigma}^2+\epsilon L}{{\sigma}^2+L^2}}\right]^N,
\end{equation}
where $\bar{X} = \frac{1}{N} \sum_i X_i$ is the empirical mean. This is essentially \eqref{eq:hoeffding2} using the fact that $q_\jbf$ are i.i.d. random variables. The point is that if one wishes to bound this probability by $\delta$, then upper bounding the right-hand-side by $\delta$ gives a means to relate $N$, $\delta$ and $\epsilon$. Instead of the exact (unknown) variance of the distribution of $q_\jbf$, an upper bound is used.

The fact that our variance bound \eqref{eq:var_bound} depends on the unitarity $u$, the quantity that one ultimately attempts to estimate, may seems strange and circular. But this is actually a feature of statistics, which is more apparent in the Bayesian view. One may have an \textit{a priori} distribution of the unitarity $u$ of the gate set and given some experimental data (the complete URB data set) one can construct a more concentrated \textit{a posteriori} distribution on the unitarity. In the frequentist view, an \textit{a priori} lower bound to the unitarity can be known with very high confidence. Then performing URB will improve the estimate of the unitarity and increase the confidence in this estimate. In principle this procedure can be done by doing several successive URB experiments, further increasing the confidence in the outcome. Note that a first lower bound can always be obtained from the average gate fidelity (by application of \eqref{eq:agf}), which is estimated using standard RB.

Finally there is one subtlety that deserves some attention. The protocol requires the experimenter to measure $\Tr[E \Gc_\jbf^{\otimes 2}({\rho})]$, but actually this is an expectation value of the measurement operator $E$ (a Hermitian observable) given the state $\Gc_\jbf^{\otimes 2}({\rho})$. This expectation value must be learned from multiple single-shot measurements of preparing the state, apply gates and measure. The outcome is inherently probabilistic (with a Bernoulli distribution) by the laws of quantum mechanics and either a click or no click is observed with the probability given by Born's rule. To estimate the expectation value $\Tr[E \Gc_\jbf^{\otimes 2}({\rho})]$, a large number of single-shot measurements must be taken and the proportion of clicks is an estimate $\Tr[E  \Gc_\jbf^{\otimes 2}({\rho})]$. In reality then, there is also some uncertainty in each data point $q_\jbf$, which propagates into increased uncertainty in the average $\bar{q}_m$. So far we have assumed the uncertainty in $\bar{q}_m$ is dominated by the uncertainty due to the randomly sampled sequences and not due to the uncertainty in determining each sequence purity $q_\jbf$. This assumption is motivated by experiments in which it is hard to store many sequences, but easy to repeat single-shot measurements of the same sequence. In these experiments it is then easy to do enough single-shot measurements of each $q_\jbf$, such that the uncertainty in $\bar{q}_m$ is dominated by the uncertainty due to the randomly sampled sequences. This assumption is however not fundamental but is related to classical hardware control of the experimenter. In the next section we will discuss the validity of this assumption, estimate the number of required single-shot measurements and show how this assumption can be dropped if one wishes to explicitly take into account finite sampling uncertainty.

\subsubsection{Finite sampling statistics}
\label{subsubsec:finite_sampling}
In the previous section it was discussed that the quantity $q_\jbf$ is actually not directly accessible, but must be estimated by performing a large number of single-shot measurements. Born's rule states that given a (two-valued) POVM  measurement $\{M, I-M\}$ and a state $\rho$, the probability of getting outcome $1$ (associated with $M$) is given by $\Tr[M \rho]$ and outcome $0$ (associated with $I-M$) is $1-\Tr[M \rho]$. This can be used to construct a probability distribution for a single-shot measurement of $q_\jbf^{(K)}$, given a fixed sequence indexed by $\jbf$. The distribution is determined by the definition of $q_\jbf^{(K)}$ and depends on the choice of implementation. Recall that $q_\jbf$ is calculated using the difference of two states $\bar{\rho} = \half(\rho - \hat{\rho})$. 

Let us denote $\bar{q}_\jbf$ an unbiased estimator for the exact $q_\jbf$ given a fixed sequence indexed by $\jbf$. Then there is uncertainty in $\bar{q}_\jbf$ due to the uniformly distributed random sequences $\jbf$ and due to the fact that $\bar{q}_\jbf$ is itself a random variable for fixed $\jbf$ (since it is an estimator for the exact $q_\jbf$). The contribution of each source of uncertainty can be quantified by the law of total variance \cite{Weiss2005}, which states that
\begin{equation}
\label{eq:law_of_total_variance}
\begin{split}
\Vbb[\bar{q}_\jbf] &= \Ebb[ \Vbb[\bar{q}_\jbf | \jbf] ] + \Vbb[ \Ebb[\bar{q}_\jbf | \jbf] ] \\
&= \Ebb[ \Vbb[\bar{q}_\jbf | \jbf] ] + \Vbb[ q_\jbf ].
\end{split}
\end{equation}
Here the quantity $\Vbb[\bar{q}_\jbf | \jbf]$ is referred to as the within-sequence variance (for the given sequence $\jbf$). It is the variance of the sequence purity $\bar{q}_\jbf$ given fixed $\jbf$ solely due to the finite sampling statistics. The quantity $\Vbb[ q_\jbf ]$ is the between-sequence variance of $q_\jbf$ and is solely due to the fact that the sequences $\jbf$ are sampled from a uniform distribution. This equation expresses that the total variance is the sum of the expected within-sequence variance (expected over the uniformly distributed random sequences) and the between-sequence variance. The quantity $\Vbb[q_\jbf]$ was bounded in this work (\eqref{eq:var_bound}). 

To examine the term $\Ebb[ \Vbb[\bar{q}_\jbf | \jbf] ]$ in \eqref{eq:law_of_total_variance}, an expression or bound on the within-sequence variance $\Vbb[\bar{q}_\jbf | \jbf]$ as a function of the number of single-shot repetitions is required. We will show how this is done for the two-copy implementation, leaving the more cumbersome (but in principle not more difficult) single-copy implementation as an open problem. Define the single-shot random variable by $x_r$, where the subscript $r$ indexes the different single-shot realizations (for $r = 1,..., R$ for some large $R$), by the following distribution:
\begin{equation}
\label{eq:distr}
\Pbb[x_r = y | \jbf] = \begin{cases}
a(1-b), & \mbox{if } y = 1, \\
ab + (1-a)(1-b),  & \mbox{if } y = 0, \\
(1-a)b, & \mbox{if } y = -1. \\
\end{cases}
\end{equation}
Here $a = \Tr[M \Gc_\jbf^\tp{2}(\rho)]$, $b = \Tr[M \Gc_\jbf^\tp{2}(\hat{\rho})]$ and $M = \half (I +  E)$ is the POVM element associated with the two-valued measurement $E$. The outcome $x_r = 1$ is interpreted as measuring a click only for $\rho$, outcome $x_r = 0$ corresponds to a click for both or neither states and outcome $x_r = -1$ is associated with a click only for $\hat{\rho}$. This is indeed the single-shot outcome measurement outcome of a $q_\jbf^{(2)}$ measurement, since 
\begin{equation}
q_\jbf^{(2)} = \Ebb[x_r | \jbf] = a-b = \Tr[E \Gc_\jbf^\tp{2}(\bar{\rho})].
\end{equation}
The natural unbiased estimator of $q_\jbf^{(2)}$ is then given by
\begin{equation}
\bar{q}_\jbf^{(2)} = \frac{1}{R} \sum_{r=1}^R x_r.
\end{equation}
The within-sequence variance $\Vbb[\bar{q}_\jbf^{(2)} | \jbf]$ is related to the variance of $x_r$ (which can be computed given the probability distribution \eqref{eq:distr}) using the fact that $x_r$ are i.i.d. and mutually uncorrelated random variables
\begin{equation*}
\Vbb[\bar{q}^{(2)}_\jbf | \jbf] = \Vbb\left[\frac{1}{R} \sum_{r=1}^R x_r \bigg| \jbf\right] = \frac{1}{R^2}\sum_{r=1}^R \Vbb\left[  x_r | \jbf\right] =  \frac{1}{R} \Vbb[x_r | \jbf].
\end{equation*}
This follows the definition of the variance and linearity of the expected value. The variance of $x_r$ (computed from the distribution \eqref{eq:distr}) is then
\begin{equation}
\Vbb[x_r | \jbf] = (a(1-a) + b(1-b)) \leq \frac{1}{2},
\end{equation}
where the upper bound is trivially obtained by maximizing over $0 \leq a,b \leq 1$. The within-sequence variance thus satisfies
\begin{equation}
\label{eq:within-sequence-var}
\Vbb[\bar{q}^{(2)}_\jbf | \jbf] = \frac{1}{R} (a(1-a) + b(1-b)) \leq \frac{1}{2R}.
\end{equation}
Hence for the two-copy implementation, the total variance is bounded by
\begin{equation}
\Vbb[\bar{q}^{(2)}_\jbf] \leq \sigma^2 + \frac{1}{2R},
\end{equation}
where $R$ is the number of single-shot measurements taken per sequence and $\sigma^2$ is the variance bound of \eqref{eq:var_bound}.

It may seem that the modification of the protocol to use the difference of two states $\bar{\rho}$ means that twice as many single-shot measurements must be taken. This is however not the case \cite{Helsen2017}. To see this, let $\Vbb_\rho$ be the variance associated with a single measurement setting on the state $\rho$. Then for the difference of two states, the variance associated with that measurement satisfies 
\begin{equation}
\Vbb_{\bar{\rho}} = \Vbb_{\half(\rho-\hat{\rho})} \leq \frac{1}{4}(\Vbb_\rho+\Vbb_{\hat{\rho}}) \leq \half \max(\Vbb_\rho, \Vbb_{\hat{\rho}}).
\end{equation}
So to the contrary, fewer sequences are required to get an accurate estimate of $\Tr[E \Gc_\jbf(\bar{\rho})]$ than of $\Tr[E \Gc_\jbf({\rho})]$. This can explicitly be seen in the two-copy implementation, where the within-sequence variance $\Vbb_{\bar{\rho}}[\bar{q}_\jbf | \jbf]$ was computed in \eqref{eq:within-sequence-var}. However, if only a single state $\rho$ were used, then $\Pbb[x_r = 1] = a$ and $\Pbb[x_r = -1] = 1-a$. Therefore the variance $\Vbb_\rho[\bar{q}_\jbf | \jbf] = \frac{1}{R} \Vbb_\rho[x_r | \jbf] = \frac{4a(1-a)}{R} \leq \frac{1}{R}$, which is indeed a factor 2 larger than in \eqref{eq:within-sequence-var}.

\subsubsection{The unbiased estimator of the sequence purity in the single-copy implementation}
\label{subsubsec:unbiased_estimator}
In the single-copy implementation care must be taken in defining an appropriate estimator of $q_\jbf^{(1)}$. Analogously to the above, one can define a random variable $x_r^{PQ}$ associated with a single-shot measurement of $\Tr[E_\Hc^{(Q)} \Gc_\jbf(\bar{\rho}_\Hc^{(P)})]$ for a fixed sequence indexed by $\jbf$, depending on the Pauli's $P$ and $Q$. Then
\begin{equation}
\Ebb[x_{r}^{PQ}|\jbf] = \Tr[E_\Hc^{(Q)} \Gc_\jbf(\bar{\rho}_\Hc^{(P)})],
\end{equation}
so that
\begin{equation}
q_\jbf^{(1)} = \frac{1}{d^2-1} \sum_{P,Q \neq I} \Ebb[x_r^{PQ}|\jbf]^2.
\end{equation}
If we denote $\bar{x}_{PQ} = \frac{1}{R} \sum_{r=1}^R x_r^{PQ}$, then one could try to estimate $q_\jbf^{(1)}$ by 
$\bar{q}_\jbf^{(1)} = \frac{1}{d^2-1} \sum_{P,Q \neq I} \bar{x}_{PQ}^2$. However, this estimate is biased and overestimates the actual value of $q_\jbf^{(1)}$, since
\begin{equation}
\begin{split}
\Ebb[\bar{x}_{PQ}^2|\jbf] &= \Ebb[\bar{x}_{PQ}|\jbf]^2 + \Vbb[\bar{x}_{PQ}|\jbf] \\
&= \Ebb[\bar{x}_{PQ}|\jbf]^2 + \frac{1}{R}\Vbb[x_r^{PQ}|\jbf]. \\
\end{split}
\end{equation}
To remedy this, one can make use of the unbiased estimator
\begin{equation}
\bar{q}_\jbf^{(1)} = \frac{1}{d^2-1} \sum_{P,Q \neq I} \bar{x}_{PQ}^2 - \frac{1}{R} s^2_{PQ},
\end{equation}
where
\begin{equation}
s^2_{PQ} = \frac{1}{R-1} \sum_{r=1}^R (x_r^{PQ} - \bar{x}_{PQ})
\end{equation}
is the unbiased estimate of $\Vbb[x_r^{PQ}|\jbf]$. It is important to take this into consideration when performing a Clifford URB experiment using the single-copy implementation, since overestimating $q_\jbf^{(1)}$ can lead to an overestimate of the unitarity obtained from the experiment.

\subsection{Fit model and variance expression}
\label{subsec:fit_model_variance_expression}
In this section we first briefly review the derivation of the fit model of URB (as derived in \cite{Wallman2015a}), slightly adapted with our modification of a traceless input operator $\bar{\rho}$.  Then we derive an expression for the variance of the sequence purity. We do so using slightly different notation, picking an orthonormal basis for the space of linear operators $\Lc(\Hc)$ (in particular we use the normalized Pauli operators). We can then vectorize any operator with respect to that basis, which we will denote with a braket-like notation $\rho \rightarrow \kket{\rho}$ and $E \rightarrow \bbra{E}$. Quantum channels can then be viewed as matrices on these vectors, $\Ec(\rho) \rightarrow \kket{\Ec(\rho)} = \pmb{\Ec} \kket{\rho}$, where we use boldface notation for the matrix representation of a quantum channel. The Hilbert-Schmidt inner product $\Tr[E^\dagger \rho]$, carries over as the vector inner product with respect to any basis, so that $\Tr[E^\dagger \rho] = \bbraket{E | \rho}$. Finally composition of channels $\Ec_1 \Ec_2 \rightarrow \pmb{\Ec_1} \pmb{\Ec_2}$ carries over as matrix multiplication. This notation is known as the natural representation, Liouville representation, or Pauli transfer matrix representation \cite{Wallman2014,Watrous2017}. See Appendix~\ref{app_subsub:liouville_rep} for more details.

Using this notation, the expected value of the sequence purity $\Ebb[q_\jbf]$ can be written as
\begin{equation}
\begin{split}
\Ebb[q_\jbf] &= \frac{1}{| \Cliff{d} |^m } \sum_\jbf \bbraket{E | \Gcbf_\jbf^\tp{2}  | \bar{\rho}} \\ &=  \bbraket{E | \left( \Gavgbf{2} \Lambdabf^\tp{2} \right)^m |\bar{\rho}},
\end{split}
\end{equation}
where
\begin{equation}
\Gavg{n} = \frac{1}{|\Cliff{d}|} \sum_{\Gc \in \Cliff{d}} \Gc^\tp{n}.
\end{equation}
The key idea behind deriving the fitting model is that $\Gavg{2}$ is the orthogonal projection onto the vector space $W = \Span \{ I, S \} \subset \Lc(\Hc \otimes \Hc)$. This is a result from representation theory of finite groups, see \autoref{lem:projection_onto_trivial_subreps} in Appendix~\ref{app_sub:rep_theory} for details. It is for this reason that the ideal state and measurement operators of \eqref{eq:ideal_operators} are elements of the subspace $W$. The operators $I$ and $S$ do not form an orthogonal basis for $W$, but the following orthonormal basis can be constructed:
\begin{align}
B_1 &= \frac{I}{d} = \sigma_0 \otimes \sigma_0,   \\
B_2 &= \frac{S - \frac{I}{d}}{\sqrt{d^2 -1}}  = \frac{1}{\sqrt{d^2 -1}} \sum_{\sigma \in \Pc^*} \sigma \otimes \sigma,
\end{align}
where $\sigma_0$ is the Hilbert-Schmidt normalized identity on $\Hc$ and $\sigma \in \Pc^*$ are the $d^2-1$ traceless normalized Pauli operators on $\Hc$. Since $\Gavg{2}$ is an orthogonal projection, it follows that $(\Gavg{2})^2 = \Gavg{2}$. Therefore we can rewrite
\begin{equation}
\label{eq:fitmodel_Mc}
\Ebb[q_\jbf] =  \bbraket{E | \Mcbf^{m-1}\Lambdabf^\tp{2}  |\bar{\rho}},
\end{equation}
where $\Mc = \Gavg{2} \Lambda^\tp{2} \Gavg{2}$. It can be shown that $\Mcbf$ (which as only support on $W$) has the following matrix entries \cite{Wallman2015a}
\begin{equation}
\label{eq:M_entries}
\pmb{\Mc} = \begin{bmatrix}
1 & 0 \\
\frac{\|\alpha(\Lambda)\|^2}{\sqrt{d^2-1}} & u(\Lambda) \\
\end{bmatrix},
\end{equation}
in the basis $\{B_1, B_2\}$, with $\alpha$ the nonunitality vector of $\Lambda$ (see \eqref{eq:Liouville_block_decomp} in Appendix~\ref{app_subsub:liouville_rep} for details). In particular this means that $u(\Lambda) = \bbraket{B_2 | \Lambdabf^\tp{2} |B_2}$, which might not be too surprising in view of \eqref{eq:unitary_definition2}. Since the input state $\bar{\rho}$ is traceless and quantum channels are trace preserving, \eqref{eq:fitmodel_Mc} is evaluated as
\begin{equation}
\label{eq:fitmodel}
\Ebb[q_\jbf] =  \bbraket{E | B_2} \bbraket{ B_2 | \bar{\rho}} u^{m-1} = B u^{m-1},
\end{equation}
where the final channel $\Lambda^\tp{2}$ has been absorbed into the state as state preparation error.  
The robustness to state preparation and measurement errors stems from the fact that every component of $\bar{\rho}$ and $E$ outside the subspace $W$ is projected out by the procedure.

In very similar fashion the variance, defined as $\Vbb[q_\jbf] = \Ebb[q_\jbf^2] - \Ebb[q_\jbf]^2$, can be computed. Using $\Tr[A]^2 = \Tr[A^\tp{2}]$, the mixed-product property of the tensor product [i.e., $(A\otimes B)(C \otimes D) = (AC) \otimes (BD)$] and linearity, we write 
\begin{equation}
\begin{split}
\Ebb[q_\jbf^2] &= \frac{1}{| \Cliff{d} |^m } \sum_\jbf \bbraket{E^\tp{2} | \Gcbf_\jbf^\tp{4}  | \bar{\rho}^\tp{2}} \\ 
&=  \bbraket{E^\tp{2} | \left( \Gavgbf{4} \Lambdabf^\tp{4} \right)^m |\bar{\rho}^\tp{2}}\\
&=  \bbraket{E^\tp{2} | \Ncbf^{m-1} \Lambdabf^\tp{4} |\bar{\rho}^\tp{2}},
\end{split}
\end{equation}
where $\Nc = \Gavg{4} \Lambda^\tp{4} \Gavg{4}$, using the fact that $\Gavg{4}$ is also an orthogonal projection (\autoref{lem:projection_onto_trivial_subreps} of Appendix~\ref{app_sub:rep_theory}), and
\begin{equation}
\Ebb[q_\jbf]^2 = \bbraket{E^\tp{2} | (\Mcbf^\tp{2})^{m-1} \Lambdabf^\tp{4}  | \bar{\rho}^\tp{2}}.
\end{equation}
 Putting it together yields the following expression for the variance
\begin{equation}
\label{eq:var_expr1}
\Vbb[q_\jbf] = \bbraket{E^\tp{2} | \Ncbf^{m-1} - (\Mcbf^\tp{2})^{m-1} | \bar{\rho}^\tp{2}},
\end{equation}
where the final channel $\Lambda^\tp{4}$ has again been absorbed into the state as state preparation error.
One of the key ingredients of understanding this expression is finding the subspace of $\Lc(\Hc^\tp{4})$ onto which $\Gavg{4}$ projects. The next section elaborates on this idea.

\subsection{Sketch of proof on variance bound}
\label{subsec:variance_bound_proof_sketch}
In this section we discuss and sketch the main ideas for the proof of our variance bound \eqref{eq:var_bound}. A complete proof is given in Appendix~\ref{app:variance_bound}, \autoref{thm:var_bound}. We actually prove a slightly stronger statement
\begin{equation}
\label{eq:54}
\begin{split}
\Vbb[q_\jbf] \leq& \, \| \bar{\rho}_\spam \|_1^2 \| \bar{E}_\spam \|_\infty^2 +  \frac{1-u^{2(m-1)}}{1-u^2} (1 - u)^2 \times \\
&\Big( \alpha^2 \beta^2 c_1(d) + \alpha^2 c_2(d) \| \bar{E}_\spam \|_\infty^2 + \beta^2 c_3(d) \| \bar{\rho}_\spam \|_1^2 \Big),    \\
\end{split}
\end{equation} 
where 
\begin{align}
\alpha &= \frac{\Tr[\bar{\rho}_\ideal  \bar{\rho}]}{\|\bar{\rho}_\ideal\|_2^2} = (d^2-1)\Tr[\bar{\rho}_\ideal  \bar{\rho}] \\
\beta &= \frac{\Tr[\bar{E}_\ideal  \bar{E} ]}{\| \bar{E}_\ideal \|_2^2} = \frac{\Tr[\bar{E}_\ideal  \bar{E} ]}{d^2-1}.
\end{align}
These quantities arise in the decomposition of the operators $\bar{\rho}, \bar{E}$ into an ideal and error parts as
\begin{equation}
\bar{\rho} = \alpha \bar{\rho}_\ideal + \bar{\rho}_\spam \quad\quad \mbox{and} \quad\quad \bar{E} = \beta \bar{E}_\ideal + \bar{E}_\spam.
\end{equation}
It can be shown that $-1 \leq \alpha, \beta \leq 1$ (see Appendix~\ref{app:variance_bound}, \autoref{lem:optimality_operators}), so that \eqref{eq:54} indeed implies \eqref{eq:var_bound}. The quantities $\alpha,\beta$ are generally unknown to the experimenter and therefore easily eliminated from the variance bound. Finally we remark that the bound on the interval length $L$ (given in \eqref{eq:interval_length}) can also be slightly improved if additional information on $\alpha$ or $\beta$ is known. See Appendix~\ref{app:variance_bound}, \autoref{lem:interval_length} for a precise statement.

Our analysis departs from the expression of the variance \eqref{eq:var_expr1}. First, let us note that fully characterizing the operator $\Nc$ seems infeasible. This was possible for the operator $\Mc$, since it only has support on the two-dimensional subspace $W$. However, the dimension of the support of $\Nc$ (the dimension of the space onto which $\Gavg{4}$ projects) is given by \cite{Zhu2016,Zhu2017,Helsen2018}
\begin{equation}
\label{eq:dim_rge_N}
| \Rge(\Nc) | = 
\begin{cases}
15 & \mbox{if } d=2; \\
29 & \mbox{if } d=4; \\
30 & \mbox{otherwise}. 
\end{cases}
\end{equation}
Therefore calculating the $|\Rge(\Nc)|^2$  matrix entries of $\Ncbf$ seems infeasible. A different approach is thus needed. We use a telescoping series expansion (see \autoref{lem:telescoping} in Appendix~\ref{app_sub:lemma_literature} for a proof) 
\begin{equation}
\label{eq:telescoping}
\Ncbf^{m} - (\Mcbf^\tp{2})^{m} = \sum_{s=1}^{m}  \Ncbf^{m-s}[\Ncbf - \Mcbf^\tp{2}](\Mcbf^\tp{2})^{s-1}
\end{equation}
in \eqref{eq:var_expr1}. The main idea of this is to study the middle operator $\Ncbf - \Mcbf^\tp{2}$ carefully and sharply bound the relevant matrix entries of this operator. The action of $(\Mcbf^\tp{2})^{s-1}$ is well understood because the full 2-by-2 matrix description of $\Mcbf$ is known (given in \eqref{eq:M_entries}). Finally the action of the remaining higher powers $\Ncbf^{m-s-1}$ are bounded more trivially, since less information in computed about $\Ncbf$. Let us make these ideas more precise now.

In the previous it was discussed that the operator $\Mcbf$ only has support on the subspace $W = \Span\{I, S\} = \Span\{B_1, B_2\}$. Therefore the analysis of the variance expression is quite different for the components of $\bar{\rho}$ and $E$ on the subspace $W$ and its orthogonal complement. In fact, this lead to the decomposition of the operators $\bar{\rho}, \bar{E}$ into an ideal and error parts as
\begin{equation}
\bar{\rho} = \alpha \bar{\rho}_\ideal + \bar{\rho}_\spam \quad\quad \mbox{and} \quad\quad \bar{E} = \beta \bar{E}_\ideal + \bar{E}_\spam,
\end{equation}
where the bar over $E$ indicates its traceless component. In fact, the identity component of $E$ does not contribute at all to $q_\jbf$ (and therefore to its mean and variance), because the input operator is traceless and all applied maps $\Gc_\jbf$ are trace preserving. So the traceless ideal components are in the traceless subspace of $W$ (spanned by $B_2$) and the error components are in the orthogonal complement $W^\perp$. In principle, plugging the above expansion into \eqref{eq:var_expr1} yields 16 different terms after distributing the tensor powers in $\bar{\rho}$ and $E$ over the sum. However, 12 factors containing mixed tensor products of ideal and error components (e.g., $\bar{E}_\ideal \otimes \bar{E}_\spam$) vanish. This is due to the structure of the space onto which $\Gavg{4}$ projects (see Appendix~\ref{app_sub:trivial_subreps_Liouville_4} for more details).
Thus we expand \eqref{eq:var_expr1} as
\begin{align}
\label{eq:ideal_ideal}
\Vbb[q_\jbf ] =  \alpha^2 \beta^2 \bbraket{\bar{E}_\ideal^{\otimes 2} | \pmb{\Nc}^{m-1} - (\pmb{\Mc}^{\otimes2})^{m-1} | \bar{\rho}_\ideal^{\otimes 2}}& \\ 
\label{eq:spam_ideal}
+ \alpha^2 \bbraket{\bar{E}_\spam^{\otimes 2} | \pmb{\Nc}^{m-1} - (\pmb{\Mc}^{\otimes2})^{m-1} | \bar{\rho}_\ideal^{\otimes 2}}& \\ 
\label{eq:ideal_spam}
+ \beta^2 \bbraket{\bar{E}_\ideal^{\otimes 2} | \pmb{\Nc}^{m-1} - (\pmb{\Mc}^{\otimes2})^{m-1} | \bar{\rho}_\spam^{\otimes 2}}& \\ 
\label{eq:spam_spam}
+ \bbraket{\bar{E}_\spam^{\otimes 2} | \pmb{\Nc}^{m-1} - (\pmb{\Mc}^{\otimes2})^{m-1} | \bar{\rho}_\spam^{\otimes 2}}&.
\end{align} 
Each of these terms is bounded separately. Here we will demonstrate the ideas of our proof using the term of \eqref{eq:spam_ideal}. The two terms \eqref{eq:ideal_ideal} and \eqref{eq:ideal_spam} are similar (only a few technical details are different; see \autoref{thm:var_bound} in Appendix~\ref{app_sub:statistical_results_proof} for precise treatment of all terms). Using the telescoping series \eqref{eq:telescoping} term \eqref{eq:spam_ideal} can be written as
\begin{equation}
\begin{split}
(\ref{eq:spam_ideal}) &= \alpha^2\sum_{s=1}^{m-1}  \bbraket{\bar{E}_\spam^\tp{2} | \Ncbf^{m-s-1}[\Ncbf - \Mcbf^\tp{2}](\Mcbf^\tp{2})^{s-1} | \bar{\rho}_\ideal^\tp{2}} \\
&=\alpha^2\sum_{s=1}^{m-1} u^{2(s-1)}  \bbraket{\bar{E}_\spam^\tp{2} | \Ncbf^{m-s-1}[\Ncbf - \Mcbf^\tp{2}] | \bar{\rho}_\ideal^\tp{2}},
\end{split}
\end{equation}
where the second line follows from the fact that $\Mcbf \kket{B_2} = u \kket{B_2}$ and $\bar{\rho}_\ideal = \frac{1}{\sqrt{d^2-1}} B_2$ (see \eqref{eq:ideal_operators_app} in Appendix~\ref{app:variance_bound}). 
The next step is analyzing 
\begin{equation}
\label{eq:expansion}
\pmb{\Nc} - \pmb{\Mc}^{\otimes2}  \kket{ \frac{1}{d^2-1}  B_2^{\otimes 2}} = \frac{1}{d^2-1}  \sum_i a_i \kket{A_i}
\end{equation}
where $a_i = \bbraket{A_i | \pmb{\Nc} - \pmb{\Mc}^{\otimes2} | B_2^{\otimes 2}}$ and $\kket{A_i}$ is a basis for the space $\Rge(\Gavgbf{4})$ on which $\Ncbf$ has support. To find the basis $\kket{A_i}$ explicitly, the following ideas from representation theory are used (see Appendix~\ref{app_sub:trivial_subreps_Liouville_4} for details). 

The map $\Gc \mapsto \Gcbf^\tp{n}$ is a group representation of the Clifford group $\Cliff{d}$ for any $n$. A fundamental result in group representation theory \cite{Fulton2004} (\autoref{lem:projection_onto_trivial_subreps} in Appendix~\ref{app_sub:rep_theory}) is that $\Gavg{n}$ is the orthogonal projection onto the trivial subspace of the representation $\Gc \mapsto \Gcbf^\tp{n}$. For $n = 2$, the trivial subspace was found to be the space $W$ \cite{Wallman2015a}, giving rise to the fit model of \eqref{eq:fitmodel}. The task at hand here is to find the trivial subspace for $n = 4$. To do so, the following is used. If $(V,R)$ is an irreducible, real representation of a group $\Cliff{d}$, then \cite{Fulton2004}
\begin{equation}
\label{eq:trivial_subpreps}
(\Span \{\sum_{v \in V} v \otimes v\}, R \otimes R)
\end{equation}
is the only trivial representation of $V \otimes V$ of the group $\Cliff{d}$ (see \autoref{lem:trivial_subrep_lemma} in Appendix~\ref{app_sub:rep_theory}). This allows us to calculate all trivial subrepresentations of $\Gc \mapsto \Gcbf^\tp{4}$, using a complete description of the irreducible representations of $\Gc \mapsto \Gcbf^\tp{2}$. These were found in \cite{Helsen2018,Zhu2016}. Therefore \eqref{eq:trivial_subpreps} provides a method to compute the $\kket{A_i}$ using the explicit description of the irreducible spaces of $\Gc \mapsto \Gcbf^\tp{2}$ found in \cite{Helsen2018}.

Hence, the following expression is obtained for \eqref{eq:spam_ideal}, using the expansion \eqref{eq:expansion}:
\begin{equation}
(\ref{eq:spam_ideal}) = \frac{\alpha^2}{d^2 - 1} \sum_{s=1}^{m-1} u^{2(s-1)} \sum_i a_i \bbraket{\bar{E}_\spam^\tp{2} | \Ncbf^{m-s-1} | A_i },
\end{equation}
where $a_i = \bbraket{A_i | \pmb{\Nc} - \pmb{\Mc}^{\otimes2} | B_2^{\otimes 2}}$ are the coefficients of the expansion. 
The factor $\frac{1}{d^2-1}$ is later absorbed into the constant $c_2(d)$ in the final result. Up until this point, equality still holds. Now we are finally in a position to start bounding the term \eqref{eq:spam_ideal}. To do so, we upper bound each $a_i$. These bounds involve constants depending on the dimension $d$ (which are all absorbed into $c_2(d)$) and are proportional to $(1-u)^2$. Finally the inner product containing $\Ncbf^{m-s-1}$ is upper bounded by a constant  depending on the dimension and proportional to $\|\bar{E}_\spam \|_\infty^2$ (and in particular independent of $m$ or $s$). 
This then gives a total bound on the term \eqref{eq:spam_ideal}, 
\begin{equation}
(\ref{eq:spam_ideal}) \leq \frac{1-u^{2(m-1)}}{1-u^2} (1 - u)^2 \alpha^2 c_2(d) \| \bar{E}_\spam \|_\infty^2,
\end{equation}
where we used the geometric series
\begin{equation}
\sum_{s=1}^{m-1} u^{2(s-1)} = \frac{1-u^{2(m-1)}}{1-u^2}.
\end{equation}
The terms \eqref{eq:ideal_ideal} and \eqref{eq:ideal_spam} can be bounded by repeating all these steps, using a different telescoping series expansion where the factors $(\Mcbf^\tp{2})^{s-1}$ and $\Ncbf^{m-s}$ are interchanged in \eqref{eq:telescoping}. The analysis is then performed by simplifying the inner product from left to right.  This involves a few technicalities, but no new ideas. In the end, only the bound on the final inner product with $\Nc^{m-s-1}$ and the proportionality constants $c_1(d), c_3(d)$ differ, as can be seen from the result \eqref{eq:54}.
Finally for the final term \eqref{eq:spam_spam}, there is not much more to do than
\begin{equation}
(\ref{eq:spam_spam}) = \bbraket{\bar{E}_\spam^{\otimes 2} | \pmb{\Nc}^{m-1} | \bar{\rho}_\spam^{\otimes 2}} \leq \| \bar{E}_\spam \|_\infty^2  \| \bar{\rho}_\spam \|_1^2,
\end{equation}
using H\"older's inequality and the fact that $\Nc$ is contractive in the induced trace norm \cite{Perez-Garcia2006}, i.e., $\| \Nc \|_{1\rightarrow 1} \leq 1$ (see \autoref{prop:inner_holder} in Appendix~\ref{app:technical_lemmas}).

%----------------------------------------------------------------------------------------
%	SECTION: Conclusions
%----------------------------------------------------------------------------------------
\section{Conclusion and future work}
\label{sec:Outlook}
In this work we have shown a significant reduction in the required number of random sequences for unitarity randomized benchmarking (URB) than previously could be justified. This reduction is achieved by analyzing the statistics of the protocol. In particular, we have provided a bound on the variance of the sequence purity. Application of a concentration inequality yields the reduction in number of sequences, provided that the variance bound is sharp enough. We have shown that in realistic parameter regimes, the required number of sequences is in the order of hundreds, when benchmarking few-qubit Clifford gates. This brings benchmarking the unitarity of few-qubit Clifford gates into the realm of experimental feasibility.

The main ingredient of this result was a sharp bound on the variance of the sequence purity. The analysis was done for a slightly modified version of the protocol. This modification leads to better guarantees on the confidence and additionally yields a linear fitting problem. Our variance bound has the attractive property that it scales quadratically in $1-u$, where $u$ is the unitarity, up to constant contribution due to state preparation and measurement (SPAM) errors. This implies that fewer sequences are required to estimate highly coherent gates. We show that the constant contribution due to SPAM errors is a fundamental property of URB (and therefore not an artifact of our bound). Furthermore our bound is asymptotically independent of the sequence length and is therefore applicable in both short and long sequence lengths. Finally our bound grows exponentially in the number of qubits comprising the system. We argue that this is an artifact of the bound, which could be improved upon. As a result, our bound becomes vacuous for large systems. However, we have shown that our bound is sharp enough to benchmark few-qubit systems (say, up to five qubits).

During the analysis of the URB protocol, we have emphasized two different implementation techniques. We have explicitly shown their optimal state preparation and measurement settings for practical implementation. We highlighted the benefits and drawbacks of each implementation and showed the statistical difference between the two.

\paragraph*{Future work.} There are a few caveats in the analysis of this work, which arise from the assumptions under which the bound holds. Each of these assumptions as summarized in \autoref{sec:results} is an open avenue for future research. First and foremost, the assumption of the gate independent error model is rather strong and never completely satisfied in practical implementations of gates. 
The analysis of the URB protocol so far has been restricted to the gate-independent noise model \cite{Wallman2015a}. There are three somewhat independent open problems with the URB protocol when one wants to generalize the model to (Markovian) gate-dependent errors. First, the behavior of the protocol must be studied. This means that the validity and deviation of the fit model must be studied under this more general noise model. Second, the statistics of the protocol can be studied in the gate-dependent error model. This aims to provide an answer to the question how many resources are required to extract the unitarity from measurement data in this more general noise model, provided that a generalized fit model is found. Finally one can attempt to relate the URB decay rate(s) in the gate-dependent setting to physically relevant quantities (like the unitarity) of the gates comprising the gate set. All three of these problems relating to gate-dependent errors are tough problems and many research focused on answering analogous questions for standard RB. For standard RB, progress has been made in terms of understanding the fit model and relating the decay rate to a physically interpretable infidelity in the gate-dependent error model  \cite{Chasseur2015,Proctor2017,Wallman2018}. However, statistical analyses of standard RB only apply to the gate-independent error model \cite{Wallman2014,Helsen2017,Granade2015}. We suspect that some of the progress made in analyzing gate-dependent RB can be modified and applied to URB, but we have left this for future work.

A second interesting avenue is exploring how unitarity randomized benchmarking behaves when the assumption of unitary 2-design is relaxed \cite{Dankert2006}. This would give rise to a protocol that can benchmark the unitarity of different gate sets that do not form a 2-design. Interesting examples are the Dihedral group \cite{Carignan-Dugas2015, Cross2015}, subgroups of monomial unitary matrices \cite{Franca2018} and subgroups of the Clifford group \cite{Hashagen2018,Brown2018}, where progress have been made for standard RB. Note that the first two of these gate sets are particularly interesting since they contain the $T$-gate. A general framework for standard RB given an arbitrary gate set is provided in \cite{Helsen2018a}. An interesting open question is whether these techniques can be applied to URB.

Finally it is interesting if the current limitations of our bound can be improved upon. In particular an open question is how to improve this bound to be asymptotically independent of the dimension, a caveat that currently renders our bound impractical for large system ($q \gg 5$). Similarly we wonder if our bound can be generalized to general multiqubit noise models that need not be unital. These lines of future work could improve the applicability of our bound.

%----------------------------------------------------------------------------------------
%	ACKNOWLEDGEMENTS
%----------------------------------------------------------------------------------------
\begin{acknowledgments}
	The authors would like to thank Michael Walter for inspiring discussions on the topic. B.D., J.H. and S.W. are funded by NWA, a NWO VIDI grant, an ERC Starting Grant QINTERNET, and NWO Zwaartekracht QSC.
\end{acknowledgments}

%----------------------------------------------------------------------------------------
%	APPENDICES
%----------------------------------------------------------------------------------------

\onecolumngrid
\appendix
\section{Preliminaries}
The appendices are devoted to proving the upper bound \eqref{eq:var_bound} (actually we prove \eqref{eq:54}, which implies \eqref{eq:var_bound}) on the variance of the sequence purity for Clifford Unitarity Randomized Benchmarking. To do so, this appendix first provides an overview of the preliminaries and sets the formal notation used in the rest of the appendices. The material covered in this appendix is not a new result. In Appendix~\ref{app:variance_bound} then the variance bound of \eqref{eq:var_bound} is proven. It also contains the proof of the interval of the sequence purity (\eqref{eq:interval_length}). Finally, all technical lemmas used in the proof of the variance bound are collected in Appendix~\ref{app:technical_lemmas}. The material in Appendices~\ref{app:variance_bound} and \ref{app:technical_lemmas} is the main result of this work.

\subsection{Notation and definitions}
In this subsection we summarize all notation used in the paper and the appendices. Suppose our principle system under investigation is a $q$-qubit system. Its state space is then represented by a $d$-dimensional Hilbert space $\Hc$, where $d=2^q$. Typically $\Hc$ is identified with $\C^d$. General vector spaces are typically denoted $V$. The dimension of a vector space is denoted $|V| = \dim(V)$. Hence $d = 2^q = | \Hc |$. 
The set of linear operators between two vector spaces $V_1, V_2$ is denoted $\Lc(V_1, V_2)$ (some references write $\mathrm{Hom}(V_1,V_2)$). We write $\Lc(V)$ as shorthand for $\Lc(V,V)$ (in the literature also written as $\mathrm{End}(V)$). It is convenient to think of $\Lc(\Hc)$ as a Hilbert space in itself, equipped with the Hilbert-Schmidt inner product. This inner product is defined as $\braket{A, B}_\mathrm{HS} = \Tr[A^\dagger B]$ for any $A,B \in \Lc(\Hc)$. It induces the Hilbert-Schmidt norm $\| A \|_2 = \sqrt{\braket{A,A}_\mathrm{HS}}$. This is in fact a special case of the more general Schatten $p$-norms (for $1 \leq p \leq \infty$), which are defined as
\begin{equation}
\| A \|_p^p = \Tr \left[(A^\dagger A)^\frac{p}{2}\right]  = \| s(A) \|_p^p = \sum_i s_i(A)^p.
\end{equation}
Here $s(A)$ denotes the vector of singular values $s_i(A)$ of $A$. The Hilbert-Schmidt norm corresponds to $p=2$. Other important special cases are the trace norm ($p=1$) and the operator norm to ($p=\infty$).

%\subsubsection{The normalized Pauli matrices}
The normalized Pauli-matrices form an orthonormal basis of $\Lc(\Hc)$ with respect to the Hilbert-Schmidt inner product. The set of normalized Pauli's is denoted
\begin{equation}
\Pc := \left\{ \frac{P}{\sqrt{d}} \bigg|  P \in  \{I,X,Y,Z\}^{\otimes q}  \right\},
\end{equation}
where $I,X,Y,Z$ denote the usual (unnormalized) Pauli matrices. The set of traceless Pauli-matrices is denoted $\Pc^* = \Pc \setminus \{ \sigma_0 \}$, where $\sigma_0 := \frac{1}{\sqrt{d}}  I^{\otimes q}$ is the normalized identity. Elements of $\Pc$ are denoted by the Greek symbols $\sigma, \tau$. For two normalized Pauli matrices $\sigma, \tau \in \Pc$, we define the normalized matrix product $\sigma \cdot \tau := \frac{1}{\sqrt{d}} \sigma\tau$. This ensures that $\| \sigma \cdot \tau \|_2 = 1$ so that $\sigma \cdot \tau \in \pm \Pc$. The tensor product between two Pauli matrices can then be conveniently omitted, so that $\sigma \tau := \sigma \otimes \tau$. This is used for brevity when writing many tensor products of normalized Pauli matrices. From here on out, we will omit the tensor product.
Finally for every normalized Pauli $\tau \in \Pc$, we define $C_\tau$ as the set of all elements of $\Pc^*$ that commute with $\tau$, except for $\tau$ itself \cite{Helsen2018}:
\begin{equation}
\label{eq:def_Ctau}
C_\tau := \{  \sigma \in \Pc^* :   \sigma \cdot \tau = \tau \cdot \sigma)   \}.
\end{equation}
In \cite{Helsen2018} it is shown that $| C_\tau | = \frac{d^2 - 4}{2}$.

The Clifford group, denoted $\Cliff{d}$, has a natural action by conjugation on the set of Pauli matrices $\Pc$. Informally speaking, the Clifford group sends Pauli matrices to Pauli matrices under conjugation. More formally speaking, the Clifford group is the normalizer of the Pauli group (the group generated by $\Pc$) in the unitary group, up to global phase:
\begin{equation}
\Cliff{d} := \{ U \in \Unitary{d} :  U \sigma U^\dagger \in \pm \Pc, \, \forall \sigma \in \Pc  \} \, / \,  \Unitary{1}.
\end{equation}
An alternative description of the Clifford group is given in terms of its generators. The group is generated as
\begin{equation}
\Cliff{d} = \braket{\{H_i, S_i, CNOT_{ij} | i,j = 1,...,q, \quad i \neq j \}} \, / \,  \Unitary{1},
\end{equation}
where $H_i$ is the Hadamard gate and $S_i$ is the $\frac{\pi}{4}$-phase gate on qubit $i$, and $CNOT_{ij}$ is the CNOT gate on qubits $i,j$.
For a more detailed introduction into the Pauli and Clifford group, see \cite{Farinholt2014} and references therein.
The size of the Clifford group is \cite{Ozols2008}
\begin{equation}
|\Cliff{d}| = \prod_{j=1}^{q} 2(4^j - 1)4^j = 2^{\Oc(q^2)}.
\end{equation}

\subsubsection{States, measurements and quantum channels}
In quantum mechanics, quantum states are described by density operators. A density operator $\rho \in \Lc(\Hc)$ satisfies two properties. It is positive semidefinite (denoted $\rho \geq 0$) and has $\Tr[\rho] = 1$. POVM elements $M \in \Lc(\Hc)$ are positive semidefinite operators with all eigenvalues smaller than one. This means that $I-M$ is also positive semidefinite and a POVM therefore satisfies $0 \leq M \leq I$. A general POVM measurement is described by a colleaction of POVM elements $\{ M_1, ..., M_n  \}$ that satisfy $\sum_{i=1}^n M_i = I$. Denote the measurement outcome associated with $M_i$ as $m_i$. Then given a state $\rho$, the probability to observe outcome $m_i$ is $\Tr[M_i \rho]$. The Hermitian observable $E \in \Lc(\Hc)$ associated with this measurement is then $E = \sum_{i=1}^n m_i M_i$. Therefore the expectation value of the measurement, given the state $\rho$, is $\braket{E}_\rho = \Tr[E \rho]$. In this work, we will only consider two-valued measurements, with associated outcomes $\pm 1$. Such a measurement is thus described by the POVM measurement $M, I-M$ and the corresponding observable is $E = M - (I-M) = 2M - I$. 

Operations on quantum states that transform one state into the other are described by quantum channels. In general, transformations of linear operators $A \in \Lc(\Hc)$ are described by a linear operator $\Ec: \Lc(\Hc) \rightarrow \Lc(\Hc)$. These linear operators are sometimes called superoperators, to distinguish them from linear operators $A \in \Lc(\Hc)$. A quantum channel is a superoperator $\Ec: \Lc(\Hc) \rightarrow \Lc(\Hc)$ that is 
\begin{itemize}
	\item completely positive (CP), i.e., $(\Ec\otimes \Ic)(A) \geq 0$ for all $0 \leq A \in \Lc(\Hc^\tp{2})$, where $\Ic$ is the identity channel; and    
	\item trace preserving (TP), i.e., $\Tr[\Ec(A)] = \Tr[A]$ for all $A \in \Lc(\Hc)$.
\end{itemize} 
Intuitively, this means that density operators are mapped to density operators. Thus quantum channels (CPTP superoperators) are indeed the operators that map quantum states to quantum states. Here generic quantum channels are denoted $\Ec$ or $\Lambda$. A quantum channel is said to be unitary (denoted $\Gc$) if $\Gc(A) = GAG^\dagger$ for some unitary $G \in \Lc(\Hc)$ and for all $A \in \Lc(\Hc)$. So unitary quantum channels (also called unitaries or gates) are denoted with a calligraphic $\Gc$ and their counterparts in $\Lc(\Hc)$ are denoted $G$. Unital maps are superoperators $\Ec$ that satisfy $\Ec(I) = I$. Note that all unitaries are unital, but the converse is not true (consider the completely depolarizing channel $\Ec(A) = \frac{\Tr[A]}{d} I$).
The space of superoperators is typically equipped with the induced Schatten-norms, defined as
\begin{equation}
\label{eq:induced_schatten_norms}
\| \Ec \|_{p \rightarrow q} = \sup_{A \in \Lc(\Hc)} \{ \| \Ec(A) \|_q : \| A \|_p = 1   \}.
\end{equation}
Important special cases are $p = q = 1$, which yields the induced trace norm and $p = q = 2$ which results in the operator norm ($\| \Ec \|_\infty = \| \Ec \|_{2 \rightarrow 2}$). 
For more details on states, measurements and quantum channels, the reader is referred to text books like \cite{Nielsen2002, Watrous2017}. In the next section, we will discuss the Liouville representation of states, measurements, and quantum channels.

\subsubsection{Liouville representation}
\label{app_subsub:liouville_rep}
Here we expand on the definition of the Liouville representation (also known as the natural or affine representation or the Pauli transfer matrix) \cite{Wallman2014, Watrous2017} introduced in the main text. This representation exploits the fact that the Pauli matrices form an orthogonal basis for the set of linear operators with respect to the Hilbert-Schmidt inner product. We can then think of linear operators $A \in \Lc(\C^d)$ as column vectors or row vectors with entries determined by the inner product with respect to a Pauli basis operator. Formally, we introduce a linear map $\kket{\cdot}: \Lc(\C^d) \rightarrow \C^{d^2}$ defined by $\kket{\sigma_i} = e_i$, where $\sigma_i$ is the $i$-th normalized Pauli matrix in $\Pc$ and $e_i$ is the $i$-th canonical basis vector of $\C^{d^2}$. The map is then extended to $\Lc(\Hc)$ by linearity, so that
\begin{equation}
\kket{A} = \sum_{\sigma_i \in \Pc} \braket{\sigma_i, A}_\mathrm{HS} \kket{\sigma_i}.
\end{equation}
The adjoint is then defined via $\bbra{A} = \kket{A}^\dagger$. As a result, the inner product carries over as
\begin{equation}
\bbraket{A | B} = \braket{A, B}_\mathrm{HS} = \Tr[A^\dagger B], \quad \forall A,B \in \Lc(\C^d).
\end{equation}
Quantum channels $\Ec: \Lc(\C^d) \rightarrow \Lc(\C^d)$ can then be viewed as matrices acting on the vectors $\kket{A}$. This matrix, called the Liouville matrix, is a map $\pmb{\Ec}: \C^{d^2} \rightarrow \C^{d^2}$ defined by $\pmb{\Ec}_{ij} = \bbraket{\sigma_i | \Ec(\sigma_j)}$ (with $\sigma_i,\sigma_j \in \Pc$). The Liouville matrix $\pmb{\Ec}$ corresponding to the quantum channel $\Ec$ is denoted in bold font to distinguish the two. The Liouville matrix representation of quantum channels naturally respects the vectorization $\kket{\cdot}$, the product (channel composition is identified with matrix multiplication), the adjoint and the tensor product. That is, for superoperators $\Ec_1, \Ec_2: \Lc(\C^d) \rightarrow \Lc(\C^d)$ and linear operators $A, B, Q \in \Lc(\C^d)$, the following relations hold:
\begin{equation}
\begin{split}
&\kket{\Ec_2 \Ec_1(A)} = \pmb{\Ec_2} \kket{\Ec_1(A)} = \pmb{\Ec_2} \pmb{\Ec_1} \kket{A}, \\
&\kket{\Ec_2 \otimes \Ec_1(A \otimes B)} = \pmb{\Ec_2} \otimes \pmb{\Ec_1} \kket{A \otimes B} = \pmb{\Ec_2} \otimes \pmb{\Ec_1} \kket{A}\kket{B} = \pmb{\Ec_2} \kket{A} \otimes \pmb{\Ec_1} \kket{B}, \\
&\kket{\Ec_1^\dagger(A)} = \pmb{\Ec_1}^\dagger \kket{A},  \\
&\Tr[Q^\dagger \Ec_1(A)] = \bbraket{Q|\Ec_1(A)} = \bbraket{Q | \pmb{\Ec_1} | A}.
\end{split}
\end{equation}
Note that with slight Dirac-notation-like ambiguity, the (not necessarily Hermitian operator) $\pmb{\Ec_1}$ is always applied to the ket $\kket{A}$ and not to the bra $\bbra{Q}$ in the last line. 
A quantum channel has a special block form of its Liouville matrix by imposing the trace-preserving property. If the first basis element of $\Pc$ is $\sigma_0 = \frac{I}{\sqrt{d}}$, a quantum channel can be written as
\begin{equation}
\label{eq:Liouville_block_decomp}
\pmb{\Ec} = \begin{bmatrix}
1 & 0 \\
\alpha(\Ec) & \pmb{\Ec}_\urm
\end{bmatrix},
\end{equation}
where $\alpha(\Ec)$ is the nonunitality vector (of length $d^2-1$) and $\pmb{\Ec}_\urm$ is the unital block (of size $d^2-1$ by $d^2-1$) of $\pmb{\Ec}$. The trace-preserving property implies that no traceless Pauli matrix in $\Pc^*$ can be mapped to $\sigma_0$, since $\bbraket{\sigma_0 | \Ec(\tau)} = \frac{\Tr[\Ec(\tau)]}{\sqrt{d}} = 0$ for all $\tau \in \Pc$. Similarly $\bbraket{\sigma_0 | \Ec(\sigma_0)} = \frac{\Tr[\Ec(\sigma_0)]}{\sqrt{d}} = 1$. This justifies the first row of \eqref{eq:Liouville_block_decomp}. 
In terms of this decomposition, the definition of the unitarity \eqref{eq:unitary_definition2} can be rewritten as 
\begin{equation}
\label{eq:def_unitarity_app}
u(\Ec) = \frac{1}{d^2 - 1} \sum_{\sigma,\tau \in \Pc^*}\bbraket{\tau | \pmb{\Ec} | \sigma}^2 =\frac{1}{d^2 - 1} \sum_{\sigma \in \Pc^*}\bbraket{\sigma | \pmb{\Ec}_\urm^\dagger \pmb{\Ec}_\urm | \sigma}  =  \frac{1}{d^2 - 1} \sum_{\sigma \in \Pc^*}\bbraket{\sigma | \pmb{\Ec}_\urm \pmb{\Ec}_\urm^\dagger | \sigma} = \frac{\Tr[\pmb{\Ec}_\urm^\dagger \pmb{\Ec}_\urm]}{d^2-1}  = \frac{\Tr[\pmb{\Ec}_\urm \pmb{\Ec}_\urm^\dagger]}{d^2-1} ,
\end{equation}
where $\pmb{\Ec}_\urm$ is slight abuse of notation for $1 \oplus \pmb{\Ec}_\urm$.

\subsection{Representation theory}
\label{app_sub:rep_theory}
Here we give a brief overview of the required representation theory of finite groups. This section will briefly provide some definitions and the results used in this work. For more details the reader can refer to textbooks like \cite{Fulton2004,Etingof2009}. 
Let $\Gbb$ denote a finite group, $V$ some finite-dimensional complex vector space. Let $\GLc(V)$ denote the general linear group on $V$ (i.e., the set of invertible linear operators on $V$). Then a representation $(V,R)$ is a map $R: \Gbb \rightarrow \GLc(V)$ that satisfies $R(g) R(h) = R(gh)$ for all $g,h \in \Gbb$. If $V$ is equipped with an inner product (making it a Hilbert space) and $R(g)$ is unitary for all $g \in \Gbb$, then $(V,R)$ is called a unitary representation of $\Gbb$. If $R$ is an injective map, then the representation is faithful. If the map $R$ is clear from the context, the representation is just referred to as $V$.

A subspace $W \subseteq V$ is called a subrepresentation of $V$ if $R(g) W \subseteq W$ for all $g \in \Gbb$. If $W = 0$ and $W = V$ are the only subrepresentations of $V$, then $V$ is an irreducible representation (often called irrep). Consider two representations $V_1, V_2$ of $\Gbb$. Then a mapping $\varphi: V_1 \rightarrow V_2$ is called an intertwining operator if $\varphi R_1(g) = R_2(g) \varphi$. Intuitively, an intertwining operator preserves the structure of a representation. The representations $V_1$ and $V_2$ are called equivalent (denoted $V_1 \cong V_2$) if there exists an intertwining operator $\varphi$ that is an isomorphism between the vector spaces. A fundamental result in representation theory of finite groups is that a representation $(V,R)$ can always be written as the direct sum of irreps.
\begin{lemma}[Maschke's Theorem \cite{Fulton2004}]
	\label{lem:maschke}
	Let $(V, R)$ be a finite-dimensional, nonzero representation of a finite group $\Gbb$. Then $(V, R)$ decomposes uniquely (up to isomorphisms and ordering) as
	\begin{equation}
	V = \bigoplus_{i=1}^k  \left(\C^{n_i} \otimes V_i\right) = \bigoplus_{i=1}^k  V_i^{\oplus n_i} \hspace{20pt} \mbox{and} \hspace{20pt} R = \bigoplus_{i=1}^k \left(I_{n_i} \otimes R_i\right) = \bigoplus_{i=1}^k R_i^{\oplus n_i},
	\end{equation}
	where the set $\{(V_i, R_i) : i=1,...,k \}$ contains mutually inequivalent, nonzero, irreducible representations occurring with multiplicity $n_i$ in the decomposition of $(V, R)$ and $I_{n_i}$ is the identity on a $n_i$-dimensional vector space. 
\end{lemma}

As an example consider the Clifford group $\Gbb = \Cliff{d} \subset \Lc(\Hc)$. Then the map $R_1: G \mapsto \Gc$ that associates the quantum channel $\Gc$ with the abstract group element $G \in \Cliff{d}$ is a representation of $\Cliff{d}$ on the space $V_1 = \Lc(\Hc)$. In fact $G$ is itself a representation (the defining representation) on $\Hc$. The Liouville representation is also a representation on the space $V_2 = \C^{d^2}$ via the map $R_2: G \mapsto \Gcbf$. The Liouville representation $(V_2, R_2)$ and the quantum channel representation $(V_1,R_1)$ are equivalent representations of the Clifford group $\Cliff{d}$. The intertwining operator that establishes this equivalence is given by $\varphi = \kket{\cdot}: V_1 \rightarrow V_2$ (defined in Appendix~\ref{app_subsub:liouville_rep}), the map that sends a linear operator $A \in V_1$ to the corresponding Liouville vector $\kket{A} \in V_2$. The intertwining property $R_2 \varphi = \varphi R_1$ is then explicitly expressed as $\Gcbf \kket{A} = \kket{\Gc(A)}$ for all $A \in V_1$ and $G \in \Cliff{d}$.

A crucial ingredient to the URB protocol is constructing the projector onto the trivial subrepresentations of a representation $(V,R)$. This is achieved in the following result.
\begin{lemma}[Projection onto trivial subrepresentations \cite{Fulton2004}]
	\label{lem:projection_onto_trivial_subreps}
	Let $(V,R)$ be any representation of a group $\Gbb$ and let $V^\Gbb := \{ v \in V : R(g)v = v, \, \forall g \in \Gbb \}$ denote the subspace on which $\Gbb$ acts trivially. Define the map $\phi: V \rightarrow V$ by
	\begin{equation}
	\phi = \frac{1}{|\Gbb|} \sum_{g \in \Gbb} R(g).
	\end{equation}
	Then $\phi$ is an intertwining operator and moreover $\phi$ is the orthogonal projection onto $V^\Gbb$.
\end{lemma}

The next lemma is crucial for the variance analysis, as it provides a method to identify the subspace of trivial representations $(V \otimes V^*)^\Gbb$, given a decomposition of $V$ into irreps.
\begin{lemma}
	\label{lem:trivial_subrep_lemma}
	Let $(V,R_V)$ and $(W,R_W)$ be unitary, irreducible finite-dimensional representations of a finite-dimensional group $\Gbb$ and let $\{v_i\}$, $\{w_i\}$ be an orthonormal basis for $V$, $W$ respectively. If $V \cong W$ are equivalent representations (and the basis vectors are labeled such that the intertwining map $\varphi$ between $V$ and $W$ maps $v_i \mapsto w_i$), then the $(V\otimes W^*, R_{V \otimes W^*})$ has one and only one trivial subrepresentation
	\begin{equation}
	(V\otimes W^*)^\Gbb = \Span \left\{  \sum_i v_i \otimes w_i^\dagger \right\}.
	\end{equation}
	 If $V$ and $W$ are not equivalent, then
	\begin{equation}
	(V\otimes W^*)^\Gbb = \emptyset.
	\end{equation}
\end{lemma}
\begin{proof}
	The proof makes use of the canonical isomorphism $\alpha: V \otimes W^* \rightarrow \Lc(W,V)$ defined by $v \otimes w^\dagger \mapsto vw^\dagger$ (extended by linearity), where $V^*$ is the dual space of $V$ (carrying the dual representation) and $vw^\dagger$ acts on $x\in W$ by $vw^\dagger x := v \braket{w,x}$ (with $\braket{\cdot , \cdot}$ the inner product on $W$). Now $\alpha$ is an intertwining operator \cite{Fulton2004}. Therefore it follows that
	\begin{equation}
	\alpha\left( (V\otimes W^*)^\Gbb \right) = \left(\Lc(W,V)\right)^\Gbb,
	\end{equation}
	since $\alpha$ preserve the structure of the representation. The subspace $\left(\Lc(W,V)\right)^\Gbb$ of trivial subrepresentations of $\Lc(W,V)$ is precisely the space of intertwining operators between the representations $W$ and $V$ \cite{Fulton2004}.
	Thus a trivial representation of $V \otimes W^*$ corresponds to an intertwining operator from $W$ to $V$. Schur's Lemma states that \cite{Fulton2004}
	\begin{equation}
	| \left(\Lc(W,V)\right)^\Gbb | = \begin{cases}
	1 & \mbox{if } V \cong W \\
	0 & \mbox{otherwise.}
	\end{cases}
	\end{equation}
	So if $V$ and $W$ are inequivalent $\alpha\left( (V\otimes W^*)^\Gbb \right) = \left(\Lc(W,V)\right)^\Gbb = \emptyset$. And if $V \cong W$, let $\phi \in \left(\Lc(W,V)\right)^\Gbb$ be the intertwining isomorphism with $\| \phi \|_\infty = 1$. Then letting $v_i = \phi(w_i)$, we can write $\phi = \sum_i v_i w_i^\dagger$, so that
	\begin{equation}
	\alpha\left( (V\otimes W^*)^\Gbb \right) = \Span \left\{\sum_i v_i w_i^\dagger \right\},
	\end{equation}
	which yields the result after applying $\alpha^{-1}$.
\end{proof}
\begin{corollary}
	If moreover the representation $V = W$ is real and thus orthogonal, then (using $V^* \cong V$) it follows that
	\begin{equation}
	(V\otimes V)^\Gbb = \Span \left\{  \sum_i v_i \otimes v_i \right\}.
	\end{equation}
\end{corollary}
\begin{corollary}
	Let $V$ be a finite-dimensional vector space carrying a group representation. By \autoref{lem:maschke} there exists a decomposition $V = \bigoplus_{i=1}^k  V_i^{\oplus n_i}$ into mutually inequivalent irreducible representations. Denote $V_{i_s}$ the $s$-th copy of the space $V_i$ ($s=1,...,n_i$) and denote $\{v_j^{(i_s)} : j = 1,..., |V_{i}| \}$ an orthonormal basis of $V_{i_s}$ that respect the isomorphisms between equivalent spaces (meaning that $v_j^{(i_s)} \mapsto v_j^{(i_{s'})}$ under the intertwining isomorphism between $V_{i_s}$ and $V_{i_{s'}}$). Then the trivial subrepresentations of $V \otimes V$ are given by
	\begin{equation}
	(V \otimes V)^\Gbb = \Span\left\{\sum_{j=1}^{|V_i|}  v_j^{(i_s)} \otimes v_j^{(i_{s'})} \bigg| \, \forall s,s' = 1,..., n_i, \, \forall i=1,...,k \right\}.
	\end{equation}
\end{corollary}
\begin{proof}
	Let us start by writing
	\begin{equation}
	V \otimes V = \bigoplus_{i,i'=1}^k \bigoplus_{s}^{n_i} \bigoplus_{s'}^{n_{i'}}  (V_{i_s} \otimes V_{{i'}_{s'}}).
	\end{equation}
	Each trivial subrepresentation is found by application of \autoref{lem:trivial_subrep_lemma} to each term in this decomposition. This makes use of the fact that $V_{i_s} \cong V_{i'_{s'}}$ are equivalent if and only if $i'= i$ by virtue of the decomposition.
\end{proof}

In Appendix~\ref{app:variance_bound} this machinery is used to find the trivial subrepresentations of the Liouville tensor-4 representation of the Clifford group $\Cliff{d}$. But first a section is given with some preliminary technical lemmas from literature that are required in the proof of our variance bound.

\subsection{Technical lemmas from literature}
\label{app_sub:lemma_literature}
In this section we review a few lemmas from literature that are required for our variance bound. Some lemmas are stated without proof and the reader is then referred to the reference for a proof. The first lemma is a telescoping series for expanding the variance expression. It is applied to quantum channels, but here presented in more general form.
\begin{lemma}[Telescoping Series \cite{Helsen2017}]
	\label{lem:telescoping}
	Let $A$ be an associative algebra with unit. Then for $a, b \in A$ and $m \in \mathbb{N}_+$,
	\begin{equation}
	a^m - b^m = \sum_{s=1}^{m} a^{m-s} (a-b) b^{s-1} = \sum_{s=1}^{m} b^{m-s} (a-b) a^{s-1}.
	\end{equation}
\end{lemma}
\begin{proof}
	By direct computation, it follows that
	\begin{equation*}
	\sum_{s=1}^{m} a^{m-s} (a-b) b^{s-1} = \sum_{s=1}^{m} a^{m-s+1} b^{s-1} - a^{m-s} b^s = \sum_{s=0}^{m-1} a^{m-s} b^{s} - \sum_{s=1}^{m} a^{m-s} b^s = a^m b^0 - a^0 b^m = a^m - b^m
	\end{equation*}
	and
	\begin{equation*}
	\sum_{s=1}^{m} b^{s-1} (a-b) a^{m-s} = \sum_{s=1}^{m} b^{s-1}a^{m-s+1} - b^{s}a^{m-1} = \sum_{s=0}^{m-1} b^{s} a^{m-s} - \sum_{s=1}^{m} b^{s} a^{m-1} = b^0a^m - a^0 b^m  = a^m - b^m. \qedhere
	\end{equation*}
\end{proof}
Note that the set of quantum channels form an associative algebra with unit, so that this lemma indeed applies to quantum channels.

Next we present a lemma that bounds the induced schatten $p\rightarrow p$ norm of a quantum channel.
\begin{lemma}[Perez-Garcia \textit{et al}. \cite{Perez-Garcia2006}]
	\label{lem:Perez-Garcia}
	Let $\Ec$ be a CPTP quantum channel on a $d$-dimensional Hilbert space $\Hc$, with $d = 2^q$ for a $q$-qubit system. Then for all $p \in [1, \infty]$,
	\begin{equation}
	\| \Ec \|_{p\rightarrow p} = \max_{A\in\Lc(\Hc)} \left\{\| \Ec(A)  \|_p : \|A\|_p = 1 \right\} \leq d^{1-\frac{1}{p}}
	\end{equation}
	and
	\begin{equation}
	\label{eq:A25}
	\| \Ec \|_{p\rightarrow p}^H:= \max_{A\in\Lc(\Hc)} \left\{\| \Ec(A)  \|_p : \|A\|_p = 1, \Tr[A] = 0, A = A^\dagger \right\}  \leq \left(\frac{d}{2}\right)^{1-\frac{1}{p}}. 
	\end{equation}
%	\begin{equation}
%	\label{eq:A2xx5}
%	\| \Ec \|_{p\rightarrow p}^H:= \max_{\substack{0 \neq A\in\Lc(\Hc) \\ \Tr[A] = 0, A = A^\dagger}} \left\{   \frac{\| \Ec(A)  \|_p}{\|A\|_p}  \right\}  \leq \left(\frac{d}{2}\right)^{1-\frac{1}{p}}. 
%	\end{equation}
	If in addition $\Ec$ is unital ($\Ec(I) = I$), then $\| \Ec \|_{p \rightarrow p} \leq 1$ for all $p \in [1,\infty]$.
\end{lemma}

The following three lemmas are used to bound the quantities $a_i$ (\eqref{eq:expansion}). First, we state a technical lemma used in \cite{Helsen2017}, which can be restated as
\begin{lemma}[Helsen \textit{et al}. \cite{Helsen2017}]
	\label{lem:diagonals_squared}
	Let $\Ec$ be a CPTP map on a $d$-dimensional Hilbert space. Then
	\begin{equation}
	0 \leq \frac{1}{d^2 - 1} \sum_{\sigma \in \Pc^*} \bbraket{\sigma| \pmb{\Ec} |\sigma }^2 - f^2  \leq  \frac{d^2-2}{d^2} (1-f)^2,
	\end{equation}
	where
	\begin{equation}
	f = \frac{1}{d^2 - 1} \sum_{\sigma \in \Pc^*} \bbraket{\sigma| \pmb{\Ec} |\sigma }
	\end{equation}
	is the randomized benchmarking decay parameter of $\Ec$.
\end{lemma}
Here this lemma is applied to channels of the form
\begin{equation}
\label{eq:unital_part}
\pmb{\Ec_1} = \begin{bmatrix}
1 & 0 \\
0 & \pmb{\Lambda}_\urm \pmb{\Lambda}_\urm^\dagger,
\end{bmatrix}
\quad\quad\mbox{and}\quad\quad
\pmb{\Ec_2} = \begin{bmatrix}
1 & 0 \\
0 & \pmb{\Lambda}_\urm^\dagger \pmb{\Lambda}_\urm,
\end{bmatrix}
\end{equation}
where $\Lambdabf_\urm$ is the unital block of the error map $\pmb{\Lambda}$ under investigation, since then $f(\Ec_1) = f(\Ec_2) = u(\Lambda)$. It is not clear that these superoperators are even a quantum channel (in particular, that they are CPTP). Therefore the following lemma provides a necessary condition on $\Lambda$ for which \eqref{eq:unital_part} are CPTP maps.
\begin{lemma}
	\label{lem:unital_part}
	Let $\Lambda$ be a CPTP quantum channel on a $d$-dimensional Hilbert space. Then the channels $\Ec_1$, $\Ec_2$ defined in \eqref{eq:unital_part} are CPTP if either $d = 2$ or if $\Lambda$ is unital (or both). Moreover $ \| \Ec_1 \|_{2 \rightarrow 2}, \| \Ec_2 \|_{2 \rightarrow 2} \leq 1$.
\end{lemma}
\begin{proof}
	If $d = 2$ (that is, if $\Lambda$ is a single-qubit channel), then the unital part of $\Lambda$, defined as 
	\begin{equation}
	\pmb{\hat{\Lambda}} = \begin{bmatrix}
	1 & 0 \\
	0 & \pmb{\Lambda}_\urm
	\end{bmatrix},
	\end{equation}
	is CPTP \cite[Theorem IV.1]{Braun2014}. For the general $d$-dimensional case, it is assumed that $\Lambda$ is unital, so that $\Lambda = \hat{\Lambda}$. So in either case, $\hat{\Lambda}$ is CPTP and unital. It can be shown that the adjoint of a CPTP and unital map is also CPTP and unital \cite[Proposition 2.18 and Theorem 2.26]{Watrous2017}. This means that $\hat{\Lambda}^\dagger$ is CPTP and unital. Therefore $\Ec_1 = \hat{\Lambda}\hat{\Lambda}^\dagger$ and $\Ec_2 = \hat{\Lambda}^\dagger\hat{\Lambda}$ are also CPTP and unital. 
	\autoref{lem:Perez-Garcia} then ensures that $\| \Ec_1 \|_{2 \rightarrow 2} \leq 1$ and $\| \Ec_2 \|_{2 \rightarrow 2} \leq 1$.
\end{proof}
Third is a lemma from matrix analysis. It is a characterization of positive semidefinite matrices in terms of its principal minors. This lemma was used on $\mathbf{I} - \pmb{\hat{\Lambda}}\pmb{\hat{\Lambda}}^\dagger$ to bound its off-diagonal terms.
\begin{lemma}[Sylvester's criterion]
	\label{lem:Sylvester}
	Let $A \in \Lc(\C^{d^2})$ be a Hermitian matrix. Then $A$ is positive semidefinite if and only if all of its principal minors are nonnegative.
\end{lemma}
\begin{proof}
	See \cite[Corollary 7.1.5 and Theorem 7.2.5]{Horn2013}
\end{proof}

Next we present two results, also from matrix analysis, that are used several times to bound inner products. The first is a trace inequality and the second is H\"older's inequality.
\begin{lemma}
	\label{lem:NeumannHolder}
	Let $A, B \in \Lc(\Hc)$ be two linear operators on a $d$-dimensional Hilbert space $\Hc$. Denote their singular values as $s_i(A), s_i(B)$ respectively with $i = 1,..., d$, both in decreasing order. Finally let $s(A)$ and $s(B)$ denote vectors with entries $s_i(A)$ and $s_i(B)$. Then
	\begin{enumerate}
		\item (Von Neumann's trace inequality) $\mathrm{Re}(\Tr[A B]) \leq \sum_{i=1}^d s_i(A) s_i(B)$, and 
		\item (H\"older's inequality) $\sum_{i=1}^d | s_i(A) s_i(B) | \leq \| s(A) \|_p \| s(B) \|_q = \| A \|_p \| B \|_q$, for any pair $p,q \in [1, \infty]$ such that $p^{-1} + q^{-1} = 1$.
	\end{enumerate}
	Since singular values are nonnegative, combining the statements yields $\mathrm{Re}(\Tr[A B]) \leq \| A \|_p \| B \|_q$ for any pair $p,q \in [1, \infty]$ such that $p^{-1} + q^{-1} = 1$.
\end{lemma}
\begin{proof}
	Statement 1 is proven for example in \cite[Theorem 8.7.6]{Horn2013} and statement 2 is proven in \cite[Theorem 31.3]{Aliprantis1998}.
\end{proof}
\begin{corollary}
	If $A, B \in \Lc(\Hc)$ are Hermitian, then $\Tr[AB]^* = \Tr[(AB)^\dagger] = \Tr[B^\dagger A^\dagger] = \Tr[BA] = \Tr[AB]$, so that $\Tr[AB]$ is real. Therefore $\Tr[AB] \leq \| A \|_p \| B \|_q$ for any $p,q \in [1,\infty]$ satisfying $p^{-1} + q^{-1} = 1$.
\end{corollary}

Finally some of our bounds use the fact that the mean of squares is larger than the square of the mean. We show this well-known fact below.
\begin{lemma}[Mean of squares is larger than square of mean]
	\label{lem:meansquares}
	Let $\{x_i\} \subset \R$ be a collection of $N$ real numbers. Then 
	\begin{equation}
	\left(\frac{1}{N} \sum_{i=1}^N x_i\right)^2 \leq \frac{1}{N} \sum_{i=1}^N x_i^2.
	\end{equation}
\end{lemma}
\begin{proof} By direct computation, it follows that
	\begin{equation}
	\begin{split}
	\left(\frac{1}{N}\sum_{i=1}^N x_i^2\right) - \left(\frac{1}{N}\sum_{i=1}^N x_i \right)^2
	&= \left(\frac{1}{N}\sum_{i=1}^N x_i^2\right) - 2\left(\frac{1}{N}\sum_{i=1}^N x_i \right)\left(\frac{1}{N}\sum_{k=1}^N x_k \right) + \left(\frac{1}{N}\sum_{k=1}^N x_k \right)^2\\
	&=\frac{1}{N}\sum_{i=1}^N \left(x_i^2 - 2 x_i \left(\frac{1}{N}\sum_{k=1}^N x_k \right) + \left(\frac{1}{N}\sum_{k=1}^N x_k \right)^2 \right) \\
	&= \frac{1}{N}\sum_{i=1}^N\left( x_i - \left(\frac{1}{N}\sum_{k=1}^N x_k\right) \right)^2  \geq 0,\\
	\end{split}  
	\end{equation}
	since it is the sum of real numbers squared, proving the result.
\end{proof}

\section{Variance bound and interval length bound}
\label{app:variance_bound}
This section is devoted to rigorously proving the variance bound \eqref{eq:var_bound}. Along the way we also prove the interval length bound \eqref{eq:interval_length}. The key ingredient of the variance bound proof is finding the trivial subrepresentations of the Liouville tensor-4 representation of the Clifford group $\Cliff{d}$. This is done in the first subsection. Then the variance bound \eqref{eq:var_bound} is proven. The technical lemmas used in this proof are collected in Appendix~\ref{app:technical_lemmas}. 

\subsection{Trivial subrepresentations of the tensor-4 Liouville representation of the Clifford group}
\label{app_sub:trivial_subreps_Liouville_4}
This section is concerned with presenting the trivial subrepresentations of the representation $G \mapsto \Gc^\tp{4}$ of the Clifford group $\Cliff{d}$. This representation is equivalent to $G \mapsto \Gcbf^\tp{4}$ by the intertwining isomorphism $\kket{\cdot}$. Therefore both are considered the same and with slight abuse of notation we refer to them both as the same representation, which we will call the tensor-4 Liouville representation. 

The key idea is to apply \autoref{lem:trivial_subrep_lemma} and its corollaries to find the trivial subrepresentations of the tensor-4 representation $G \mapsto \Gc^\tp{4}$. This requires a full description of the Liouville tensor-2 representation $G \mapsto \Gc^\tp{2}$ in terms of its irreducible components. This was studied in \cite{Zhu2017,Helsen2018}. Let us denote $V = \Lc(\Hc \otimes \Hc)$ as the space that carries the tensor-2 representation. The present problem is therefore to find the trivial subrepresentations of $V \otimes V$, given a decomposition of $V$ into irreducible representations.

In an earlier result \cite{Zhu2016} the multiplicity of the trivial representation in $V \otimes V$ was calculated. They found that
\begin{equation}
| (V \otimes V)^\Cliff{d} | =  
\begin{cases}
15 & \mbox{if } d=2; \\
29 & \mbox{if } d=4; \\
30 & \mbox{otherwise}, 
\end{cases}
\end{equation}
which is a justification of \eqref{eq:dim_rge_N}. First we will discuss the decomposition of $V$ into irreducible representations \cite{Helsen2018}, and next we will apply \autoref{lem:trivial_subrep_lemma} to find $(V \otimes V)^\Cliff{d}$ explicitly.

The full decomposition of the Liouville tensor-2 representation $(V, R)$ given by $R: \Cliff{d} \rightarrow \GLc(V) : G \mapsto \Gc^\tp{2}$ is studied in \cite{Helsen2018}. We will review the result of this work here, following their notation. A summary of the relevant subspaces is given in \autoref{fig:V_TS}. First, the representation $V$ is decomposed in the following subrepresentations, defined by
\begin{equation}
\label{eq:bases}
\begin{split}
V_S &:= \Span \{  \frac{\sigma \tau + \tau \sigma}{\sqrt{2}} :  \sigma, \tau \in \Pc^*, \sigma \neq \tau \}, \\
V_A &:= \Span \{  \frac{\sigma \tau - \tau \sigma}{\sqrt{2}} :  \sigma, \tau \in \Pc^*, \sigma \neq \tau \}, \\
V_d &:= \Span \{  \sigma \sigma :  \sigma\in \Pc^* \}, \\
V_{r,l} &:= \Span \{ \sigma_0 \sigma, \sigma \sigma_0 : \sigma \in \Pc^*  \}, \\
V_{\ideal} &:= \Span \{  B_1 = \sigma_0 \sigma_0 \}. \\
\end{split}
\end{equation}
Recall that the tensor symbol is omitted for brevity (so $\sigma\tau$ means $\sigma \otimes \tau$ here). Each of these spaces carries a subrepresentation and furthermore $V = V_\ideal \oplus V_{r,l} \oplus V_d \oplus V_S \oplus V_A$.
Finally let us define the traceless, symmetric subspace as 
\begin{equation}
\label{eq:V_TS}
V_{TS} := V_S \oplus V_d.
\end{equation}
Since the ideal input and measurement operators for the URB protocol $\bar{\rho}_\ideal, \bar{E}_\ideal$ (as defined in \eqref{eq:ideal_operators}, see also \eqref{eq:ideal_operators_app}) are elements of $V_{TS}$ and since $\Lambda^\tp{2}(V_{TS}) \subseteq V_{TS}$ by the trace-preserving property of $\Lambda$ and the symmetry with respect to swapping the two copies of $\Hc$, the only relevant subspace of $V$ is $V_{TS}$. Therefore we continue our analysis of $V_{TS}$.

The space $V_d$ can be broken up into the two subrepresentations
\begin{equation}
\label{eq:bases2}
V_0 := \Span \left\{ B_2 = \frac{1}{d^2-1} \sum_{\sigma \in \Pc^*} \sigma\sigma \right\}  \quad\quad\mbox{and}\quad\quad V_{1,2} := V_d \setminus V_0. \\
\end{equation}
In the single-qubit case ($q=1$), the spaces $V_S$ and $V_{1,2}$ are irreducible, therefore fully characterizing $V_{TS} = V_0 \oplus V_{1,2} \oplus V_S$. However, if $q \geq 2$ the space $V_{1,2}$ breaks into two irreps, indexed by the index set $\Zc_{1,2}$. For $q = 2$, $V_S$ breaks into four irreps, while for $q \geq 3$ it breaks into five irreps, which will be indexed by $\Zc_S$. So the space $V_{TS}$ breaks up into the following number of irreps
\begin{equation}
| \Zc_{TS} | = 
\begin{cases}
3 & \mbox{if } q=1; \\
7 & \mbox{if } q=2; \\
8 & \mbox{if } q\geq 3,
\end{cases}
\end{equation}
where $\Zc_{TS} := \Zc_{d} \cup \Zc_S = \{ 0 \} \cup \Zc_{1,2} \cup \Zc_S$.  A summary of all the subspaces of $V_{TS}$ that carry subrepresentations is given in \autoref{fig:V_TS}, together with the dimensions of the spaces. In \cite{Helsen2018} it is shown that all irreducible representations contained in $V_{TS} = V_d \oplus V_S$ indexed by $\Zc_{TS}$ are mutually inequivalent. Therefore it follows from \autoref{lem:trivial_subrep_lemma} that there are precisely $|\Zc_{TS}|$ trivial subrepresentations contained in $V_{TS} \otimes V_{TS}$. The lemma also provides an explicit method of finding them, given a basis for $V_i$ from \cite{Helsen2018}. 

\begin{figure}
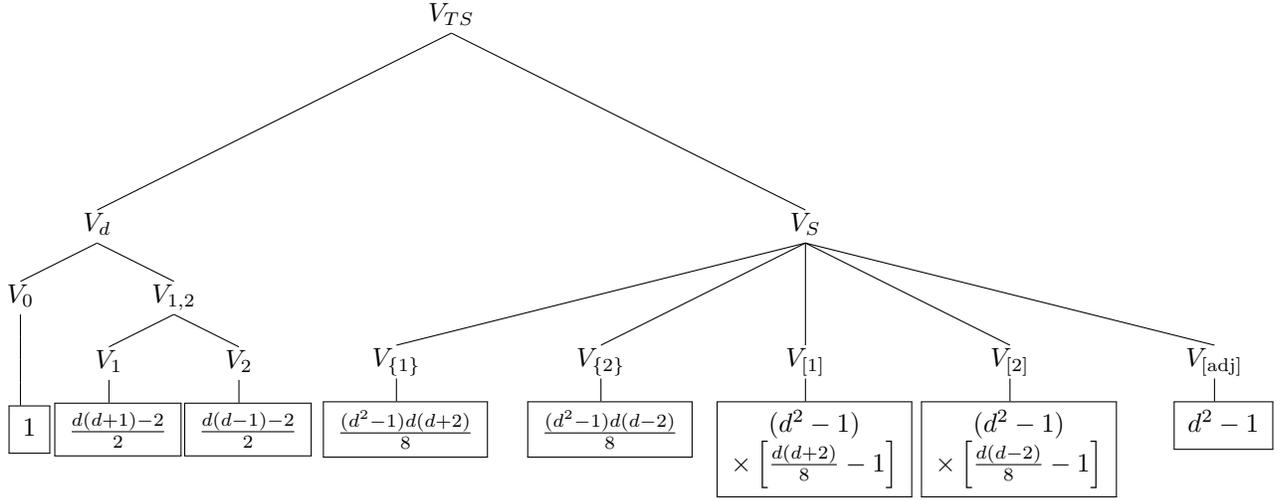

	\qtreecenterfalse
	\Tree[
		.$V_{TS}$ 
			[.$V_d$ 
				[.$V_0$  [ [ [ $1$ !{\qframesubtree} ]]]] !\qsetw{0.9cm}
				[.$V_{1,2}$ 
					[.$V_1$ $\frac{d(d+1)-2}{2}$ !{\qframesubtree} ] !\qsetw{2cm}
					[.$V_2$ $\frac{d(d-1)-2}{2}$ !{\qframesubtree} ]
		]] 
		[.$V_S$ !\qsetw{6cm}
			[.$V_{\{1\}}$ $\frac{(d^2-1)d(d+2)}{8}$ !{\qframesubtree} ]
			[.$V_{\{2\}}$ $\frac{(d^2-1)d(d-2)}{8}$ !{\qframesubtree} ]
			[.$V_{[1]}$ {$(d^2-1)$\\ $\times\left[\frac{d(d+2)}{8}-1\right]$} !{\qframesubtree} ] !\qsetw{3cm}
			[.$V_{[2]}$ {$(d^2-1)$\\ $\times\left[\frac{d(d-2)}{8}-1\right]$} !{\qframesubtree} ] 
			[.$V_{\adj}$ $d^2-1$ !{\qframesubtree} ]]]
	\caption{Hierarchy of subspaces contained within the traceless, symmetric subspace $V_{TS}$, carrying the relevant subrepresentation of the Liouville tensor-4 representation $G \mapsto \Gcbf^\tp{4}$. Every child node is a subspace (that also carries a subrepresentation) of its parent node and all child nodes direct sum to their parent. Leaf nodes represent the final irreducible subspaces and their dimension are shown in the box below each leave node. 	
	Definitions of the composite spaces are given in the main text (\eqref{eq:bases}, \eqref{eq:V_TS} and \eqref{eq:bases2}; for the definitions of the irreducible spaces, see \cite{Helsen2018}). 
	For $d=2$ or $d=4$ there are certain subspaces with $|V_i| \leq 0$. This means that such a subspace is empty and therefore not present in the decomposition. Summing the dimensions of the child nodes together, yields the following sizes for the decomposable spaces: $|V_{1,2}| = d^2-2$, $|V_{d}| = d^2-1$, $|V_{S}| = \half(d^2-1)(d^2-2)$ and $|V_{TS}| = \half d^2(d^2-1)$.}
	\label{fig:V_TS}
\end{figure}

Let $\Bc_i$ denote an orthonormal basis for $V_i$, for $i \in \Zc_{TS}$. Then since all irreps indexed by $\Zc_{TS}$ are mutually inequivalent, \autoref{lem:trivial_subrep_lemma} gives an explicit way to compute the trivial subreps of $(V_{TS} \otimes V_{TS})$ as
\begin{equation}
\label{eq:A_i}
A_i = \frac{1}{\sqrt{|V_i|}}\sum_{v_i \in \Bc_i} v_i  v_i, \quad \forall i \in \Zc_{TS},
\end{equation}
where the normalization constant is to normalize $A_i$ with respect to the Hilbert-Schmidt norm $\| A_i \|_2 = 1$. In the multiqubit case where $V_{1,2}$ and $V_S$ are not irreducible, it is still useful to define
\begin{equation}
\label{eq:A_j_not_irrep}
A_{j} = \frac{1}{\sqrt{|V_j|}} \sum_{i \in \Zc_j} \sqrt{|V_i|} A_i, \quad \quad j \in \{ S; \, d; \, 1,2 \}.
\end{equation}
In fact, this allows us to explicitly find $A_{1,2}$ from $A_d$ and $A_0$. Using the basis for $V_0$, $V_d$ and $V_S$ (in \eqref{eq:bases} and \eqref{eq:bases2}), we therefore explicitly find
\begin{align}
\label{eq:A0}
A_0 &:= B_2 B_2 = \frac{1}{d^2 - 1} \sum_{\sigma,\tau \in \Pc^*} \sigma\sigma\tau\tau, \\
\label{eq:A12}
A_{1,2} &:= \frac{1}{\sqrt{d^2-2}} \left( \sum_{\sigma \in \Pc^*} \sigma^\tp{4} - A_0  \right), \\
\label{eq:AS}
A_S &:= \sqrt{\frac{1}{2 (d^2-1)(d^2-2)}}  \sum_{\substack{\sigma, \tau  \in \Pc^* \\ \sigma \neq \tau}} \sigma\tau\sigma\tau + \sigma\tau\tau\sigma.
\end{align}
No explicit expression is needed for any $i \in \Zc_S$ or $i \in \Zc_{1,2}$ if $V_S$ and $V_{1,2}$ are reducible (which happens in the multiqubit case), because bounds are defined in terms of $A_S$ and $A_{1,2}$. The only exception to this is $i = \adj \in \Zc_S$. The space $V_\adj \subset V_S$, which carries an irrep, is defined by \cite{Helsen2018}
\begin{equation}
V_\adj = \Span \left\{  v_\tau^\adj = \frac{1}{2 \sqrt{|C_\tau|}} \sum_{\sigma \in C_\tau}  \sigma(\sigma \cdot \tau) + (\sigma \cdot \tau)\sigma    \Big| \tau \in \Pc^*  \right\},
\end{equation}
where $\cdot$ indicates the normalized matrix product and where $C_\tau$ is the set of all elements of $\Pc^*$ that commute with $\tau$ as defined in \eqref{eq:def_Ctau}.
The corresponding trivial subrepresentation, as computed using \eqref{eq:A_i}, is
\begin{equation}
\label{eq:Aadj}
A_\adj = \frac{1}{2(d^2-4)\sqrt{d^2 - 1}} \sum_{\tau \in \Pc^*} \left(\sum_{\sigma \in C_\tau} (\sigma \cdot \tau) \sigma  + \sigma (\sigma \cdot \tau)\right)^\tp{2}.
\end{equation}
In the next section, we use the trivial subrepresentations of the Liouville tensor-4 representation to prove our variance bound.

\subsection{Statement and proof of the variance bound and interval length bound}
\label{app_sub:statistical_results_proof}
In this section we will state and prove our main theorem on the variance bound and prove the interval in which the average sequence purity is found. We also show the optimality of the ideal input and measurement operators. First, we will recapture some of the most important definitions and results discussed in the main text. The point of departure is the expression for the variance of \eqref{eq:var_expr1} 
\begin{equation}
\label{eq:var_expr1_appendix}
\Vbb[q_\jbf] = \bbraket{\bar{E}^\tp{2} | \Ncbf^{m-1} - (\Mcbf^\tp{2})^{m-1} | \bar{\rho}^\tp{2}},
\end{equation}
where the operators are defined as
\begin{equation}
\label{eq:N_def}
{\Mc} := \Gavg{2} {\Lambda}^{\otimes 2} \Gavg{2}, \quad\quad
{\Nc} := \Gavg{4} {\Lambda}^{\otimes 4} \Gavg{4}, \quad\quad
\Gavg{n} := \frac{1}{|\Cliff{d}|} \sum_{\Gc \in \Cliff{d}} \Gc^\tp{n}.
\end{equation}
Here $q_\jbf$ is the sequence purity due to the sequence $\jbf$. As discussed in the main text, $\Mcbf$ only has support on the space $W = \Span\{ B_1, B_2 \} \subset \Lc(\Hc \otimes \Hc)$ \cite{Wallman2015a}, where 
\begin{align}
\label{def:def_B1}
B_1 &= \frac{I}{d} = \sigma_0  \sigma_0,   \\
\label{def:def_B2}
B_2 &= \frac{S - B_1}{\sqrt{d^2 -1}}  = \frac{1}{\sqrt{d^2 -1}} \sum_{\sigma \in \Pc^*} \sigma  \sigma.
\end{align}
In particular the matrix elements of $\Mcbf$ with respect to this basis (see also \eqref{eq:M_entries}) as
\begin{equation}
\label{eq:M_entries_app}
\pmb{\Mc} = \begin{bmatrix}
1 & 0 \\
\frac{\|\alpha(\Lambda)\|^2}{\sqrt{d^2-1}} & u(\Lambda) \\
\end{bmatrix}.
\end{equation}
From this it follows that (see also \eqref{eq:def_unitarity_app})
\begin{equation}
\label{eq:unitarity_eigenvalue}
\Mcbf \kket{B_2} = u \kket{B_2},
\end{equation}
which implies that $\bbraket{B_2 | \Mcbf | B_2} = \bbraket{B_2 | \Lambdabf^\tp{2} | B_2} = u$, since $\Gavgbf{2} \kket{B_2} = \kket{B_2}$ and $B_2$ is normalized. This is used in the analysis of \eqref{eq:var_expr1_appendix}.

In \eqref{eq:var_expr1_appendix} the measurement $E$ is replaced with its the traceless counterpart $\bar{E}$, which is defined as
\begin{equation}
\bar{E} := E - \frac{\Tr[E]}{d^2} I = E - \bbraket{B_1 | E} B_1.
\end{equation}
Since $\bar{\rho}$ is traceless by construction and $\Gc_\jbf$ is trace-preserving, it follows that $q_\jbf = \bbraket{E | \Gc_\jbf^\tp{2} | \bar{\rho}} = \bbraket{\bar{E} | \Gc_\jbf^\tp{2} | \bar{\rho}}$. This justifies the replacement of $E$ by $\bar{E}$ is all expectation value and variance expressions. In our analysis it is advantageous to think of $\bar{E}$ instead of $E$, since then
$\bar{E}_\ideal, \bar{\rho}_\ideal \propto B_2$. The ideal state and measurement operators were defined in \eqref{eq:ideal_operators}. For completeness, they are
\begin{align}
\label{eq:ideal_operators_appendix}
E_\ideal &= S =  B_1 + \sqrt{d^2-1} B_2, \\
\rho_\ideal &= \frac{I + S}{d(d+1)} = \frac{1}{d} B_1 + \frac{\sqrt{d^2-1}}{d(d+1)} B_2, \\
\hat{\rho}_\ideal &= \frac{I - S}{d(d-1)} = \frac{1}{d} B_1 - \frac{\sqrt{d^2-1}}{d(d-1)} B_2, 
\end{align}
from which it follows that
\begin{equation}
\label{eq:ideal_operators_app}
\bar{E}_\ideal = \sqrt{d^2-1} B_2 \quad\quad \mbox{and}\quad\quad \bar{\rho}_\ideal = \frac{\rho_\ideal - \hat{\rho}_\ideal}{2}= \frac{1}{\sqrt{d^2-1}} B_2.
\end{equation}
The implemented operators $\bar{\rho}$ and $E$ can then be decomposed into an ideal part and an error part as
\begin{align}
\label{eq:alpha_app}
\alpha &:= \frac{\bbraket{\bar{\rho}_\ideal | \bar{\rho}}}{\bbraket{\bar{\rho}_\ideal|\bar{\rho}_\ideal}} = (d^2-1) \bbraket{\bar{\rho}_\ideal | \bar{\rho}} , &  \bar{\rho}_\spam &:= \bar{\rho} - \alpha \bar{\rho}_\ideal, \\
\label{eq:beta_app}
\beta &:= \frac{\bbraket{\bar{E}_\ideal | \bar{E} }}{\bbraket{\bar{E}_\ideal|\bar{E}_\ideal}} = \frac{1}{d^2-1} \bbraket{\bar{E}_\ideal | \bar{E} }, &  \bar{E}_\spam &:= \bar{E} - \beta \bar{E}_\ideal.
\end{align}
This decomposition is chosen such that $\Tr[\bar{\rho}_\ideal \bar{\rho}_\spam] = \Tr[\bar{E}_\ideal \bar{E}_\spam] = 0$. It can be shown that the ideal operators $\bar{\rho}_\ideal$, $\bar{E}_\ideal$ are in fact ideal, in the sense that they maximize the prefactor $B$ in the fit model $\Ebb[q_\jbf] = B u^{m-1}$ (and also minimize the variance as we will see). The prefactor $B$ is given by (see \eqref{eq:fitmodel})
\begin{equation}
\label{eq:B21}
B = \bbraket{E | \Gavgbf{2} | \bar{\rho}} = \bbraket{\bar{E} | B_2} \bbraket{B_2 | \bar{\rho}} = \alpha \beta.
\end{equation}
The ideal operators $\bar{\rho}_\ideal$, $\bar{E}_\ideal$ will yield $B = 1$. The following lemma shows that this is in fact optimal.
\begin{lemma}[Optimality of ideal operators]
	\label{lem:optimality_operators}
	The prefactor $B$ in the fit model for URB as given in \eqref{eq:B21} satisfies $| B | \leq 1$ for all input and measurement operators $\bar{\rho}, E$. 
\end{lemma}
\begin{proof}
	Let us write the two-valued measurement $E$ with outcomes $\pm1$ in terms of its POVM elements $\{M, I - M\}$, so that $E = M - (I-M) = 2M - I$. By definition $M$ satisfies $0 \leq M \leq I$. Since $\Gavg{2}$ is a CPTP map  and $\rho, \hat{\rho} \geq 0$ are quantum states, it follows that $\Gavg{2}(\rho), \Gavg{2}(\hat{\rho}) \geq 0$. Using the fact that $\Tr[AB] \geq 0$ for all positive semidefinite operators $A,B \geq 0$, it follows that
	\begin{equation}
	0 = \bbraket{ 0 |  \Gavgbf{2} | \rho} \leq \bbraket{M | \Gavgbf{2} | \rho} \leq \bbraket{ I | \Gavgbf{2}  |\rho} = 1.
	\end{equation}
	In terms of the measurement $E$, this means that $-1 \leq \bbraket{E | \Gavgbf{2} | \rho} \leq 1$. Analogously, this holds for $\hat{\rho}$. Since $\bar{\rho} = \half(\rho-\hat{\rho})$ is follows that $-1\leq B = \bbraket{E | \Gavgbf{2} | \bar{\rho}} \leq 1$.
\end{proof}
\begin{corollary}
	The quantities $\alpha, \beta$ as defined in \eqref{eq:alpha_app} and \eqref{eq:beta_app} satisfy $-1 \leq \alpha, \beta \leq 1$.
\end{corollary}
\begin{proof}
	\autoref{lem:optimality_operators} and \eqref{eq:B21} show that $-1 \leq \alpha \beta \leq 1$ for all $\bar{\rho},E$. Note that $\alpha$ only depends on $\bar{\rho}$ and $\beta$ only on $E$. Therefore if we fix $\bar{\rho} = \bar{\rho}_\ideal$ (which implies $\alpha = 1$), then we have $-1 \leq \beta \leq 1$. Analogously fixing $E = E_\ideal$ (which implies $\beta = 1$) yields $-1 \leq \alpha \leq 1$. 
\end{proof}
Very similar reasoning also gives the bound on the interval in which the sequence purity $q_\jbf^{(K)}$ lies (see \eqref{eq:interval_length}). This bound will be proven in the following lemma.
\begin{lemma}[Bound on interval lengths]
	\label{lem:interval_length}
	Let $q_\jbf^{(K)}$ denote the sequence purity of the $K$-copy implementation due to the random sequence $\jbf$ as defined in \eqref{eq:q_jbf^1} and \eqref{eq:q_jbf^2}: 
	\begin{equation}
	q_\jbf^{(1)} = \frac{1}{d^2-1}\sum_{P,Q \neq I} \bbraket{E_\Hc^{(Q)}  |\Gcbf_\jbf| \bar{\rho}_\Hc^{(P)}}^2 \qquad \mbox{and} \qquad q_\jbf^{(2)} = \bbraket{E | \Gc_\jbf^\tp{2} | \bar{\rho}}.
	\end{equation}
	Assume that $\alpha,\beta \geq 0$ (equivalent to $\Tr[\bar{\rho}_\ideal  \bar{\rho}] \geq 0$ and $\Tr[\bar{E}_\ideal  \bar{E}] \geq 0$ stated in \autoref{subsec:interval_length}). Then for all operators $\bar{\rho}, E$ (which are the effective operators in the single-copy implementation, see \eqref{eq:state_12}), all CPTP error maps $\Lambda$ and all sequences of Clifford gates indexed by $\jbf$,
	\begin{align}
	q_\jbf^{(1)} &\in [0, \alpha\beta + \beta \| \bar{\rho}_\spam \|_1 + \alpha \| \bar{E}_\spam \|_\infty + \| \bar{\rho}_\spam \|_1\| \bar{E}_\spam \|_\infty], \\
	q_\jbf^{(2)} &\in [-\beta \| \bar{\rho}_\spam \|_1 - \alpha\| \bar{E}_\spam \|_\infty - \| \bar{\rho}_\spam \|_1\| \bar{E}_\spam \|_\infty, 1].
	\end{align}
\end{lemma}
\begin{corollary}
	The interval length for $q_\jbf^{(1)}$ and $q_\jbf^{(2)}$ can be bounded independent of $\alpha,\beta$ by using that $\alpha,\beta \leq 1$ (\autoref{lem:optimality_operators}) as $L = 1 +  \| \bar{\rho}_\spam \|_1 +  \| \bar{E}_\spam \|_\infty + \| \bar{\rho}_\spam \|_1\| \bar{E}_\spam \|_\infty$.
\end{corollary}
\begin{proof}
	Starting with the two-copy implementation, let us write $E =  M -  (I - M) = 2M -  I$, where $0 \leq M \leq I$ is a POVM element (the measurement $E$ is described by the POVM set $\{M,I-M\}$, assigning outcome $1$ to $M$ and $-1$ to $I-M$). Then using the fact that $\Gc^\tp{2}_\jbf(\rho) \geq 0$ is positive semidefinite, it follows that
	\begin{equation*}
	0 = \Tr[0 \Gc^\tp{2}_\jbf(\rho)]\leq \Tr[M \Gc^\tp{2}_\jbf(\rho)] \leq \Tr[I \Gc^\tp{2}_\jbf(\rho)] = 1,
	\end{equation*} 
	expressing that $\Tr[M \Gc^\tp{2}_\jbf(\rho)]$ is indeed the probability associated with obtaining outcome $M$. Therefore $-1 \leq \Tr[E \Gc^\tp{2}_\jbf(\rho)] \leq 1$. Exactly the same argument holds for $\hat{\rho}$, so that (recall that $\bar{\rho} = \half(\rho - \hat{\rho})$)
	\begin{equation}
	\label{eq:B29}
	-1 \leq q_\jbf^{(2)} = \Tr[E \Gc_\jbf^\tp{2}(\bar{\rho})] \leq 1.
	\end{equation}
	The lower bound can be improved by using the decomposition \eqref{eq:alpha_app} and \eqref{eq:beta_app} to write $\bar{\rho} = \alpha \bar{\rho}_\ideal  + \bar{\rho}_\spam$ and $\bar{E} = \beta \bar{E}_\ideal  + \bar{E}_\spam$. Then 
	\begin{equation}
	\label{eq:B30}
	q_\jbf^{(2)} = \alpha \beta \Tr[\bar{E}_\ideal \Gc_\jbf^\tp{2}(\bar{\rho}_\ideal)] + \alpha \Tr[\bar{E}_\spam \Gc_\jbf^\tp{2}(\bar{\rho}_\ideal)] + \beta \Tr[\bar{E}_\ideal \Gc_\jbf^\tp{2}(\bar{\rho}_\spam)] +  \Tr[\bar{E}_\spam \Gc_\jbf^\tp{2}(\bar{\rho}_\spam)].
	\end{equation}
	The first term satisfies $\Tr[\bar{E}_\ideal \Gc_\jbf^\tp{2}(\bar{\rho}_\ideal)] \leq 1$ by \eqref{eq:B29} (which holds for all $E$, $\bar{\rho}$ so in particular for $E_\ideal$, $\bar{\rho}_\ideal$). However, we also find that 
	\begin{equation}
	\label{eq:B31}
	\Tr[\bar{E}_\ideal \Gc_\jbf^\tp{2}(\bar{\rho}_\ideal)] = \Tr[B_2 \Gc_\jbf^\tp{2}(B_2)] =  \frac{1}{d^2-1} \sum_{\sigma,\tau \in \Pc^*} \Tr[\sigma \Gc_\jbf(\tau)]^2 \geq 0.
	\end{equation}
	The remaining three terms in \eqref{eq:B30} are bounded using \autoref{prop:inner_holder}, which yields (using $\alpha,\beta \geq 0$)
	\begin{equation}
	\label{eq:B32}
	\begin{split}
	\alpha |\Tr[\bar{E}_\spam \Gc_\jbf^\tp{2}(\bar{\rho}_\ideal)]| &\leq \alpha \|\bar{E}_\spam\|_\infty \| \bar{\rho}_\ideal \| = \alpha \|\bar{E}_\spam\|_\infty   \\
	\beta |\Tr[\bar{E}_\ideal \Gc_\jbf^\tp{2}(\bar{\rho}_\spam)]| &\leq  \beta \| \bar{E}_\ideal \|_\infty \|\bar{\rho}_\spam\|_1 = \beta \|\bar{\rho}_\spam\|_1 \\ 
	|\Tr[\bar{E}_\spam \Gc_\jbf^\tp{2}(\bar{\rho}_\spam)]| & \leq \| \bar{E}_\spam \|_\infty \|\bar{\rho}_\spam\|_1
	\end{split}
	\end{equation}
	So by combining \eqref{eq:B30}, \eqref{eq:B31} and \eqref{eq:B32}, we find that 
	\begin{equation}
	q_\jbf^{(2)} \geq 0 - \alpha\| \bar{E}_\spam \|_\infty -\beta \| \bar{\rho}_\spam \|_1  - \| \bar{\rho}_\spam \|_1\| \bar{E}_\spam \|_\infty. 
	\end{equation}
 	The above argument also holds in the single-copy implementation if we let $E = E_\mathrm{eff}$ and $\bar{\rho} = \bar{\rho}_\mathrm{eff}$ as defined in \eqref{eq:state_12}. 
 	However, now we use it to upper bound $q_\jbf^{(1)}$. It follows that
	\begin{equation}
	\begin{split}
	q_\jbf^{(1)} &= \alpha \beta \Tr[\bar{E}_\ideal \Gc_\jbf^\tp{2}(\bar{\rho}_\ideal)] + \alpha \Tr[\bar{E}_\spam \Gc_\jbf^\tp{2}(\bar{\rho}_\ideal)] + \beta \Tr[\bar{E}_\ideal \Gc_\jbf^\tp{2}(\bar{\rho}_\spam)] +  \Tr[\bar{E}_\spam \Gc_\jbf^\tp{2}(\bar{\rho}_\spam)] \\
	&\leq \alpha\beta +\beta \| \bar{\rho}_\spam \|_1 + \alpha\| \bar{E}_\spam \|_\infty + \| \bar{\rho}_\spam \|_1\| \bar{E}_\spam \|_\infty.
	\end{split}
	\end{equation}
	The lower bound $q_\jbf^{(1)} \geq 0$ follows directly from the fact that it is defined as the sum of real numbers squared.
\end{proof}
So far we have recaptured the essential definitions and notations, shown optimality of the ideal operators and proven a bound in the interval in which the sequence purity $q_\jbf$ lies. Next we will state our variance bound \eqref{eq:var_bound} and give the complete proof.

\begin{thm}[Variance bound]
	\label{thm:var_bound}
	Let $\Hc$ be a $d$-dimensional Hilbert space, with $d = 2^q$ for a $q$-qubit system. Let $E \in V = \Lc(\Hc \otimes \Hc)$ be the Hermitian observable associated with a two-valued measurement with outcomes $\pm 1$ and $\rho, \hat{\rho} \in V = \Lc(\Hc \otimes \Hc)$ be two quantum states on two copies of the system. Consider the URB experiment (using the states and measurement $\rho, \hat{\rho},E$) of the Clifford group $\Cliff{d}$, assuming that a noisy implementation of $\Gc \in \Cliff{d}$ is given by $\tilde{\Gc} = \Gc \Lambda$, where $\Lambda$ is a CPTP map. In this experiment the sequence purity is $q_\jbf = \bbraket{\bar{E} | (\Mcbf^\tp{2})^{m-1} | \bar{\rho}}$, with $\Mc$ defined in \eqref{eq:N_def}. 
	
	Under the assumption that $d = 2$ or $\Lambda$ is unital (that is, $\Lambda(I) = I$), the following bound on the variance $\Vbb[q_\jbf]$ holds
	\begin{equation}
	\label{eq:var_bound_appendix}
	\Vbb[q_\jbf] \leq  \sigma^2 = \frac{1-u^{2(m-1)}}{1-u^2} (1 - u)^2 \Big( \alpha^2 \beta^2 c_1(d) + \alpha^2 c_2(d) \| \bar{E}_\spam \|_\infty^2 + \beta^2 c_3(d) \| \bar{\rho}_\spam \|_1^2 \Big)    +    \| \bar{\rho}_\spam \|_1^2 \| \bar{E}_\spam \|_\infty^2,   
	\end{equation}
	where $u$ is the unitarity of $\Lambda$, $m$ is the length of the sequence indexed by $\jbf$, $c_i(d)$ are functions only of the dimension $d$ and $\alpha$, $\beta$, $\bar{\rho}_\spam$ and $\bar{E}_\spam$ are defined in \eqref{eq:alpha_app} and \eqref{eq:beta_app}. Precise definitions of the dimension-dependent functions $c_i(d)$ will be given in the proof, but closed form expressions are messy and therefore not written down explicitly. Asymptotically, these functions satisfy
	\begin{equation}
	c_1(d) = \Oc(1), \quad\quad c_2(d) = \Oc(d), \quad\quad c_3(d) = \Oc(d^2).
	\end{equation}  
\end{thm}
\begin{proof}
	We start from the derived expression for the variance \eqref{eq:var_expr1_appendix}. First, let us decompose the state and measurement operators in ideal and error components as (see Eqs.~(\ref{eq:ideal_operators_app})-(\ref{eq:beta_app}))
	\begin{equation}
	\bar{\rho} = \alpha \bar{\rho}_\ideal + \bar{\rho}_\spam \quad\quad\mbox{and}\quad\quad \bar{E} = \beta \bar{E}_\ideal + \bar{E}_\spam.
	\end{equation}
	Define again $W = \Span\{ B_1, B_2 \} \subset V$, with $B_1$, $B_2$ defined in \eqref{def:def_B1} and \eqref{def:def_B2} respectively. Then the ideal components $\bar{\rho}_\ideal$ and $\bar{E}_\ideal$ are in $W$ and the error components $\bar{\rho}_\spam$ and $\bar{E}_\spam$ are in the orthogonal complement $W^\perp$. Plugging this expansion into \eqref{eq:var_expr1_appendix} in principle yields 16 terms. However, the 12 terms with an ideal component tensor error component (e.g., $\bar{\rho}_\spam \otimes \bar{\rho}_\ideal$) vanish, because both 
	\begin{align}
	\label{eq:B27}
	(\Gavg{2})^\tp{2}(W \otimes W^\perp) &= (\Gavg{2})^\tp{2}(W^\perp \otimes W) = \emptyset, \\
	\label{eq:B28}
	\Gavg{4}(W \otimes W^\perp) &= \Gavg{4}(W^\perp \otimes W) =  \emptyset.
	\end{align}
	\eqref{eq:B27} is easy to see because $\Gavg{2}$ is the orthogonal projection onto $W$. \eqref{eq:B28} follows from the fact that $W$ carries the trivial subrepresentations of the Liouville tensor-2 representation and $W^\perp$ carries all other necessarily nontrivial subrepresentations. By \autoref{lem:trivial_subrep_lemma} the spaces $W^\perp \otimes W$ and $W \otimes W^\perp$ (which are representations of the Liouville tensor-4 representation) do not carry trivial subrepresentations. Hence $\Gavg{4}$, the projector onto the trivial subrepresentations of the Liouville tensor-4 rep (by \autoref{lem:projection_onto_trivial_subreps}), does not project onto any subspace of $W^\perp \otimes W$ and $W \otimes W^\perp$. This justifies the following expression for the variance 
	\begin{align}
	\label{eq:ideal_ideal_app}
	\Vbb[q_\jbf ] =& \, \alpha^2 \beta^2 \bbraket{B_2^\tp{2} | \pmb{\Nc}^{m-1} - (\pmb{\Mc}^\tp{2})^{m-1} | B_2^\tp{2}} \\ 
	\label{eq:spam_ideal_app}
	&+ \frac{\alpha^2}{{d^2-1}}  \bbraket{\bar{E}_\spam^\tp{2} | \pmb{\Nc}^{m-1} - (\pmb{\Mc}^\tp{2})^{m-1} | B_2^\tp{2}} \\ 
	\label{eq:ideal_spam_app}
	&+ (d^2-1)\beta^2 \bbraket{B_2^\tp{2} | \pmb{\Nc}^{m-1} - (\pmb{\Mc}^\tp{2})^{m-1} | \bar{\rho}_\spam^\tp{2}} \\ 
	\label{eq:spam_spam_app}
	&+ \bbraket{\bar{E}_\spam^\tp{2} | \pmb{\Nc}^{m-1} - (\pmb{\Mc}^\tp{2})^{m-1} | \bar{\rho}_\spam^\tp{2}},
	\end{align} 
	where the expressions of \eqref{eq:ideal_operators_app} are used for the ideal operators $\bar{\rho}_\ideal, \bar{E}_\ideal$. We will analyze each of the four terms separately. 
	The term we start with is \eqref{eq:spam_ideal_app}, since this term most clearly conveys the idea of our analysis. Then the terms \eqref{eq:ideal_ideal_app} and \eqref{eq:ideal_spam_app} are treated in similar fashion, but with a small additional technicality. Finally the term \eqref{eq:spam_spam_app} is treated in a totally different fashion.
	
	The analysis of \eqref{eq:spam_ideal_app} starts by using \autoref{lem:telescoping} (telescoping series lemma), so that we can write this term as
	\begin{equation}
	\begin{split}
	\eqref{eq:spam_ideal_app} &= \frac{\alpha^2}{{d^2-1}} \sum_{s=1}^{m-1}  \bbraket{\bar{E}_\spam^\tp{2} | \Ncbf^{m-s-1}[\Ncbf - \Mcbf^\tp{2}](\Mcbf^\tp{2})^{s-1} | B_2^\tp{2}} \\
	&= \frac{\alpha^2}{{d^2-1}} \sum_{s=1}^{m-1} u^{2(s-1)}  \bbraket{\bar{E}_\spam^\tp{2} | \Ncbf^{m-s-1}[\Ncbf - \Mcbf^\tp{2}] | B_2^\tp{2}}. \\
	\end{split}
	\end{equation}
	In the second line we used that $\Mcbf \kket{B_2} = u \kket{B_2}$. The idea is to expand $\Ncbf - \Mcbf^\tp{2}  \kket{B_2^\tp{2}}$ in the basis $\{ A_i : i \in \Zc_{TS} \}$ of the subspace $V_{TS} \otimes V_{TS} \subset \Rge(\Gavg{4}) \subset V \otimes V$. $V_{TS}$ is the trace-preserving, symmetric subspace of $V = \Lc(\Hc^\tp{2})$, as defined in Appendix~\ref{app_sub:trivial_subreps_Liouville_4}. The restriction of $\Gavg{4}$ to $V_{TS} \otimes V_{TS}$ is justified by the fact that $\Lambda^\tp{2}(B_2) \in V_{TS}$. Hence we expand
	\begin{equation}
	\label{eq:def_ai}
	\Ncbf - \Mcbf^\tp{2} \kket{B_2^\tp{2}} = \sum_{i \in \Zc_{TS}} a_i \kket{A_i}, \qquad \mbox{where} \qquad 	a_i := \bbraket{A_i | \Ncbf - \Mcbf^\tp{2} | B_2^\tp{2}}.
	\end{equation} 
	Therefore \eqref{eq:spam_ideal_app} can be written as	
	\begin{equation}
	\eqref{eq:spam_ideal_app} = \frac{\alpha^2}{d^2 - 1}  \sum_{s=1}^{m-1} u^{2(s-1)}  \sum_{i \in \Zc_{TS}} a_i \bbraket{\bar{E}_\spam^\tp{2} | \Ncbf^{m-s-1} | A_i}.
	\end{equation}
	
	For the terms \eqref{eq:ideal_ideal_app} and \eqref{eq:ideal_spam_app}, something similar is done. The telescoping series (\autoref{lem:telescoping}) is now written in the other way. Therefore we can write \eqref{eq:ideal_ideal_app} as
	\begin{align}
	\label{eq:B50}
	\eqref{eq:ideal_ideal_app} &= \alpha^2 \beta^2 \sum_{s=1}^{m-1}  \bbraket{B_2^\tp{2} | (\Mcbf^\tp{2})^{s-1}[\Ncbf - \Mcbf^\tp{2}]\Ncbf^{m-s-1} | B_2^\tp{2}} \\
	\label{eq:B51}
	&= \alpha^2 \beta^2 \sum_{s=1}^{m-1} u^{2(s-1)}  \bbraket{B_2^\tp{2} | [\Ncbf - \Mcbf^\tp{2}]\Ncbf^{m-s-1} | B_2^\tp{2}} 
	\end{align}
	The step from \eqref{eq:B50} to \eqref{eq:B51} is not immediately clear, since
	\begin{equation}
	\bbra{B_2 B_2} (\Mcbf^\tp{2})^{s-1} = x_{11}^{(s)} \bbra{B_1 B_1} +  x_{12}^{(s)} \bbra{B_1 B_2} + x_{21}^{(s)} \bbra{B_2 B_1} + u^{2(s-1)} \bbra{B_2 B_2},
	\end{equation}
	for some coefficients $x_{11}^{(s)}, x_{12}^{(s)}, x_{21}^{(s)} \in \R$. However we show that \eqref{eq:B51} is justified, since
	\begin{equation}
	\bbraket{B_k B_l |  [\Ncbf - \Mcbf^\tp{2}]\Ncbf^{m-s-1} | B_2^\tp{2}} = 0, \qquad \mbox{if } k = 1 \mbox{ or } l = 1.
	\end{equation}
	This follows from the trace-preserving properties of $\Nc, \Mc$, the tracelessness of $B_2$ and the fact that $B_1 = \frac{I}{d}$. In particular,
	\begin{align}
	\bbraket{B_k B_l |  \Ncbf^{m-s} | B_2^\tp{2}} &= \frac{1}{|\Cliff{d}|^{m-s}} \sum_\jbf  \bbraket{B_k |  \Gcbf_\jbf^\tp{2} | B_2} \bbraket{B_l |  \Gcbf_\jbf^\tp{2} | B_2} = 0, \\
	\bbraket{B_k B_l |  \Mcbf^\tp{2} \Ncbf^{m-s-1} | B_2^\tp{2}} &= \frac{1}{|\Cliff{d}|^{m-s-1}} \sum_\jbf  \bbraket{B_k |  \Mcbf \Gcbf_\jbf^\tp{2} | B_2} \bbraket{B_l | \Mcbf \Gcbf_\jbf^\tp{2} | B_2} = 0,
	\end{align}
	if $l = 1$ or $k = 1$, since $\bbraket{B_1 | \Mcbf \Gcbf_\jbf^\tp{2} | B_2} = 0$ and $\bbraket{B_1 |  \Gcbf_\jbf^\tp{2} | B_2} = 0$. This justifies \eqref{eq:B51}. Next we use a similar expansion
	\begin{equation}
	\label{eq:def_bi}
	\bbra{B_2^\tp{2}} \Ncbf - \Mcbf^\tp{2} = \sum_{i \in \Zc_{TS}} b_i \bbra{A_i}, \qquad \mbox{where} \qquad 	
	b_i := \bbraket{B_2^\tp{2} | \Ncbf - \Mcbf^\tp{2} | A_i}.
	\end{equation}
	Therefore we arrive at  
	\begin{equation}
	\eqref{eq:ideal_ideal_app} 	= \alpha^2 \beta^2 \sum_{s=1}^{m-1} u^{2(s-1)}  \sum_{i \in \Zc_{TS}} b_i \bbraket{A_i | \Ncbf^{m-s-1} | B_2^\tp{2}}.
	\end{equation}
	Similarly to the analysis \eqref{eq:ideal_ideal_app}, we can write \eqref{eq:ideal_spam_app} as
	\begin{equation}
	\eqref{eq:ideal_spam_app} = (d^2-1) \beta^2 \sum_{s=1}^{m-1} u^{2(s-1)}  \sum_{i \in \Zc_{TS}} b_i \bbraket{A_i | \Ncbf^{m-s-1} | \bar{\rho}_\spam^\tp{2}}. \\
	\end{equation}
	Finally, we slightly rewrite \eqref{eq:spam_spam_app} by noting that $\eqref{eq:spam_spam_app} = \bbraket{\bar{E}_\spam^\tp{2} | \pmb{\Nc}^{m-1}  | \bar{\rho}_\spam^\tp{2}}$, because $\Mcbf \kket{\rho_\spam} = 0$. We therefore arrive at the following expression of the variance
	\begin{align}
	\label{eq:var_expr2_id_id}
	\Vbb[q_\jbf ] =& \, \alpha^2 \beta^2 \sum_{s=1}^{m-1} u^{2(s-1)}  \sum_{i \in \Zc_{TS}}  a_i  \bbraket{B_2^\tp{2} | \Ncbf^{m-s-1} | A_i}   \\ 
	\label{eq:var_expr2_sp_id}
	&+ \frac{1}{{d^2-1}} \alpha^2  \sum_{s=1}^{m-1} u^{2(s-1)}  \sum_{i \in \Zc_{TS}}  a_i  \bbraket{\bar{E}_\spam^\tp{2} | \Ncbf^{m-s-1} | A_i}   \\ 
	\label{eq:var_expr2_id_sp}
	&+ (d^2-1) \beta^2   \sum_{s=1}^{m-1} u^{2(s-1)}  \sum_{i \in \Zc_{TS}}  b_i  \bbraket{A_i | \Ncbf^{m-s-1} | \bar{\rho}_\spam^\tp{2}}   \\  
	\label{eq:var_expr2_sp_sp}
	&+ \bbraket{\bar{E}_\spam^\tp{2} | \pmb{\Nc}^{m-1} | \bar{\rho}_\spam^\tp{2}}.
	\end{align}
	This expression is still exact, as we have only expanded each term in the equation. 
	
	The variance bound is obtained by bounding the remaining inner products and the quantities $a_i, b_i$ in this expression. This technical task is delegated to Appendix~\ref{app:technical_lemmas}, with a number of technical propositions that compute bounds on the quantities above. We summarize the results here.	The bounds on $a_i$ and $b_i$ for $i \in \{0; \adj; S; 1,2 \}$ are obtained under the assumption that $d=2$ or that $\Lambda$ is unital in Propositions \ref{prop:bound_ai_bi_lower}, \ref{prop:bound_a0}, \ref{prop:bound_a12}, \ref{prop:bound_aS}, \ref{prop:bound_aadj}, and \ref{prop:bound_bi}. In summary 
	\begin{align}
	0 &= a_0 = b_0, \\
	0 &\leq a_{1,2}, b_{1,2} \leq \frac{\sqrt{d^2-2}}{d^2} (1-u)^2, \\
	0 &\leq a_S, b_S \leq \sqrt{\frac{d^2 - 2}{d^2-1}} \sqrt{2} (1-u)^2, \\
	0 &\leq a_\adj, b_\adj \leq \sqrt{d^2 - 1} (1-u)^2. 
	\end{align}	
	In the case of $d \geq 4$, bounds on $a_i$ are needed for $i \in \Zc_{1,2} \cup \Zc_S \setminus\{\adj\}$ in terms of the above bounds on $a_S$ and $a_{1,2}$. To do so, we use \eqref{eq:A_j_not_irrep}, which states
	\begin{align}
	\sqrt{|V_{1,2}|} A_{1,2} &= \sum_{i \in \Zc_{1,2}} \sqrt{|V_i|} A_i & \sqrt{|V_S|} A_S &= \sum_{i \in \Zc_{S}} \sqrt{|V_i|} A_i. 
	\end{align}
	From this it follows that
	\begin{align}
	\sqrt{|V_{1,2}|} a_{1,2} &= \sum_{i \in \Zc_{1,2}} \sqrt{|V_i|} a_i, & \sqrt{|V_S|} a_S &= \sum_{i \in \Zc_{S}} \sqrt{|V_i|} a_i, \\
	\sqrt{|V_{1,2}|} b_{1,2} &= \sum_{i \in \Zc_{1,2}} \sqrt{|V_i|} b_i & \sqrt{|V_S|} b_S &= \sum_{i \in \Zc_{S}} \sqrt{|V_i|} b_i.
	\end{align}
	Thus, since $a_i, b_i \geq 0$ by \autoref{prop:bound_ai_bi_lower}, these equations imply the following bounds
	\begin{align}
	a_i &\leq \sqrt{\frac{{|V_{1,2}|}}{|V_i|}} a_{1,2}, & b_i &\leq \sqrt{\frac{{|V_{1,2}|}}{|V_i|}} b_{1,2}, & & \forall i \in \Zc_{1,2} \\ 	
	a_i &\leq \sqrt{\frac{{|V_{S}|}}{|V_i|}} a_{S}, & b_i &\leq \sqrt{\frac{{|V_{S}|}}{|V_i|}} b_{S}, & & \forall i \in \Zc_{S} \setminus \{\adj\}.
	\end{align}
	The size of the relevant spaces (as derived in \cite{Helsen2018}) was summarized in \autoref{fig:V_TS}.
	The inner products in \eqref{eq:var_expr2_id_id}-\eqref{eq:var_expr2_sp_sp} are bounded using Propositions \ref{prop:inner_holder}, \ref{prop:bound_norms_Ai} and \ref{prop:ideal_inner}. \autoref{prop:inner_holder} is applicable since $\Nc^m$ is a CPTP map for any $m \in \N$, since CPTP maps are closed under composition. Now $\Nc$ is CPTP because $\Nc$ is the convex combination of the CPTP sequences $\Gc_\jbf$ and a convex combination of CPTP maps is CPTP. 
	The results of Propositions \ref{prop:inner_holder}, \ref{prop:bound_norms_Ai} and \ref{prop:ideal_inner} are summarized as follows:
	\begin{align}
	\bbraket{A_i | \Ncbf^{m-s-1} | B_2^\tp{2}} &\leq \frac{1}{\sqrt{|V_i|}}, \\
	\label{eq:B46}
	\bbraket{\bar{E}_\spam^\tp{2} | \pmb{\Nc}^{m-s-1}  | A_i} &\leq d^2 \| \bar{E}_\spam \|_\infty^2 , \\
	\label{eq:B47}
	\bbraket{A_i | \pmb{\Nc}^{m-s-1}  | \bar{\rho}_\spam^\tp{2}} &\leq \sqrt{\frac{6}{(d-2)(d-1)}} \| \bar{\rho}_\spam \|_1^2, \\
	\bbraket{\bar{E}_\spam^\tp{2} | \pmb{\Nc}^{m-1}  | \bar{\rho}_\spam^\tp{2}} &\leq \| \bar{E}_\spam \|_\infty^2 \| \bar{\rho}_\spam \|_1^2, 
	\end{align}
	where we have used that $\| A^\tp{k} \|_p = \| A \|_p^k$ for any $k\in \N$ and $p \in [1,\infty]$. \eqref{eq:B46} and \eqref{eq:B47} have single-qubit specific ($d=2$) improvements (derived in \autoref{prop:bound_norms_Ai}), using the fact that $V_{1,2}$ and $V_{S}$ actually are irreducible subrepresentations. Since we have explicit expressions for $A_{1,2}$ and $A_{S}$ (\eqref{eq:A12} and \eqref{eq:AS} respectively), their norms can be computed directly. Using this gives the improved single-qubit bounds,
	\begin{align}
	\bbraket{\bar{E}_\spam^\tp{2} | \pmb{\Nc}^{m-s-1}  | A_S} &\leq \frac{5}{\sqrt{3}} \| \bar{E}_\spam \|_\infty^2, & 
	\bbraket{\bar{E}_\spam^\tp{2} | \pmb{\Nc}^{m-s-1}  | A_{1,2}} &\leq 2 \sqrt{2} \| \bar{E}_\spam \|_\infty^2, \\
	\bbraket{A_S | \pmb{\Nc}^{m-s-1}  | \bar{\rho}_\spam^\tp{2}} &\leq \frac{1}{\sqrt{3}} \| \bar{\rho}_\spam \|_1^2, &
	\bbraket{A_{1,2} | \pmb{\Nc}^{m-s-1}  | \bar{\rho}_\spam^\tp{2}} &\leq \frac{\sqrt{2}}{3} \| \bar{\rho}_\spam \|_1^2.
	\end{align}
	Plugging all of these bounds into \eqref{eq:var_expr2_id_id}-\eqref{eq:var_expr2_sp_sp} and using the geometric series
	\begin{equation}
	\sum_{s=1}^{m-1} u^{2(s-1)} = \frac{1 - u^{2(m-1)}}{1-u^2}
	\end{equation}
	will yield the bound \eqref{eq:var_bound_appendix}
	\begin{equation}
	\Vbb[q_\jbf] \leq  \sigma^2 = \frac{1-u^{2(m-1)}}{1-u^2} (1 - u)^2 \Big( \alpha^2 \beta^2 c_1(d) + \alpha^2 c_2(d) \| \bar{E}_\spam \|_\infty^2 + \beta^2 c_3(d) \| \bar{\rho}_\spam \|_1^2 \Big)    +    \| \bar{\rho}_\spam \|_1^2 \| \bar{E}_\spam \|_\infty^2,   
	\end{equation}
	where 
	\begin{align*}
	c_1(d) &= \begin{dcases}
	\frac{\sqrt{2}}{4} \frac{1}{\sqrt{2}} + \sqrt{\frac{2}{3}} \sqrt{2} \frac{1}{\sqrt{3}} = \frac{11}{12}, &  \mbox{if } d=2, \\
	\mathrlap{\frac{\sqrt{d^2-2}}{d^2} \sum_{i \in \Zc_{1,2}} {\frac{\sqrt{|V_{1,2}|}}{|V_i|}} + \sqrt{2}\sqrt{\frac{d^2-2}{d^2-1}} \sum_{i \in \Zc_{S} \setminus \{\adj \}} \frac{\sqrt{|V_S|}}{|V_i|} + \frac{\sqrt{d^2-1}}{\sqrt{|V_\adj|}},}
	%spacing
	\hphantom{(d^2-1) \sqrt{\frac{6}{(d-2)(d-1)}} \left(\frac{\sqrt{d^2-2}}{d^2} \sum_{i \in \Zc_{1,2}} {\frac{\sqrt{|V_{1,2}|}}{\sqrt{|V_i|}}} + \sqrt{2}\sqrt{\frac{d^2-2}{d^2-1}} \sum_{i \in \Zc_{S} \setminus \{\adj \}} \frac{\sqrt{|V_S|}}{\sqrt{|V_i|}} + \sqrt{d^2-1}\right),}
	 & \mbox{if } d\geq 4, \\
	\end{dcases} \\
	c_2(d) &= \begin{dcases} 
	\frac{1}{3}\left(\frac{\sqrt{2}}{4} 2 \sqrt{2} + \sqrt{\frac{2}{3}} \sqrt{2} \frac{5}{\sqrt{3}}\right) = \frac{13}{9}, & \mbox{if } d=2, \\
	\mathrlap{\frac{d^2}{d^2-1}\left(\frac{\sqrt{d^2-2}}{d^2} \sum_{i \in \Zc_{1,2}} {\frac{\sqrt{|V_{1,2}|}}{\sqrt{|V_i|}}} + \sqrt{2}\sqrt{\frac{d^2-2}{d^2-1}} \sum_{i \in \Zc_{S} \setminus \{\adj \}} \frac{\sqrt{|V_S|}}{\sqrt{|V_i|}} + \sqrt{d^2-1}\right),} 
	%spacing
	\hphantom{(d^2-1) \sqrt{\frac{6}{(d-2)(d-1)}} \left(\frac{\sqrt{d^2-2}}{d^2} \sum_{i \in \Zc_{1,2}} {\frac{\sqrt{|V_{1,2}|}}{\sqrt{|V_i|}}} + \sqrt{2}\sqrt{\frac{d^2-2}{d^2-1}} \sum_{i \in \Zc_{S} \setminus \{\adj \}} \frac{\sqrt{|V_S|}}{\sqrt{|V_i|}} + \sqrt{d^2-1}\right),}
	& \mbox{if } d\geq 4, \\
	\end{dcases} \\
	c_3(d) &= \begin{dcases}
	3\left(\frac{\sqrt{2}}{4} \frac{\sqrt{2}}{3} + \sqrt{\frac{2}{3}} \sqrt{2} \frac{1}{\sqrt{3}}\right) = \frac{5}{2}, & \mbox{if } d=2, \\
	(d^2-1) \sqrt{\frac{6}{(d-2)(d-1)}} \left(\frac{\sqrt{d^2-2}}{d^2} \sum_{i \in \Zc_{1,2}} {\frac{\sqrt{|V_{1,2}|}}{\sqrt{|V_i|}}} + \sqrt{2}\sqrt{\frac{d^2-2}{d^2-1}} \sum_{i \in \Zc_{S} \setminus \{\adj \}} \frac{\sqrt{|V_S|}}{\sqrt{|V_i|}} + \sqrt{d^2-1}\right),  & \mbox{if } d\geq 4.
	\end{dcases}
	\end{align*}
	The size of the spaces $V_i$ in these equations are found in \autoref{fig:V_TS}. The asymptotic behavior of the dimension-dependent functions $c_i(d)$ can be found if all relevant dimensions of the spaces are plugged into the above equations.
	\end{proof}

\section{Bounds on individual quantities in the proof}
\label{app:technical_lemmas}
This section provides the technical lemmas and propositions referred to in the previous section. They are collected here together in an attempt not to clutter the main line of the proof. Most of these technical lemmas put a bound on quantities arising in the proof of \autoref{thm:var_bound}. 

We start by bounding the $a_i$. Only bounds on $a_0$, $a_S$, $a_{1,2}$ and $a_{\adj}$ are provided. In the multiqubit case where $V_S$ and $V_{1,2}$ are not irreducible representations, the quantities $a_i$ for $i \in \Zc_S, \Zc_{d}$  are bounded by $a_S$ and $a_{1,2}$. The only exception is $i = \adj$, for which we provide a separate bound. Let us start with showing that all $a_i$ and $b_i$ are nonnegative.
\begin{proposition}[Lower bound on $a_i$ and $b_i$]
	\label{prop:bound_ai_bi_lower}
	For all CPTP $\Lambda$ and all $i \in \Zc_{TS}$, one has
	\begin{align}
	a_i &= \bbraket{A_i | \Ncbf - \Mcbf^\tp{2} | B_2B_2} \geq 0,& b_i &= \bbraket{B_2B_2 | \Ncbf - \Mcbf^\tp{2} | A_i} \geq 0.
	\end{align}
\end{proposition}
\begin{proof}
	If $i = 0$, then \autoref{prop:bound_a0} will show that $a_0 = 0$, which includes this lower bound. For all other $i \in \Zc_{TS} \setminus \{0\}$, we show that $\Mc^\tp{2} \kket{A_i} = 0$. This is because $\Mc^\tp{2}$ is supported on $W^\tp{2} = \Span\{B_1 B_1, B_1 B_2, B_2B_1, B_2B_2 \}$, where $A_0 = B_2B_2$. But $B_1B_1, B_1B_2, B_2B_1 \in (V_{TS}^\tp{2})^\perp$. Since $A_i \in V_{TS}^\tp{2} \setminus \Span\{A_0\}$ the claim follows. Therefore $a_i = \bbraket{B_2 B_2 | \Ncbf | A_i}$. Using the definitions of $\Nc$ (\eqref{eq:N_def}) and $A_i$ (\eqref{eq:A_i}), it follows that 
	\begin{equation}
	\begin{split}
	a_i &= \frac{1}{|V_i| |\Cliff{d}|^2} \sum_{\Gc, \Gc' \in \Cliff{d}} \sum_{v_i \in \Bc_i} \bbraket{v_i v_i | \Gcbf^\tp{4} \Lambdabf^\tp{4} \pmb{\Gc'}^\tp{4} | B_2 B_2} \\
	&= \frac{1}{|V_i| |\Cliff{d}|^2} \sum_{\Gc, \Gc' \in \Cliff{d}} \sum_{v_i \in \Bc_i} \bbraket{v_i | \Gcbf^\tp{2} \Lambdabf^\tp{2} \pmb{\Gc'}^\tp{2} | B_2}^2 \geq 0, \\
	\end{split}
	\end{equation}
	which is nonnegative as it is the sum of real numbers squared. Analogously,
	\begin{equation*}
	b_i = \frac{1}{|V_i| |\Cliff{d}|^2} \sum_{\Gc, \Gc' \in \Cliff{d}} \sum_{v_i \in \Bc_i} \bbraket{B_2 | \Gcbf^\tp{2} \Lambdabf^\tp{2} \pmb{\Gc'}^\tp{2} | v_i}^2 \geq 0. \qedhere
	\end{equation*}
\end{proof}
Next we show that $a_0$ vanishes.
\begin{proposition}[Bound on $a_0$]
	\label{prop:bound_a0}
	Let $a_0$ be defined by \eqref{eq:def_ai}. Then for all CPTP quantum channels $\Lambda$, $a_0 = 0$.
\end{proposition}
\begin{proof}
	By definition of \eqref{eq:def_ai} it follows that (using that $A_0 = B_2 B_2$ by definition of \eqref{eq:A0})
	\begin{equation}
	\begin{split}
	a_0 &= \bbraket{A_0 | \pmb{\Nc} - \pmb{\Mc}^{\otimes 2} | B_2  B_2} = \bbraket{B_2  B_2 | \pmb{\Lambda}^{\otimes 4} - \pmb{\Lambda}^{\otimes 4} | B_2  B_2} = 0,
	\end{split} 
	\end{equation}
	since $\Gavgbf{4} \kket{B_2 B_2} = (\Gavgbf{2})^{\otimes 2} \kket{B_2 B_2} =  \kket{B_2 B_2}$.
\end{proof}
The next proposition gives a bound on $a_{1,2}$.
\begin{proposition}[Bound on $a_{1,2}$]
	\label{prop:bound_a12}
	Let $a_{1,2}$ be defined as in \eqref{eq:def_ai} and let $\Lambda$ be a CPTP map. If $\Lambda$ is a single-qubit channel (i.e., if $d = 2$) or if $\Lambda$ is unital [i.e., $\Lambda(I) = I$], then 
	\begin{equation}
	\label{eq:a12bound}
	a_{1,2} =  \frac{1}{\sqrt{d^2 - 2}} \left( \frac{1}{d^2-1} \sum_{\sigma \in \Pc^*} \bbraket{\sigma | \pmb{\Lambda}_\urm \pmb{\Lambda}_\urm^\dagger | \sigma }^2    -  u^2 \right)  \leq \frac{\sqrt{d^2-2}}{d^2} (1-u)^2.
	\end{equation}
\end{proposition}
\begin{proof}
	By the definition \eqref{eq:def_ai}, $a_{1,2} = \bbraket{A_{1,2} | \pmb{\Nc} - \pmb{\Mc}^{\otimes 2} | B_2  B_2}$, where
	\begin{equation}
	B_2  B_2 = \frac{1}{d^2 - 1} \sum_{\sigma,\tau \in \Pc^*} \sigma\sigma\tau\tau \qquad \mbox{and} \qquad
	A_{1,2} = \frac{1}{\sqrt{d^2-2}} \left( \sum_{\sigma \in \Pc^*} \sigma^{\otimes 4}    -    A_0 \right)
	\end{equation}
	were defined in \eqref{eq:A0} and \eqref{eq:A12}, respectively. Therefore $a_{1,2}$ is computed as (recalling that $A_0 = B_2 B_2$ and using \eqref{eq:unitarity_eigenvalue})  
	\begin{equation}
	\label{eq:a12def}
	\begin{split}
	a_{1,2} &= \bbraket{A_{1,2} | \pmb{\Nc} - \pmb{\Mc}^{\otimes 2} | B_2  B_2} \\
	&= \frac{1}{(d^2-1)  \sqrt{d^2-2} }  \left(\sum_{\sigma, \hat{\sigma}, \hat{\tau} \in \Pc^*} \bbraket{\sigma\sigma\sigma\sigma | \pmb{\Lambda}^{\otimes 4} | \hat{\sigma}\hat{\sigma}\hat{\tau}\hat{\tau}  } -  \bbraket{B_2 B_2 | \pmb{\Lambda}^{\otimes 4} | B_2 B_2}   \right) \\
	&= \frac{1}{  \sqrt{d^2-2} }  \left( \frac{1}{d^2-1} \sum_{\sigma, \hat{\sigma}, \hat{\tau} \in \Pc^*} \bbraket{\sigma | \pmb{\Lambda} | \hat{\sigma}}^2 \bbraket{\sigma | \pmb{\Lambda} | \hat{\tau}}^2  -  u^2   \right) \\
	&= \frac{1}{  \sqrt{d^2-2} }  \left( \frac{1}{d^2-1} \sum_{\sigma \in \Pc^*} \bbraket{\sigma | \pmb{\Lambda}_\urm \pmb{\Lambda}_\urm^\dagger | {\sigma}}^2  -  u^2   \right),
	\end{split}
	\end{equation}
	where in the last step, the following was used 
	\begin{equation}
	\label{eq:LLdagger_trick}
	\sum_{\hat{\sigma} \in \Pc^*} \bbraket{\sigma | \pmb{\Lambda} | \hat{\sigma}} \bbraket{\tau | \pmb{\Lambda} | \hat{\sigma}} = \sum_{\hat{\sigma} \in \Pc^*} \bbraket{\sigma | \pmb{\Lambda} | \hat{\sigma}} \bbraket{\hat{\sigma} | \pmb{\Lambda}^\dagger | \tau } = \bbraket{\sigma | \pmb{\Lambda}_\urm \pmb{\Lambda}_\urm^\dagger | {\tau}}, \quad\quad \forall \sigma,\tau \in \Pc^*,
	\end{equation}
	abusing notation slightly by writing $\pmb{\Lambda}_\urm$ instead of $1 \oplus \pmb{\Lambda}_\urm$ and using the fact that $\sum_{\hat{\sigma} \in \Pc^*} \kketbra{\hat{\sigma}}$ is the projection onto the unital block.
		
	The bound of \eqref{eq:a12bound} is then shown as follows. The idea is to apply \autoref{lem:diagonals_squared} to the map
	\begin{equation}
	\pmb{\Ec} := \begin{bmatrix}
	1 & 0 \\
	0 & \pmb{\Lambda}_\urm \pmb{\Lambda}_\urm^\dagger
	\end{bmatrix},
	\end{equation}
	since this map is constructed such that 
	\begin{equation}
	f(\Ec) = \frac{1}{d^2-1} \sum_{\sigma \in \Pc^*} \bbraket{\sigma | \pmb{\Ec} | {\sigma}} = \frac{1}{d^2-1} \sum_{\sigma \in \Pc^*} \bbraket{\sigma | \pmb{\Lambda}_\urm \pmb{\Lambda}_\urm^\dagger | {\sigma}} = u(\Lambda)
	\end{equation}  
	and
	\begin{equation}
	\frac{1}{d^2-1} \sum_{\sigma \in \Pc^*} \bbraket{\sigma | \pmb{\Lambda}_\urm \pmb{\Lambda}_\urm^\dagger | {\sigma}}^2  -  u(\Lambda)^2 = \frac{1}{d^2-1} \sum_{\sigma \in \Pc^*} \bbraket{\sigma | \pmb{\Ec} | {\sigma}}^2  -  f(\Ec)^2.
	\end{equation}
	Application of \autoref{lem:diagonals_squared} requires the map $\Ec$ to be CPTP. This is guaranteed by \autoref{lem:unital_part}, using the assumption that $\Lambda$ is a single-qubit or unital channel. Therefore \autoref{lem:diagonals_squared} applied to the channel $\Ec$ defined above, yields
	\begin{equation}
	\frac{1}{d^2-1}\sum_{\sigma \in \Pc^*} \bbraket{\sigma | \pmb{\Lambda}_\urm \pmb{\Lambda}_\urm^\dagger | {\sigma}}^2  -  u(\Lambda)^2 = \frac{1}{d^2-1} \sum_{\sigma \in \Pc^*} \bbraket{\sigma | \pmb{\Ec} | {\sigma}}^2  -  f(\Ec)^2 \leq  \frac{d^2 - 2}{d^2} (1-f(\Ec))^2 = \frac{d^2 - 2}{d^2} (1-u(\Lambda))^2.
	\end{equation}
	Plugging this into \eqref{eq:a12def} yields the result.
\end{proof}
The next proposition bounds the quantity $a_S$.
\begin{proposition}[Bound on $a_S$]
	\label{prop:bound_aS}
	Let $a_{S}$ be defined as in \eqref{eq:def_ai} and let $\Lambda$ be a CPTP map. If $\Lambda$ is a single-qubit channel (i.e., if $d = 2$) or if $\Lambda$ is unital [i.e., $\Lambda(I) = I$], then 
	\begin{equation}
	a_S = \frac{\sqrt{2}}{(d^2-1)^{\frac{3}{2}}(d^2-2)^\half} \sum_{\substack{\sigma, \tau  \in \Pc^* \\ \sigma \neq \tau}} \bbraket{\sigma | \pmb{\Lambda}_\urm \pmb{\Lambda}_\urm^\dagger | {\tau}}^2 \leq \sqrt{\frac{d^2 - 2}{d^2-1}} \sqrt{2} (1-u)^2.
	\end{equation}
\end{proposition}
\begin{proof}
	First, let us show the evaluation of $a_S$. 
	By the definition \eqref{eq:def_ai}, $a_{S} = \bbraket{A_{S} | \pmb{\Nc} - \pmb{\Mc}^{\otimes 2} | B_2  B_2}$, where
	\begin{equation}
	B_2  B_2 = \frac{1}{d^2 - 1} \sum_{\sigma,\tau \in \Pc^*} \sigma\sigma\tau\tau \qquad \mbox{and} \qquad
	A_S = \sqrt{\frac{1}{2 (d^2-1)(d^2-2)}}  \sum_{\substack{\sigma, \tau  \in \Pc^* \\ \sigma \neq \tau}} \sigma\tau\sigma\tau + \sigma\tau\tau\sigma
	\end{equation}
	were defined in \eqref{eq:A0} and \eqref{eq:A12} respectively. 
	Therefore $a_S$ is computed as
	\begin{equation}
	\label{eq:319}
	\begin{split}
	a_S &= \bbraket{A_{S} | \pmb{\Nc} - \pmb{\Mc}^{\otimes 2} | B_2  B_2} \\
	&= \frac{1}{\sqrt{2} (d^2-1)^{\frac{3}{2}}(d^2-2)^\half}  \sum_{\substack{\sigma, \tau, \hat{\sigma}, \hat{\tau}  \in \Pc^* \\ \sigma \neq \tau}} \bbraket{\sigma\tau\sigma\tau + \sigma\tau\tau\sigma | \pmb{\Lambda}^{\otimes 4} | \hat{\sigma}\hat{\sigma}\hat{\tau}\hat{\tau}} \\
	&= \frac{1}{\sqrt{2} (d^2-1)^{\frac{3}{2}}(d^2-2)^\half}   \sum_{\substack{\sigma, \tau, \hat{\sigma}, \hat{\tau}  \in \Pc^* \\ \sigma \neq \tau}} 2\bbraket{\sigma | \pmb{\Lambda} | \hat{\sigma}} \bbraket{\tau | \pmb{\Lambda} | \hat{\sigma}} \bbraket{\sigma | \pmb{\Lambda} | \hat{\tau}} \bbraket{\tau | \pmb{\Lambda} | \hat{\tau}}   \\
	&=  \frac{\sqrt{2}}{(d^2-1)^{\frac{3}{2}}(d^2-2)^\half} \sum_{\substack{\sigma, \tau  \in \Pc^* \\ \sigma \neq \tau}} \bbraket{\sigma | \pmb{\Lambda}_\urm \pmb{\Lambda}_\urm^\dagger | {\tau}}^2 \\
	&= \frac{\sqrt{2}}{(d^2-1)^{\frac{3}{2}}(d^2-2)^\half} \sum_{\substack{\sigma, \tau  \in \Pc^* \\ \sigma \neq \tau}} \bbraket{\sigma | \mathbf{I} - \pmb{\Lambda}_\urm \pmb{\Lambda}_\urm^\dagger | {\tau}}^2.
	\end{split}
	\end{equation}
	In the fourth step, the trick of \eqref{eq:LLdagger_trick} was again used. In the final step, it is used that $\bbraket{\sigma | \pmb{\Lambda}_\urm \pmb{\Lambda}_\urm^\dagger | {\tau}}^2$ is the square of off-diagonal matrix elements of $\pmb{\Lambda}_\urm \pmb{\Lambda}_\urm^\dagger$, so that $\bbraket{\sigma | \pmb{\Lambda}_\urm \pmb{\Lambda}_\urm^\dagger | {\tau}}^2 = \bbraket{\sigma | \mathbf{I} - \pmb{\Lambda}_\urm \pmb{\Lambda}_\urm^\dagger | {\tau}}^2$.
	
	The bound is derived as follows. Under the stated assumption that $\Lambda$ is a single-qubit or unital channel, \autoref{lem:unital_part} guarantees that $\| {\Lambda}_\urm {\Lambda}_\urm^\dagger \|_{2\rightarrow 2} \leq 1$. Here $\| \cdot \|_{2\rightarrow 2}$ is the induced Schatten 2-norm (see \eqref{eq:induced_schatten_norms}). Since $\bbraket{A | B} = \Tr[A^\dagger B]$ for any $A,B \in \Lc(\Hc)$ (and therefore $\| A \|_2 = \| \kket{A} \|_2$ for all $A \in \Lc(\Hc)$), it follows that $\| \pmb{\Lambda}_\urm \pmb{\Lambda}_\urm^\dagger \|_{2\rightarrow 2} = \| {\Lambda}_\urm {\Lambda}_\urm^\dagger \|_{2\rightarrow 2}$. But the operator norm (Schatten $\infty$-norm) on matrices is just the induced $2 \rightarrow 2$ norm, so that it can be concluded that 
	$\| \pmb{\Lambda}_\urm \pmb{\Lambda}_\urm^\dagger \|_\infty =\| \pmb{\Lambda}_\urm \pmb{\Lambda}_\urm^\dagger \|_{2\rightarrow 2} = \| {\Lambda}_\urm {\Lambda}_\urm^\dagger \|_{2\rightarrow 2} \leq 1.$ 
	Together with the fact that a matrix of the form $\Lambdabf_\urm \Lambdabf_\urm^\dagger$ is itself positive semidefinite, this implies that the matrix $\mathbf{I} - \pmb{\Lambda}_\urm \pmb{\Lambda}_\urm^\dagger \geq 0$ is also positive semidefinite as a matrix (not to be confused with being a positive superoperator). 
	Now the key idea is to bound the off-diagonal elements of the symmetric positive semidefinite matrix $\mathbf{I} - \pmb{\Lambda}_\urm \pmb{\Lambda}_\urm^\dagger$ by the diagonal elements using the Sylvester's Criterion for positive semidefinite matrices (\autoref{lem:Sylvester}). This criterion states that a Hermitian matrix is positive semidefinite if and only if all of its principal minors are nonnegative. Here we use the only if part, since it has been established that $\mathbf{I} - \pmb{\Lambda}_\urm \pmb{\Lambda}_\urm^\dagger$ is positive semidefinite. In particular we use that the positive semidefiniteness of $\mathbf{I} - \pmb{\Lambda}_\urm \pmb{\Lambda}_\urm^\dagger$ implies that all of its second order minors are nonnegative. This means that
	\begin{equation}
	\bbraket{\sigma |\mathbf{I} - \pmb{\Lambda}_\urm \pmb{\Lambda}_\urm^\dagger | {\sigma}}\bbraket{\tau | \mathbf{I} - \pmb{\Lambda}_\urm \pmb{\Lambda}_\urm^\dagger | {\tau}} - \bbraket{\sigma | \mathbf{I} - \pmb{\Lambda}_\urm \pmb{\Lambda}_\urm^\dagger | {\tau}}^2 \geq 0, \quad \quad \forall \sigma, \tau \in \Pc^*, \, \sigma \neq \tau.
	\end{equation} 
	Plugging this into \eqref{eq:319} yields 
	\begin{equation}
	\label{eq:321}
	\begin{split}
	a_S &\leq \frac{\sqrt{2}}{(d^2-1)^{\frac{3}{2}}(d^2-2)^\half} \sum_{\substack{\sigma, \tau  \in \Pc^* \\ \sigma \neq \tau}} \bbraket{\sigma |\mathbf{I} - \pmb{\Lambda}_\urm \pmb{\Lambda}_\urm^\dagger | {\sigma}}\bbraket{\tau | \mathbf{I} - \pmb{\Lambda}_\urm \pmb{\Lambda}_\urm^\dagger | {\tau}} \\
	&= \frac{\sqrt{2}}{(d^2-1)^{\frac{3}{2}}(d^2-2)^\half} \left( \left(   \sum_{\sigma  \in \Pc^*} \bbraket{\sigma |\mathbf{I} - \pmb{\Lambda}_\urm \pmb{\Lambda}_\urm^\dagger | {\sigma}} \right)^2 - \sum_{\sigma  \in \Pc^*} \bbraket{\sigma |\mathbf{I} - \pmb{\Lambda}_\urm \pmb{\Lambda}_\urm^\dagger | {\sigma}}^2   \right). \\
	\end{split}
	\end{equation}
	The final step is to use that the mean of squares is larger than the square of the mean (\autoref{lem:meansquares}). This means in our setting that 
	\begin{equation}
	\frac{1}{d^2-1}\sum_{\sigma  \in \Pc^*} \bbraket{\sigma |\mathbf{I} - \pmb{\Lambda}_\urm \pmb{\Lambda}_\urm^\dagger | {\sigma}}^2 \geq \left(\frac{1}{d^2-1}\sum_{\sigma  \in \Pc^*} \bbraket{\sigma |\mathbf{I} - \pmb{\Lambda}_\urm \pmb{\Lambda}_\urm^\dagger | {\sigma}}\right)^2.
	\end{equation}
	Multiplying by $-(d^2-1)$ and plugging into \eqref{eq:321} yields the bound:
	\begin{align}
	a_S &\leq \frac{\sqrt{2}}{(d^2-1)^{\frac{3}{2}}(d^2-2)^\half} \left(1 - \frac{1}{d^2-1}\right)\left(   \sum_{\sigma  \in \Pc^*} \bbraket{\sigma |\mathbf{I} - \pmb{\Lambda}_\urm \pmb{\Lambda}_\urm^\dagger | {\sigma}} \right)^2  \\
	&= \frac{\sqrt{2}}{(d^2-1)^{\frac{3}{2}}(d^2-2)^\half} \frac{d^2-2}{d^2-1}\left(   (d^2-1)(1-u) \right)^2 \\
	&= \sqrt{\frac{d^2 - 2}{d^2-1}} \sqrt{2} (1-u)^2,
	\end{align}
	using the definition of $u$ (\eqref{eq:def_unitarity_app}) and the fact that $u(\Ic) = 1$.
\end{proof}
Finally, a bound on $a_\adj$ is presented.
\begin{proposition}[Bound on $a_\adj$]
	\label{prop:bound_aadj}
	Let $a_{\adj}$ be defined as in \eqref{eq:def_ai} and let $\Lambda$ be a CPTP map. If $\Lambda$ is a single-qubit channel (i.e., if $d = 2$) or if $\Lambda$ is unital [i.e., $\Lambda(I) = I$], then 
	\begin{equation}
	a_\adj = \frac{2}{(d^2-4)(d^2-1)^{\frac{3}{2}}} \sum_{\tau \in \Pc^*} \left( \sum_{\sigma \in C_\tau} \bbraket{\sigma \cdot \tau | \pmb{\Lambda}_\urm \pmb{\Lambda}_\urm^\dagger |\sigma} \right)^2 \leq \sqrt{d^2 - 1} (1-u)^2,
	\end{equation}
	where $C_\tau$ is the set of all normalized Pauli's that commute with $\tau$ (except for $\tau$ and $\sigma_0$), as defined in \eqref{eq:def_Ctau}.
\end{proposition}
\begin{proof}
	By the definition \eqref{eq:def_ai}, $a_{\adj} = \bbraket{A_{S} | \pmb{\Nc} - \pmb{\Mc}^{\otimes 2} | B_2  B_2}$, where
	\begin{equation}
	B_2  B_2 = \frac{1}{d^2 - 1} \left(\sum_{\hat{\sigma} \in \Pc^*} \hat{\sigma}\hat{\sigma}\right)^\tp{2} \qquad \mbox{and} \qquad
	A_\adj = \frac{1}{2(d^2-4)\sqrt{d^2 - 1}} \sum_{\tau \in \Pc^*} \left(\sum_{\sigma \in C_\tau} (\sigma \cdot \tau) \sigma  + \sigma (\sigma \cdot \tau)\right)^\tp{2}
	\end{equation}
	were defined in \eqref{eq:A0} and \eqref{eq:Aadj} respectively. 
	Therefore $a_\adj$ is computed as
	\begin{equation}
	\begin{split}
	a_\adj 	&= \frac{1}{2(d^2-4)(d^2-1)^{\frac{3}{2}}} \sum_{\tau \in \Pc^*}   \bbraket{ \left(\sum_{\sigma \in C_\tau} (\sigma \cdot \tau) \sigma  + \sigma (\sigma \cdot \tau)\right)^\tp{2} | \pmb{\Lambda}^\tp{4} | \left(\sum_{\hat{\sigma} \in \Pc^*} \hat{\sigma}\hat{\sigma} \right)^\tp{2}} \\
	&= \frac{1}{2(d^2-4)(d^2-1)^{\frac{3}{2}}} \sum_{\tau \in \Pc^*}   \left( \sum_{\sigma \in C_\tau} \sum_{\hat{\sigma} \in \Pc^*}  \bbraket{  (\sigma \cdot \tau) \sigma  + \sigma (\sigma \cdot \tau) | \pmb{\Lambda}^\tp{2} |  \hat{\sigma}\hat{\sigma} }\right)^2 \\
	&= \frac{1}{2(d^2-4)(d^2-1)^{\frac{3}{2}}} \sum_{\tau \in \Pc^*}   \left( \sum_{\sigma \in C_\tau} \sum_{\hat{\sigma} \in \Pc^*}  2 \bbraket{  \sigma \cdot \tau   | \pmb{\Lambda} |  \hat{\sigma} }   \bbraket{  \sigma    | \pmb{\Lambda} |  \hat{\sigma} }   \right)^2 \\
	&= \frac{2}{(d^2-4)(d^2-1)^{\frac{3}{2}}} \sum_{\tau \in \Pc^*}   \left( \sum_{\sigma \in C_\tau} \bbraket{  \sigma \cdot \tau   | \pmb{\Lambda}_\urm \Lambdabf_\urm^\dagger |  \sigma }  \right)^2, \\
	\end{split}
	\end{equation}
	where in the final line we used again the trick of \eqref{eq:LLdagger_trick}. Our bound on this quantity again starts with using the fact that the mean of the squares is larger than the square of the mean (\autoref{lem:meansquares}), yielding for all $\tau \in \Pc^*$
	\begin{equation}
	 \left( \frac{2}{d^2-4} \sum_{\sigma \in C_\tau} \bbraket{  \sigma \cdot \tau   | \pmb{\Lambda}_\urm \Lambdabf_\urm^\dagger |  \sigma } \right)^2 \leq  \frac{2}{d^2-4} \sum_{\sigma \in C_\tau} \bbraket{  \sigma \cdot \tau   | \pmb{\Lambda}_\urm \Lambdabf_\urm^\dagger |  \sigma }^2. 
	\end{equation}
	Multiplying with $\frac{d^2-4}{2}$ and plugging into the above yields 
	\begin{equation}
	a_S \leq \frac{1}{(d^2-1)^{\frac{3}{2}}} \sum_{\tau \in \Pc^*} \sum_{\sigma \in C_\tau} \bbraket{  \sigma \cdot \tau   | \pmb{\Lambda}_\urm \Lambdabf_\urm^\dagger |  \sigma }^2.
	\end{equation}
	Now we use the facts that $\sigma \cdot \tau \neq \sigma$  to write this as 
	\begin{equation}
	a_S \leq \frac{1}{(d^2-1)^{\frac{3}{2}}} \sum_{\tau \in \Pc^*} \sum_{\sigma \in C_\tau} \bbraket{  \sigma \cdot \tau   | \mathbf{I} - \pmb{\Lambda}_\urm \Lambdabf_\urm^\dagger |  \sigma }^2,
	\end{equation}
	where $\mathbf{I} - \pmb{\Lambda}_\urm \Lambdabf_\urm^\dagger$ is a positive semidefinite matrix, since $\| \Lambdabf_\urm \Lambdabf_\urm^\dagger \|_\infty = \| \Lambda_\urm \Lambda_\urm^\dagger \|_{2\rightarrow 2} \leq 1$ under the stated assumptions on $\Lambda$ by \autoref{lem:unital_part} and the fact that a matrix of the form $\Lambdabf_\urm \Lambdabf_\urm^\dagger$ is itself positive semidefinite. This allows us again to use Sylvester's criterion (\autoref{lem:Sylvester}) to bound off-diagonal terms by diagonal terms by using the fact that all minors of degree 2 of $\mathbf{I} - \pmb{\Lambda}_\urm \Lambdabf_\urm^\dagger$ must be nonnegative:
	\begin{equation}
	\bbraket{  \sigma   | \mathbf{I} - \pmb{\Lambda}_\urm \Lambdabf_\urm^\dagger |  \sigma }\bbraket{  \sigma \cdot \tau   | \mathbf{I} - \pmb{\Lambda}_\urm \Lambdabf_\urm^\dagger |  \sigma \cdot \tau  } - \bbraket{  \sigma \cdot \tau   | \mathbf{I} - \pmb{\Lambda}_\urm \Lambdabf_\urm^\dagger |  \sigma }^2 \geq 0, \quad\quad \forall \tau \in \Pc^*, \; \forall \sigma \in C_\tau.
	\end{equation}
	Therefore, we arrive at
	\begin{equation}
	\begin{split}
	a_S &\leq \frac{1}{(d^2-1)^{\frac{3}{2}}} \sum_{\tau \in \Pc^*} \sum_{\sigma \in C_\tau} \bbraket{  \sigma   | \mathbf{I} - \pmb{\Lambda}_\urm \Lambdabf_\urm^\dagger |  \sigma }\bbraket{  \sigma \cdot \tau   | \mathbf{I} - \pmb{\Lambda}_\urm \Lambdabf_\urm^\dagger |  \sigma \cdot \tau} \\
	&\leq \frac{1}{(d^2-1)^{\frac{3}{2}}} \sum_{\tau, \sigma \in \Pc^*} \bbraket{  \sigma   | \mathbf{I} - \pmb{\Lambda}_\urm \Lambdabf_\urm^\dagger |  \sigma }\bbraket{  \tau   | \mathbf{I} - \pmb{\Lambda}_\urm \Lambdabf_\urm^\dagger |   \tau} \\
	&= \sqrt{d^2-1}(1-u)^2,
	\end{split}
	\end{equation}
	where in the second line the sum over $\sigma \in C_\tau$ was completed to the sum over $\sigma \in \Pc^*$ by adding all the nonnegative terms $\bbraket{  \sigma   | \mathbf{I} - \pmb{\Lambda}_\urm \Lambdabf_\urm^\dagger |  \sigma }\bbraket{  \sigma \cdot \tau   | \mathbf{I} - \pmb{\Lambda}_\urm \Lambdabf_\urm^\dagger |  \sigma \cdot \tau}$ with $\sigma \in \Pc^* \setminus C_\tau$ for each $\tau \in \Pc^*$. All these terms are nonnegative because they are the product of diagonal elements of positive-semidefinite matrices, which must be nonnegative.
\end{proof}
This completes the set of propositions to bound the quantities $a_i$. The quantities $b_i$ are strongly related to the quantities $a_i$, and we will show that they satisfy the same upper bounds. More precisely, the next proposition establishes that all bounds on $a_i$ also hold for $b_i$, for $i \in \{  1,2 ; S ; 0 ; \adj  \}$.
\begin{proposition}[Bounds on $b_i$]
	\label{prop:bound_bi}
	Let $\Lambda$ be a CPTP map. Assume that $d = 2$ or that $\Lambda$ is unital. Let $a_i = \bbraket{A_i | \Ncbf - \Mcbf^\tp{2} | B_2 B_2}$ and $b_i = \bbraket{B_2B_2 | \Ncbf - \Mcbf^\tp{2} | A_i}$ as above. Then
	\begin{align}
	b_0 &= a_0 = 0, \\
	\label{eq:claim2}
	b_{1,2} &= a_{1,2} \leq \frac{\sqrt{d^2-2}}{d^2} (1-u)^2, \\
	\label{eq:claim3}
	b_S &= a_S \leq \sqrt{\frac{d^2 - 2}{d^2-1}} \sqrt{2} (1-u)^2, \\
	b_\adj &\leq \sqrt{d^2 - 1} (1-u)^2. 
	\end{align}
\end{proposition}
\begin{proof}
	The equality $b_0 = a_0 = \bbraket{B_2B_2 | \Ncbf - \Mcbf^\tp{2} | B_2 B_2}$ immediately follows from the fact that $A_0 = B_2B_2$. Thus $b_0 = 0$ by \autoref{prop:bound_a0}. In general, $b_i$ can be written as
	\begin{equation}
	b_i = \bbraket{B_2B_2 | \Ncbf - \Mcbf^\tp{2} | A_i} = \bbraket{A_i | \Ncbf^\dagger - (\Mcbf^\tp{2})^\dagger | B_2B_2}.
	\end{equation}
	Now since $\Gavg{n}$ are orthogonal projections, $(\Gavg{n})^\dagger = \Gavg{n}$. Therefore $\Nc^\dagger = \Gavg{4} (\Lambda^\dagger)^\tp{4} \Gavg{4}$ and $\Mc^\dagger = \Gavg{2} (\Lambda^\dagger)^\tp{2} \Gavg{2}$. Thus, $b_i$ and $a_i$ are related by $b_i(\Lambda) = a_i(\Lambda^\dagger)$. That is, $b_i$ can be obtained from $a_i$ by replacing $\Lambda$ with $\Lambda^\dagger$ in the exact expressions.
	
	We first show that this implies $b_{1,2} = a_{1,2}$ and $b_S = a_S$. This follows from the two identities (using only the trick of \eqref{eq:LLdagger_trick} over and over again)
	\begin{align}
	\label{eq:C28}
	\sum_{\sigma \in \Pc^*} \bbraket{\sigma | \pmb{\Lambda}_\urm \pmb{\Lambda}_\urm^\dagger | \sigma }^2 &=  \sum_{\sigma \in \Pc^*} \bbraket{\sigma | \pmb{\Lambda}_\urm | \hat{\sigma} }^2 \bbraket{\sigma| \pmb{\Lambda}_\urm | \hat{\sigma} }^2  = \sum_{\sigma \in \Pc^*} \bbraket{\hat{\sigma} | \pmb{\Lambda}_\urm^\dagger \pmb{\Lambda}_\urm | \hat{\sigma} }^2, \\
	\label{eq:C29}
	\sum_{\sigma,\tau \in \Pc^*} \bbraket{\sigma | \pmb{\Lambda}_\urm \pmb{\Lambda}_\urm^\dagger | \tau }^2 &=  \sum_{\sigma,\tau,\hat{\sigma},\hat{\tau} \in \Pc^*} \bbraket{\sigma | \pmb{\Lambda}_\urm | \hat{\sigma} } \bbraket{\tau| \pmb{\Lambda}_\urm | \hat{\sigma} } \bbraket{\sigma | \pmb{\Lambda}_\urm | \hat{\tau} } \bbraket{\tau| \pmb{\Lambda}_\urm | \hat{\tau} }  = \sum_{\hat{\sigma},\hat{\tau} \in \Pc^*} \bbraket{\hat{\sigma} | \pmb{\Lambda}_\urm^\dagger \pmb{\Lambda}_\urm | \hat{\tau} }^2.
	\end{align} 
	Now \eqref{eq:C28} implies that $a_{1,2}(\Lambda) = a_{1,2}(\Lambda^\dagger) = b_{1,2}(\Lambda)$. Subtracting \eqref{eq:C28} from \eqref{eq:C29} implies that $a_S(\Lambda) = a_S(\Lambda^\dagger) = b_S(\Lambda)$. This shows the second and third claim of this proposition (\eqref{eq:claim2} and \eqref{eq:claim3}), using the bounds and expressions for $a_{1,2}$ and $a_S$ from \autoref{prop:bound_a12} and \autoref{prop:bound_aS}
	
	For $b_\adj$ it is not clear that $b_\adj$ equals $a_\adj$. However, by copying the technique of the proof of \autoref{prop:bound_aadj} we show that the same bounds hold. Since $b_\adj(\Lambda) = a_\adj(\Lambda^\dagger)$, \autoref{prop:bound_aadj} implies that
	\begin{equation}
	b_\adj = \frac{2}{(d^2-4)(d^2-1)^{\frac{3}{2}}} \sum_{\tau \in \Pc^*}   \left( \sum_{\sigma \in C_\tau} \bbraket{  \sigma \cdot \tau   | \pmb{\Lambda}_\urm^\dagger \Lambdabf_\urm |  \sigma }  \right)^2.
	\end{equation}
	The bound is proven in exactly the same spirit as \autoref{prop:bound_aadj}. We first bound the square of the mean by the mean of the squares (\autoref{lem:meansquares}) and then use that  $\bbraket{  \sigma \cdot \tau   | \pmb{\Lambda}_\urm^\dagger \Lambdabf_\urm |  \sigma }^2 = \bbraket{  \sigma \cdot \tau   | \mathbf{I} - \pmb{\Lambda}_\urm^\dagger \Lambdabf_\urm |  \sigma }^2$ (since $\sigma \cdot \tau \neq \pm \sigma$). The matrix $\mathbf{I} - \pmb{\Lambda}_\urm^\dagger \Lambdabf_\urm$ is then shown to be positive semidefinite using  $\| \pmb{\Lambda}_\urm^\dagger \Lambdabf_\urm \|_\infty \leq 1$ (by the assumptions on $\Lambda$ and \autoref{lem:unital_part}) together with the fact that $\pmb{\Lambda}_\urm^\dagger \Lambdabf_\urm \geq 0$ is positive semidefinite. Thus Sylvester's criterion can be applied (\autoref{lem:Sylvester})Therefore
	\begin{equation}
	\begin{split}
	b_\adj &= \frac{2}{(d^2-4)(d^2-1)^{\frac{3}{2}}} \sum_{\tau \in \Pc^*}   \left( \sum_{\sigma \in C_\tau} \bbraket{  \sigma \cdot \tau   | \pmb{\Lambda}_\urm^\dagger \Lambdabf_\urm |  \sigma }  \right)^2 \\
	&\leq  \frac{1}{(d^2-1)^{\frac{3}{2}}} \sum_{\tau \in \Pc^*} \sum_{\sigma \in C_\tau} \bbraket{  \sigma \cdot \tau   | \pmb{\Lambda}_\urm^\dagger \Lambdabf_\urm |  \sigma }^2 \\
	&=  \frac{1}{(d^2-1)^{\frac{3}{2}}} \sum_{\tau \in \Pc^*} \sum_{\sigma \in C_\tau} \bbraket{  \sigma \cdot \tau   | \mathbf{I} - \pmb{\Lambda}_\urm^\dagger \Lambdabf_\urm |  \sigma }^2 \\
	&\leq  \frac{1}{(d^2-1)^{\frac{3}{2}}} \sum_{\tau \in \Pc^*} \sum_{\sigma \in C_\tau} \bbraket{  \sigma  | \mathbf{I} - \pmb{\Lambda}_\urm^\dagger \Lambdabf_\urm |  \sigma } \bbraket{  \sigma \cdot \tau   | \mathbf{I} - \pmb{\Lambda}_\urm^\dagger \Lambdabf_\urm |  \sigma \cdot \tau } \\
	&\leq  \frac{1}{(d^2-1)^{\frac{3}{2}}} \sum_{\tau, \sigma \in \Pc^*} \bbraket{  \sigma  | \mathbf{I} - \pmb{\Lambda}_\urm^\dagger \Lambdabf_\urm |  \sigma } \bbraket{  \tau   | \mathbf{I} - \pmb{\Lambda}_\urm^\dagger \Lambdabf_\urm |  \tau } \\
	&= \sqrt{d^2 - 1} (1-u)^2,
	\end{split}
	\end{equation}
	where in the last inequality the sum is completed by adding the nonnegative terms $\bbraket{  \sigma  | \mathbf{I} - \pmb{\Lambda}_\urm^\dagger \Lambdabf_\urm |  \sigma } \bbraket{  \tau   | \mathbf{I} - \pmb{\Lambda}_\urm^\dagger \Lambdabf_\urm |  \tau }$ for all $\tau \in \Pc^*$ and $\sigma \in \Pc^* \setminus C_\tau$. Note that this is the same bound as on $a_\adj$.
\end{proof}

Finally two more propositions are needed to bound the inner products in the expanded variance expression. The tool for this is the following. This proposition is formulated for any general CPTP map $\Ec$ and Hermitian operators $X,Y \in \Lc(\Hc)$. This theorem is applicable to inner products involving the map $\Nc^{m-s-1}$, since this is a CPTP map.
\begin{proposition}
	\label{prop:inner_holder}
	Let $\Ec$ be a CPTP map on a general Hilbert space $\Hc$. Then for any pair of Hermitian operators $X,Y \in \Lc(\Hc)$
	\begin{equation}
	\bbraket{X | \pmb{\Ec} | Y} \leq \| X \|_\infty \| Y \|_1.
	\end{equation}
\end{proposition}
\begin{proof}
	By Von Neumann's trace inequality and H\"olders inequality (\autoref{lem:NeumannHolder}) it follows that
	\begin{equation}
	\bbraket{X | \pmb{\Ec} | Y} = \Tr[X \Ec(Y) ] \leq \| X \|_\infty \| \Ec(Y) \|_1,
	\end{equation}
	using that $X$ and $\Ec(Y)$ are Hermitian. We then use the induced trace norm (the $1\rightarrow 1$ norm) and the fact that the map $\Ec$ is a CPTP map so that $\| \Ec \|_{1 \rightarrow 1} \leq 1$  (\autoref{lem:Perez-Garcia}). Therefore	
	\begin{equation}
	\| X \|_\infty \| \Ec(Y) \|_1 \leq \| X \|_\infty \| \Ec \|_{1 \rightarrow 1}  \| Y \|_1 \leq  \| X \|_\infty \| Y \|_1.
	\end{equation}
	Putting this together proves the bound.
\end{proof}
In order to apply the above proposition to the inner products occurring in the variance proof, a bound on the norms of the operators $A_i$ with $i \in \Zc_{TS}$ is needed.
\begin{proposition}[Norm bounds on $A_i$]
	\label{prop:bound_norms_Ai}
	Let $\{A_i: i \in \Zc_{TS} \}$ be defined as in \eqref{eq:A_i}. Then for $d \geq 4$ the following bounds hold
	\begin{equation}
	\| A_i \|_1 \leq d^2 \quad\quad\mbox{and}\quad\quad \|A_i \|_\infty \leq \sqrt{\frac{6}{(d-2)(d-1)}}, \quad\quad \forall i \in \Zc_{TS}.
	\end{equation}
	If $d = 2$, then $\Zc_{TS} = \{ S;  \, 1,2 \}$, and
	\begin{align}
	\| A_S \|_1     &= \frac{5}{\sqrt{3}}, & \| A_S \|_\infty     &= \frac{1}{\sqrt{3}}, \\
	\| A_{1,2} \|_1 &= 2\sqrt{2},          & \| A_{1,2} \|_\infty &= \frac{\sqrt{2}}{3}. 
	\end{align}
\end{proposition}
\begin{proof}
	For the $d=2$ case,  the norms can be computed directly, since $A_S$ and $A_{1,2}$ are explicitly defined in \eqref{eq:A12}-\eqref{eq:AS}. By direct computation the result follows. For $d \geq 4$, the trace norm bound is trivial, since
	\begin{equation}
	\| A_i \|_1 \leq \sqrt{d^4} \| A_i \|_2 = d^2,
	\end{equation}
	by H\"older's inequality. The last equality uses the fact that $A_i$ are Hilbert-Schmidt normalized ($\| A_i \|_2 = 1$). The effort of the proof is in the bound on $\| A_i \|_\infty$. 
	
	The proof of this statement uses the description of the tensor-2 Liouville representation of \cite{Zhu2016} over \cite{Helsen2018}, since their description is basis-free. Ref. \cite{Zhu2016} considers the action of the Clifford group $\Cliff{d}$ on $\Hc^{\otimes 4}$. The representation $\Hc^{\otimes 4}$ of the Clifford group $\Cliff{d}$ decomposes as
	\begin{equation}
	\label{eq:C45}
	\Hc^{\otimes 4} = \bigoplus_{k} W_k \otimes \C^{d_k}
	\end{equation}
	where $ W_k $ are irreducible, pairwise inequivalent representations of the Clifford group that occur with multiplicity $d_k$. Here $k$ is just an index for the irreducible, inequivalent representations. Descriptions of these spaces and explicit expressions for their dimensions are given in \cite{Zhu2016} (there the index $k$ runs over Young Diagrams $\lambda$ and signs $s$). We will show that 
	\begin{equation}
	\| A_i \|_\infty \leq \max_k \frac{1}{\sqrt{| W_k |}}.
	\end{equation}
	Since the dimensions of all $W_k$ are given, the maximization can easily be done.
	
	Using the intertwining isomorphism $\Lc(\Hc) \simeq \Hc \otimes \Hc^*$ the tensor-4 Liouville representation on $\Lc(\Hc^\tp{4})$ can be written in terms of the decomposition \eqref{eq:C45}:
	\begin{equation}
	\Lc(\Hc^{\otimes 4}) = \bigoplus_{k,l} \Lc(W_l, W_k) \otimes \Lc(\C^{d_l}, \C^{d_k}).
	\end{equation}
	In principle $\Lc(W_l, W_k)$ are not irreducible representations. However, only the trivial subrepresentations of $\Lc(W_l, W_k)$ (denoted $(\Lc(W_l, W_k))^\Cliff{d}$) are relevant, since
	\begin{equation}
	(\Lc(\Hc^{\otimes 4}))^\Cliff{d} = \bigoplus_{k,l} (\Lc(W_l, W_k))^\Cliff{d} \otimes \Lc(\C^{d_l}, \C^{d_k}).
	\end{equation}
	The key point is that every element $\varphi \in (\Lc(W_l, W_k))^\Cliff{d}$ is an intertwining operator between the representations $W_k$ and $W_l$ \cite{Fulton2004}. By Schur's Lemma \cite{Fulton2004} and the fact that $W_k$ are mutually inequivalent irreducible representations it follows that $\varphi \propto \delta_{k,l} I_{W_k}$. Therefore
	\begin{equation}
	\label{eq:45}
	(\Lc(\Hc^{\otimes 4}))^\Cliff{d} = \bigoplus_{k} \Span\{ I_{W_k} \} \otimes \Lc(\C^{d_k}).
	\end{equation}
	This description provides a simple orthogonal basis for the space $(\Lc(\Hc^{\otimes 4}))^\Cliff{d}$, namely
	\begin{equation}
	\mathcal{A} = \{ P_{W_k} \otimes E_{m,n} | k;  m,n = 1,...,d_k  \},
	\end{equation}
	where $P_{W_k}$ is the orthogonal projection onto ${W_k}$ and $\{E_{m,n} | m,n = 1,...,d_k\}$ is the canonical (or any other) orthonormal basis of $\Lc(\C^{d_k})$. Normalizing with respect to the Hilbert-Schmidt norm yields the orthonormal basis operators
	\begin{equation}
	A_{k,m,n} = \frac{P_{W_k}}{\sqrt{|W_k|}}  \otimes E_{m,n}.
	\end{equation}
	Note that our basis operators $\{A_i: i \in \Zc_{TS} \}$ might be different than these $A_{k,m,n}$. However, these $A_i$ also span trivial subrepresentations of $\Lc(\Hc^\tp{4})$, so $A_i \in \Lc(\Hc^\tp{4})^\Cliff{d}$. We now show that $\| A \| \leq \max_k |W_k|^{-\half}$ for all $A \in \Lc(\Hc^\tp{4})^\Cliff{d}$ such that $\|A\|_2 = 1$. Therefore this bound holds in particular for our $A_i$ of interest. To do so, $A$ is written in the basis $\mathcal{A}$ as
	\begin{equation}
	A = \sum_{k} \sum_{m,n = 1}^{d_k} \alpha_{k,m,n}  A_{k,m,n}, \qquad \mbox{s.t.} \qquad
	\sum_{k} \sum_{m,n = 1}^{d_k} |\alpha_{k,m,n}|^2  = 1.
	\end{equation}
	Now we use that the operator $A \in (\Lc(\Hc^{\otimes 4}))^\Cliff{d}$ is block diagonal with respect to the spaces $\Span\{ I_{W_k} \} \otimes \Lc(\C^{d_k})$ (see \eqref{eq:45}). Therefore the infinity norm can be computed as the maximum over $k$ of the infinity norm of $A$ restricted to $\Span\{ I_{W_k} \} \otimes \Lc(\C^{d_k})$, yielding
	\begin{equation}
	\| A \|_\infty = \left\| \sum_{k} \sum_{m,n = 1}^{d_k} \alpha_{k,m,n}  A_{k,m,n} \right\|_\infty =\max_{k}   \left\| \sum_{m,n = 1}^{d_k} \alpha_{k,m,n}  A_{k,m,n} \right\|_\infty = \max_{k} \left\| \frac{P_{W_k}}{\sqrt{|W_k|}}    \otimes \sum_{m,n = 1}^{d_k} \alpha_{k,m,n}   E_{m,n} \right\|_\infty.
	\end{equation}
	Using some basic properties of the Schatten $p$-norms, this is bounded as follows
	\begin{equation}
	\| A \|_\infty = \max_{k} \left\| \frac{P_{W_k}}{\sqrt{|W_k|}} \right\|_\infty \left\| \sum_{m,n = 1}^{d_k} \alpha_{k,m,n}   E_{m,n} \right\|_\infty = \max_{k} \frac{\left\| P_{W_k} \right\|_\infty}{\sqrt{|W_k|}}  \left\| \sum_{m,n = 1}^{d_k} \alpha_{k,m,n}   E_{m,n} \right\|_\infty \leq  \max_k \frac{1}{\sqrt{|W_k|}},  
	\end{equation}
	using that $\| P_{W_k} \|_\infty = 1$ and 
	\begin{equation}
	\left\| \sum_{m,n = 1}^{d_k} \alpha_{k,m,n}   E_{m,n} \right\|_\infty \leq  \left\| \sum_{m,n = 1}^{d_k} \alpha_{k,m,n}   E_{m,n} \right\|_2 \leq \sum_{m,n = 1}^{d_k} |\alpha_{k,m,n}|^2   \left\| E_{m,n} \right\|_2 = 1.
	\end{equation}
	By Lemma 1 of \cite{Zhu2016}, which gives all dimensions $|W_k|$, it follows that \begin{equation}
	\| A \|_\infty \leq \max_{k} \frac{1}{\sqrt{|W_k|}}  = \sqrt{\frac{6}{(d-1)(d-2)}},
	\end{equation}
	provided that $d = 2^q \geq 4$, $q \in \N$. This proves the last bound.
\end{proof}
Finally, there is one inner product in the proof of \autoref{thm:var_bound} for which a sharper bound can be found than using \autoref{prop:inner_holder} and \autoref{prop:bound_norms_Ai}. This sharper bound is given in the following proposition.

\begin{proposition}
	\label{prop:ideal_inner}
	Let $\Nc$ be defined as in \eqref{eq:N_def}, with $\Lambda$ a single-qubit or unital quantum channel. Then for any $m \in \N$ the following bound holds
	\begin{equation}
	\bbraket{A_i | \Ncbf^m | B_2 B_2} \leq \frac{1}{\sqrt{|V_i|}}, \quad\quad \forall i \in \Zc_{TS}.
	\end{equation}
\end{proposition}
\begin{proof}
	Slightly rewriting the inner product yields 
	\begin{equation}
	\bbraket{A_i | \Ncbf^m | B_2 B_2} = \bbraket{A_i | \Nc^m(B_2 B_2)}.
	\end{equation}
	From the definition of $\Nc$ \eqref{eq:N_def} it follows that 
	\begin{equation}
	\Nc^m(B_2 B_2) = \frac{1}{|\Cliff{d}|^{m}} \sum_{\jbf} \Gc_\jbf^\tp{4}  (B_2 B_2) = \frac{1}{|\Cliff{d}|^{m}} \sum_{\jbf} [\Gc_\jbf^\tp{2}  (B_2) ]^\tp{2},
	\end{equation}
	where the sum is over all noisy sequences of length $m$ indexed by $\jbf$ (such that $\jbf$ is a multi-index of length $m$).
	We will show that $\| \Gc_\jbf^\tp{2}  (B_2) \|_2 \leq 1$.  we treat the multiqubit and single-qubit case separately. In the multiqubit case, we have
	\begin{equation}
	\label{eq:C58}
	\| \Gc_\jbf^\tp{2}  (B_2) \|_2 \leq \| \Gc_\jbf^\tp{2}\|_{2 \rightarrow 2}  \| B_2 \|_2 = \| \Gc_\jbf \|_{2 \rightarrow 2}^2.
	\end{equation}
	The inequality follows from the definition of the induced Schatten norms (see \eqref{eq:induced_schatten_norms}). The equality is due to the fact that $\| B_2 \|_2 = 1$ is normalized. Under the assumption that $\Lambda$ is unital, the entire sequence $\Gc_\jbf$ is unital. Therefore by \autoref{lem:Perez-Garcia} (Perez-Garcia), $\| \Gc_\jbf \|_{2 \rightarrow 2}^2 \leq 1$. This shows that $\| \Gc_\jbf^\tp{2}  (B_2) \|_2 \leq 1$.
		
 	In case of a single-qubit, nonunital error channels $\Lambda$, some extra care must be taken. Let us denote $\Lc(\Hc)^H := \{ A \in \Lc(\Hc) :  \Tr[A] = 0, A = A^\dagger \} = \Span_{\R} \{ \sigma : \sigma \in \Pc^* \}$ as the traceless Hermitian subspace of $\Lc(\Hc)$. 
	This space is a vector space over $\R$, with an orthonormal basis $\Pc^*$. Since $\Gc_\jbf$ is positive (and thus maps Hermitian operators to Hermitian operators) and trace-preserving, it maps the traceless Hermitian subspace $\Lc(\Hc)^H$ to itself. Observe that $B_2 \in (\Lc(\Hc)^H)^\tp{2}$. Therefore restrict $\Gc_\jbf^\tp{2}$ to $(\Lc(\Hc)^H)^\tp{2}$. This results in
	\begin{equation}
	\| \Gc_\jbf^\tp{2}  (B_2) \|_2 = \| \Gc_\jbf^\tp{2} \big|_{(\Lc(\Hc)^H)^\tp{2}}  (B_2) \|_2 \leq \| \Gc_\jbf^\tp{2} \big|_{(\Lc(\Hc)^H)^\tp{2}}  \|_{2\rightarrow 2}  \|  B_2 \|_2 =  \| \Gc_\jbf \big|_{\Lc(\Hc)^H}  \|_{2\rightarrow 2}^2.
	\end{equation}
	The first equality is the restriction of $\Gc_\jbf$ to the traceless Hermitian subspace. The inequality follows from the definition of the induced Schatten norm (\eqref{eq:induced_schatten_norms}). The final equality is due to the fact that $\| B_2 \|_2 = 1$. The key point of restricting to the traceless Hermitian subspace $\| \Gc_\jbf \big|_{\Lc(\Hc)^H}  \|_{2\rightarrow 2}$ allows for the application of statement \eqref{eq:A25} of \autoref{lem:Perez-Garcia} (Perez-Garcia). By the lemma (where $\| \Gc_\jbf \big|_{\Lc(\Hc)^H}  \|_{2\rightarrow 2}$ is denoted $\| \Gc_\jbf \|_{2\rightarrow 2}^H$), we have
	\begin{equation}
	\| \Gc_\jbf \big|_{\Lc(\Hc)^H}  \|_{2\rightarrow 2} \leq \sqrt{\frac{d}{2}},
	\end{equation}
	which in the single-qubit case means $\| \Gc_\jbf \big|_{\Lc(\Hc)^H}  \|_{2\rightarrow 2} \leq 1$. Therefore, we also have $\| \Gc_\jbf^\tp{2}  (B_2) \|_2 \leq 1$ in the single-qubit, nonunital case.
	
	We have thus established that $\| \Gc_\jbf^\tp{2}  (B_2) \|_2 \leq 1$ for single-qubit or unital noise maps $\Lambda$. Therefore, the following upper bound is valid
	\begin{equation}
	\begin{split}
	\bbraket{A_i | \Ncbf^m | B_2 B_2} &=  \frac{1}{|\Cliff{d}|^{m}} \sum_{\jbf  } \bbraket{ A_i | [\Gc_\jbf^\tp{2}  (B_2) ]^\tp{2}} \\
	&\leq \frac{1}{|\Cliff{d}|^{m}} \sum_{\jbf  } \max_{\substack{Q \in (\Lc(\Hc)^H)^\tp{2} \\ \| Q \|_2 \leq 1}} \bbraket{A_i | Q^\tp{2}} \\ 
	&= \max_{\substack{Q \in (\Lc(\Hc)^H)^\tp{2} \\ \| Q \|_2 \leq 1}} \bbraket{A_i | Q^\tp{2}}.
	\end{split}
	\end{equation}
	In the second line, we have replaced the particular operator $\Gc_\jbf^\tp{2}  (B_2) \in (\Lc(\Hc)^H)^\tp{2}$ which satisfies $\| \Gc_\jbf^\tp{2}  (B_2) \|_2 \leq 1$ with the maximization over all operators $Q \in (\Lc(\Hc)^H)^\tp{2}$ that satisfy $\| Q \|_2 \leq 1$.
	To continue, we use the definition of $A_i$ (\eqref{eq:A_i}), which is given by
	\begin{equation}
	A_i = \frac{1}{\sqrt{|V_i|}}\sum_{s=1}^{|V_i|} v^{(i)}_s v^{(i)}_s, \quad \forall i \in \Zc_{TS},
	\end{equation}
	where $\{ v^{(i)}_s \}$ is an orthonormal basis of $V_i \subset (\Lc(\Hc)^H)^\tp{2}$. Let us expand $Q$ in this basis,
	\begin{equation}
	\label{eq:C63}
	Q = q_\perp v_\perp^{(i)} + \sum_{s=1}^{|V_i|} q_s v^{(i)}_s  \quad\quad \mbox{s.t.} \quad\quad |q_\perp|^2 + \sum_{s=1}^{|V_i|} |q_{s}|^2 \leq 1, \quad q_\perp, q_s \in \C,  \; \forall s=1,...,|V_i|.
	\end{equation}
	Here $q_\perp v_\perp^{(i)}$ is the component of $Q$ in the space orthogonal to $V_i$, i.e., $q_\perp v_\perp^{(i)} \in (\Lc(\Hc)^H)^\tp{2} \setminus V_i$. The condition on $q_\perp$ and the $q_s$ follow from the requirement that $\| Q \|_2 \leq 1$. Actually, there are additional constraints on $q_\perp$ and the $q_s$ needed to ensure that $Q$ is traceless and Hermitian, but these constraints are not necessary to prove the result.
	Using the expansion \eqref{eq:C63} it follows that 
	\begin{equation}
	\max_{\substack{Q \in (\Lc(\Hc)^H)^\tp{2} \\ \| Q \|_2 \leq 1}} \bbraket{A_i | Q^\tp{2}} 
	\leq \max_{ \substack{\{ q_s \} \\ \sum_{s} |q_{s}|^2 \leq 1} } \frac{1}{\sqrt{|V_i|}}  \sum_{s,t,k = 1}^{|V_i|} |q_s q_t| |\bbraket{v^{(i)}_k  v^{(i)}_k | v^{(i)}_s v^{(i)}_t}| = \max_{ \substack{\{ q_s \} \\ \sum_{s} |q_{s}|^2 \leq 1} }  \frac{1}{\sqrt{|V_i|}}  \sum_{k=1}^{|V_i|} |q_k|^2 \leq \frac{1}{\sqrt{|V_i|}},
	\end{equation}
	using the fact that $\bbraket{v^{(i)}_k  v^{(i)}_k | v^{(i)}_s v^{(i)}_t} = \delta_{sk}\delta_{tk}$ by orthonormality of the basis. This completes the proof.
\end{proof}

%----------------------------------------------------------------------------------------
%	REFERENCE LIST
%----------------------------------------------------------------------------------------
\twocolumngrid
\bibliography{RB}

\begin{thebibliography}{55}%
\makeatletter
\providecommand \@ifxundefined [1]{%
 \@ifx{#1\undefined}
}%
\providecommand \@ifnum [1]{%
 \ifnum #1\expandafter \@firstoftwo
 \else \expandafter \@secondoftwo
 \fi
}%
\providecommand \@ifx [1]{%
 \ifx #1\expandafter \@firstoftwo
 \else \expandafter \@secondoftwo
 \fi
}%
\providecommand \natexlab [1]{#1}%
\providecommand \enquote  [1]{``#1''}%
\providecommand \bibnamefont  [1]{#1}%
\providecommand \bibfnamefont [1]{#1}%
\providecommand \citenamefont [1]{#1}%
\providecommand \href@noop [0]{\@secondoftwo}%
\providecommand \href [0]{\begingroup \@sanitize@url \@href}%
\providecommand \@href[1]{\@@startlink{#1}\@@href}%
\providecommand \@@href[1]{\endgroup#1\@@endlink}%
\providecommand \@sanitize@url [0]{\catcode `\\12\catcode `\$12\catcode
  `\&12\catcode `\#12\catcode `\^12\catcode `\_12\catcode `\%12\relax}%
\providecommand \@@startlink[1]{}%
\providecommand \@@endlink[0]{}%
\providecommand \url  [0]{\begingroup\@sanitize@url \@url }%
\providecommand \@url [1]{\endgroup\@href {#1}{\urlprefix }}%
\providecommand \urlprefix  [0]{URL }%
\providecommand \Eprint [0]{\href }%
\providecommand \doibase [0]{http://dx.doi.org/}%
\providecommand \selectlanguage [0]{\@gobble}%
\providecommand \bibinfo  [0]{\@secondoftwo}%
\providecommand \bibfield  [0]{\@secondoftwo}%
\providecommand \translation [1]{[#1]}%
\providecommand \BibitemOpen [0]{}%
\providecommand \bibitemStop [0]{}%
\providecommand \bibitemNoStop [0]{.\EOS\space}%
\providecommand \EOS [0]{\spacefactor3000\relax}%
\providecommand \BibitemShut  [1]{\csname bibitem#1\endcsname}%
\let\auto@bib@innerbib\@empty
%</preamble>
\bibitem [{\citenamefont {Emerson}\ \emph {et~al.}(2005)\citenamefont
  {Emerson}, \citenamefont {Alicki},\ and\ \citenamefont
  {Zyczkowski}}]{Emerson2005}%
  \BibitemOpen
  \bibfield  {author} {\bibinfo {author} {\bibfnamefont {Joseph}\ \bibnamefont
  {Emerson}}, \bibinfo {author} {\bibfnamefont {Robert}\ \bibnamefont
  {Alicki}}, \ and\ \bibinfo {author} {\bibfnamefont {Karol}\ \bibnamefont
  {Zyczkowski}},\ }\bibfield  {title} {\enquote {\bibinfo {title} {{Scalable
  noise estimation with random unitary operators}},}\ }\href {\doibase
  10.1088/1464-4266/7/10/021} {\bibfield  {journal} {\bibinfo  {journal} {J.
  Opt. B}\ }\textbf {\bibinfo {volume} {7}},\ \bibinfo {pages} {347--352}
  (\bibinfo {year} {2005})}\BibitemShut {NoStop}%
\bibitem [{\citenamefont {Knill}\ \emph {et~al.}(2008)\citenamefont {Knill},
  \citenamefont {Leibfried}, \citenamefont {Reichle}, \citenamefont {Britton},
  \citenamefont {Blakestad}, \citenamefont {Jost}, \citenamefont {Langer},
  \citenamefont {Ozeri}, \citenamefont {Seidelin},\ and\ \citenamefont
  {Wineland}}]{Knill2008}%
  \BibitemOpen
  \bibfield  {author} {\bibinfo {author} {\bibfnamefont {E.}~\bibnamefont
  {Knill}}, \bibinfo {author} {\bibfnamefont {D.}~\bibnamefont {Leibfried}},
  \bibinfo {author} {\bibfnamefont {R.}~\bibnamefont {Reichle}}, \bibinfo
  {author} {\bibfnamefont {J.}~\bibnamefont {Britton}}, \bibinfo {author}
  {\bibfnamefont {R.~B.}\ \bibnamefont {Blakestad}}, \bibinfo {author}
  {\bibfnamefont {J.~D.}\ \bibnamefont {Jost}}, \bibinfo {author}
  {\bibfnamefont {C.}~\bibnamefont {Langer}}, \bibinfo {author} {\bibfnamefont
  {R.}~\bibnamefont {Ozeri}}, \bibinfo {author} {\bibfnamefont
  {S.}~\bibnamefont {Seidelin}}, \ and\ \bibinfo {author} {\bibfnamefont
  {D.~J.}\ \bibnamefont {Wineland}},\ }\bibfield  {title} {\enquote {\bibinfo
  {title} {{Randomized benchmarking of quantum gates}},}\ }\href {\doibase
  10.1103/PhysRevA.77.012307} {\bibfield  {journal} {\bibinfo  {journal} {Phys.
  Rev. A}\ }\textbf {\bibinfo {volume} {77}},\ \bibinfo {pages} {012307}
  (\bibinfo {year} {2008})}\BibitemShut {NoStop}%
\bibitem [{\citenamefont {Magesan}\ \emph {et~al.}(2011)\citenamefont
  {Magesan}, \citenamefont {Gambetta},\ and\ \citenamefont
  {Emerson}}]{Magesan2011}%
  \BibitemOpen
  \bibfield  {author} {\bibinfo {author} {\bibfnamefont {Easwar}\ \bibnamefont
  {Magesan}}, \bibinfo {author} {\bibfnamefont {J~M}\ \bibnamefont {Gambetta}},
  \ and\ \bibinfo {author} {\bibfnamefont {Joseph}\ \bibnamefont {Emerson}},\
  }\bibfield  {title} {\enquote {\bibinfo {title} {{Scalable and Robust
  Randomized Benchmarking of Quantum Processes}},}\ }\href {\doibase
  10.1103/PhysRevLett.106.180504} {\bibfield  {journal} {\bibinfo  {journal}
  {Phys. Rev. Lett.}\ }\textbf {\bibinfo {volume} {106}},\ \bibinfo {pages}
  {180504} (\bibinfo {year} {2011})}\BibitemShut {NoStop}%
\bibitem [{\citenamefont {Magesan}\ \emph
  {et~al.}(2012{\natexlab{a}})\citenamefont {Magesan}, \citenamefont
  {Gambetta},\ and\ \citenamefont {Emerson}}]{Magesan2012a}%
  \BibitemOpen
  \bibfield  {author} {\bibinfo {author} {\bibfnamefont {Easwar}\ \bibnamefont
  {Magesan}}, \bibinfo {author} {\bibfnamefont {Jay~M.}\ \bibnamefont
  {Gambetta}}, \ and\ \bibinfo {author} {\bibfnamefont {Joseph}\ \bibnamefont
  {Emerson}},\ }\bibfield  {title} {\enquote {\bibinfo {title} {{Characterizing
  quantum gates via randomized benchmarking}},}\ }\href {\doibase
  10.1103/PhysRevA.85.042311} {\bibfield  {journal} {\bibinfo  {journal} {Phys.
  Rev. A}\ }\textbf {\bibinfo {volume} {85}},\ \bibinfo {pages} {042311}
  (\bibinfo {year} {2012}{\natexlab{a}})}\BibitemShut {NoStop}%
\bibitem [{\citenamefont {Chow}\ \emph {et~al.}(2010)\citenamefont {Chow},
  \citenamefont {DiCarlo}, \citenamefont {Gambetta}, \citenamefont {Motzoi},
  \citenamefont {Frunzio}, \citenamefont {Girvin},\ and\ \citenamefont
  {Schoelkopf}}]{Chow2010}%
  \BibitemOpen
  \bibfield  {author} {\bibinfo {author} {\bibfnamefont {J.~M.}\ \bibnamefont
  {Chow}}, \bibinfo {author} {\bibfnamefont {L.}~\bibnamefont {DiCarlo}},
  \bibinfo {author} {\bibfnamefont {J.~M.}\ \bibnamefont {Gambetta}}, \bibinfo
  {author} {\bibfnamefont {F.}~\bibnamefont {Motzoi}}, \bibinfo {author}
  {\bibfnamefont {L.}~\bibnamefont {Frunzio}}, \bibinfo {author} {\bibfnamefont
  {S.~M.}\ \bibnamefont {Girvin}}, \ and\ \bibinfo {author} {\bibfnamefont
  {R.~J.}\ \bibnamefont {Schoelkopf}},\ }\bibfield  {title} {\enquote {\bibinfo
  {title} {{Optimized driving of superconducting artificial atoms for improved
  single-qubit gates}},}\ }\href {\doibase 10.1103/PhysRevA.82.040305}
  {\bibfield  {journal} {\bibinfo  {journal} {Phys. Rev. A}\ }\textbf {\bibinfo
  {volume} {82}},\ \bibinfo {pages} {040305} (\bibinfo {year}
  {2010})}\BibitemShut {NoStop}%
\bibitem [{\citenamefont {Olmschenk}\ \emph {et~al.}(2010)\citenamefont
  {Olmschenk}, \citenamefont {Chicireanu}, \citenamefont {Nelson},\ and\
  \citenamefont {Porto}}]{Olmschenk2010}%
  \BibitemOpen
  \bibfield  {author} {\bibinfo {author} {\bibfnamefont {S.}~\bibnamefont
  {Olmschenk}}, \bibinfo {author} {\bibfnamefont {R.}~\bibnamefont
  {Chicireanu}}, \bibinfo {author} {\bibfnamefont {K.~D.}\ \bibnamefont
  {Nelson}}, \ and\ \bibinfo {author} {\bibfnamefont {J.~V.}\ \bibnamefont
  {Porto}},\ }\bibfield  {title} {\enquote {\bibinfo {title} {{Randomized
  benchmarking of atomic qubits in an optical lattice}},}\ }\href {\doibase
  10.1088/1367-2630/12/11/113007} {\bibfield  {journal} {\bibinfo  {journal}
  {New J. Phys.}\ }\textbf {\bibinfo {volume} {12}},\ \bibinfo {pages} {113007}
  (\bibinfo {year} {2010})}\BibitemShut {NoStop}%
\bibitem [{\citenamefont {Gaebler}\ \emph {et~al.}(2012)\citenamefont
  {Gaebler}, \citenamefont {Meier}, \citenamefont {Tan}, \citenamefont
  {Bowler}, \citenamefont {Lin}, \citenamefont {Hanneke}, \citenamefont {Jost},
  \citenamefont {Home}, \citenamefont {Knill}, \citenamefont {Leibfried},\ and\
  \citenamefont {Wineland}}]{Gaebler2012}%
  \BibitemOpen
  \bibfield  {author} {\bibinfo {author} {\bibfnamefont {J.~P.}\ \bibnamefont
  {Gaebler}}, \bibinfo {author} {\bibfnamefont {A.~M.}\ \bibnamefont {Meier}},
  \bibinfo {author} {\bibfnamefont {T.~R.}\ \bibnamefont {Tan}}, \bibinfo
  {author} {\bibfnamefont {R.}~\bibnamefont {Bowler}}, \bibinfo {author}
  {\bibfnamefont {Y.}~\bibnamefont {Lin}}, \bibinfo {author} {\bibfnamefont
  {D.}~\bibnamefont {Hanneke}}, \bibinfo {author} {\bibfnamefont {J.~D.}\
  \bibnamefont {Jost}}, \bibinfo {author} {\bibfnamefont {J.~P.}\ \bibnamefont
  {Home}}, \bibinfo {author} {\bibfnamefont {E.}~\bibnamefont {Knill}},
  \bibinfo {author} {\bibfnamefont {D.}~\bibnamefont {Leibfried}}, \ and\
  \bibinfo {author} {\bibfnamefont {D.~J.}\ \bibnamefont {Wineland}},\
  }\bibfield  {title} {\enquote {\bibinfo {title} {{Randomized benchmarking of
  multiqubit gates}},}\ }\href {\doibase 10.1103/PhysRevLett.108.260503}
  {\bibfield  {journal} {\bibinfo  {journal} {Phys. Rev. Lett.}\ }\textbf
  {\bibinfo {volume} {108}},\ \bibinfo {pages} {260503} (\bibinfo {year}
  {2012})}\BibitemShut {NoStop}%
\bibitem [{\citenamefont {Barends}\ \emph {et~al.}(2014)\citenamefont
  {Barends}, \citenamefont {Kelly}, \citenamefont {Megrant}, \citenamefont
  {Veitia}, \citenamefont {Sank}, \citenamefont {Jeffrey}, \citenamefont
  {White}, \citenamefont {Mutus}, \citenamefont {Fowler}, \citenamefont
  {Campbell}, \citenamefont {Chen}, \citenamefont {Chen}, \citenamefont
  {Chiaro}, \citenamefont {Dunsworth}, \citenamefont {Neill}, \citenamefont
  {O'Malley}, \citenamefont {Roushan}, \citenamefont {Vainsencher},
  \citenamefont {Wenner}, \citenamefont {Korotkov}, \citenamefont {Cleland},\
  and\ \citenamefont {Martinis}}]{Barends2014}%
  \BibitemOpen
  \bibfield  {author} {\bibinfo {author} {\bibfnamefont {R.}~\bibnamefont
  {Barends}}, \bibinfo {author} {\bibfnamefont {J.}~\bibnamefont {Kelly}},
  \bibinfo {author} {\bibfnamefont {A.}~\bibnamefont {Megrant}}, \bibinfo
  {author} {\bibfnamefont {A.}~\bibnamefont {Veitia}}, \bibinfo {author}
  {\bibfnamefont {D.}~\bibnamefont {Sank}}, \bibinfo {author} {\bibfnamefont
  {E.}~\bibnamefont {Jeffrey}}, \bibinfo {author} {\bibfnamefont {T.~C.}\
  \bibnamefont {White}}, \bibinfo {author} {\bibfnamefont {J.}~\bibnamefont
  {Mutus}}, \bibinfo {author} {\bibfnamefont {A.~G.}\ \bibnamefont {Fowler}},
  \bibinfo {author} {\bibfnamefont {B.}~\bibnamefont {Campbell}}, \bibinfo
  {author} {\bibfnamefont {Y.}~\bibnamefont {Chen}}, \bibinfo {author}
  {\bibfnamefont {Z.}~\bibnamefont {Chen}}, \bibinfo {author} {\bibfnamefont
  {B.}~\bibnamefont {Chiaro}}, \bibinfo {author} {\bibfnamefont
  {A.}~\bibnamefont {Dunsworth}}, \bibinfo {author} {\bibfnamefont
  {C.}~\bibnamefont {Neill}}, \bibinfo {author} {\bibfnamefont
  {P.}~\bibnamefont {O'Malley}}, \bibinfo {author} {\bibfnamefont
  {P.}~\bibnamefont {Roushan}}, \bibinfo {author} {\bibfnamefont
  {A.}~\bibnamefont {Vainsencher}}, \bibinfo {author} {\bibfnamefont
  {J.}~\bibnamefont {Wenner}}, \bibinfo {author} {\bibfnamefont {A.~N.}\
  \bibnamefont {Korotkov}}, \bibinfo {author} {\bibfnamefont {A.~N.}\
  \bibnamefont {Cleland}}, \ and\ \bibinfo {author} {\bibfnamefont {John~M.}\
  \bibnamefont {Martinis}},\ }\bibfield  {title} {\enquote {\bibinfo {title}
  {{Superconducting quantum circuits at the surface code threshold for fault
  tolerance}},}\ }\href {\doibase 10.1038/nature13171} {\bibfield  {journal}
  {\bibinfo  {journal} {Nature}\ }\textbf {\bibinfo {volume} {508}},\ \bibinfo
  {pages} {500--503} (\bibinfo {year} {2014})}\BibitemShut {NoStop}%
\bibitem [{\citenamefont {Muhonen}\ \emph {et~al.}(2015)\citenamefont
  {Muhonen}, \citenamefont {Laucht}, \citenamefont {Simmons}, \citenamefont
  {Dehollain}, \citenamefont {Kalra}, \citenamefont {Hudson}, \citenamefont
  {Freer}, \citenamefont {Itoh}, \citenamefont {Jamieson}, \citenamefont
  {McCallum}, \citenamefont {Dzurak},\ and\ \citenamefont
  {Morello}}]{Muhonen2015}%
  \BibitemOpen
  \bibfield  {author} {\bibinfo {author} {\bibfnamefont {J.~T.}\ \bibnamefont
  {Muhonen}}, \bibinfo {author} {\bibfnamefont {A.}~\bibnamefont {Laucht}},
  \bibinfo {author} {\bibfnamefont {S.}~\bibnamefont {Simmons}}, \bibinfo
  {author} {\bibfnamefont {J.~P.}\ \bibnamefont {Dehollain}}, \bibinfo {author}
  {\bibfnamefont {R.}~\bibnamefont {Kalra}}, \bibinfo {author} {\bibfnamefont
  {F.~E.}\ \bibnamefont {Hudson}}, \bibinfo {author} {\bibfnamefont
  {S.}~\bibnamefont {Freer}}, \bibinfo {author} {\bibfnamefont {K.~M.}\
  \bibnamefont {Itoh}}, \bibinfo {author} {\bibfnamefont {D.~N.}\ \bibnamefont
  {Jamieson}}, \bibinfo {author} {\bibfnamefont {J.~C.}\ \bibnamefont
  {McCallum}}, \bibinfo {author} {\bibfnamefont {A.~S.}\ \bibnamefont
  {Dzurak}}, \ and\ \bibinfo {author} {\bibfnamefont {A.}~\bibnamefont
  {Morello}},\ }\bibfield  {title} {\enquote {\bibinfo {title} {{Quantifying
  the quantum gate fidelity of single-atom spin qubits in silicon by randomized
  benchmarking}},}\ }\href {\doibase 10.1088/0953-8984/27/15/154205} {\bibfield
   {journal} {\bibinfo  {journal} {J. Phys.: Condens. Matter}\ }\textbf
  {\bibinfo {volume} {27}},\ \bibinfo {pages} {154205} (\bibinfo {year}
  {2015})}\BibitemShut {NoStop}%
\bibitem [{\citenamefont {Xia}\ \emph {et~al.}(2015)\citenamefont {Xia},
  \citenamefont {Lichtman}, \citenamefont {Maller}, \citenamefont {Carr},
  \citenamefont {Piotrowicz}, \citenamefont {Isenhower},\ and\ \citenamefont
  {Saffman}}]{Xia2015}%
  \BibitemOpen
  \bibfield  {author} {\bibinfo {author} {\bibfnamefont {T.}~\bibnamefont
  {Xia}}, \bibinfo {author} {\bibfnamefont {M.}~\bibnamefont {Lichtman}},
  \bibinfo {author} {\bibfnamefont {K.}~\bibnamefont {Maller}}, \bibinfo
  {author} {\bibfnamefont {A.~W.}\ \bibnamefont {Carr}}, \bibinfo {author}
  {\bibfnamefont {M.~J.}\ \bibnamefont {Piotrowicz}}, \bibinfo {author}
  {\bibfnamefont {L.}~\bibnamefont {Isenhower}}, \ and\ \bibinfo {author}
  {\bibfnamefont {M.}~\bibnamefont {Saffman}},\ }\bibfield  {title} {\enquote
  {\bibinfo {title} {{Randomized benchmarking of single-qubit gates in a 2D
  array of neutral-atom qubits}},}\ }\href {\doibase
  10.1103/PhysRevLett.114.100503} {\bibfield  {journal} {\bibinfo  {journal}
  {Phys. Rev. Lett.}\ }\textbf {\bibinfo {volume} {114}},\ \bibinfo {pages}
  {100503} (\bibinfo {year} {2015})}\BibitemShut {NoStop}%
\bibitem [{\citenamefont {Magesan}\ \emph
  {et~al.}(2012{\natexlab{b}})\citenamefont {Magesan}, \citenamefont
  {Gambetta}, \citenamefont {Johnson}, \citenamefont {Ryan}, \citenamefont
  {Chow}, \citenamefont {Merkel}, \citenamefont {da~Silva}, \citenamefont
  {Keefe}, \citenamefont {Rothwell}, \citenamefont {Ohki}, \citenamefont
  {Ketchen},\ and\ \citenamefont {Steffen}}]{Magesan2012}%
  \BibitemOpen
  \bibfield  {author} {\bibinfo {author} {\bibfnamefont {Easwar}\ \bibnamefont
  {Magesan}}, \bibinfo {author} {\bibfnamefont {Jay~M.}\ \bibnamefont
  {Gambetta}}, \bibinfo {author} {\bibfnamefont {B.~R.}\ \bibnamefont
  {Johnson}}, \bibinfo {author} {\bibfnamefont {Colm~A.}\ \bibnamefont {Ryan}},
  \bibinfo {author} {\bibfnamefont {Jerry~M.}\ \bibnamefont {Chow}}, \bibinfo
  {author} {\bibfnamefont {Seth~T.}\ \bibnamefont {Merkel}}, \bibinfo {author}
  {\bibfnamefont {Marcus~P.}\ \bibnamefont {da~Silva}}, \bibinfo {author}
  {\bibfnamefont {George~A.}\ \bibnamefont {Keefe}}, \bibinfo {author}
  {\bibfnamefont {Mary~B.}\ \bibnamefont {Rothwell}}, \bibinfo {author}
  {\bibfnamefont {Thomas~A.}\ \bibnamefont {Ohki}}, \bibinfo {author}
  {\bibfnamefont {Mark~B.}\ \bibnamefont {Ketchen}}, \ and\ \bibinfo {author}
  {\bibfnamefont {M.}~\bibnamefont {Steffen}},\ }\bibfield  {title} {\enquote
  {\bibinfo {title} {{Efficient measurement of quantum gate error by
  interleaved randomized benchmarking}},}\ }\href {\doibase
  10.1103/PhysRevLett.109.080505} {\bibfield  {journal} {\bibinfo  {journal}
  {Phys. Rev. Lett.}\ }\textbf {\bibinfo {volume} {109}},\ \bibinfo {pages}
  {080505} (\bibinfo {year} {2012}{\natexlab{b}})}\BibitemShut {NoStop}%
\bibitem [{\citenamefont {Wallman}\ \emph
  {et~al.}(2015{\natexlab{a}})\citenamefont {Wallman}, \citenamefont {Granade},
  \citenamefont {Harper},\ and\ \citenamefont {Flammia}}]{Wallman2015a}%
  \BibitemOpen
  \bibfield  {author} {\bibinfo {author} {\bibfnamefont {Joel}\ \bibnamefont
  {Wallman}}, \bibinfo {author} {\bibfnamefont {Chris}\ \bibnamefont
  {Granade}}, \bibinfo {author} {\bibfnamefont {Robin}\ \bibnamefont {Harper}},
  \ and\ \bibinfo {author} {\bibfnamefont {Steven~T.}\ \bibnamefont
  {Flammia}},\ }\bibfield  {title} {\enquote {\bibinfo {title} {{Estimating the
  coherence of noise}},}\ }\href {\doibase 10.1088/1367-2630/17/11/113020}
  {\bibfield  {journal} {\bibinfo  {journal} {New J. Phys.}\ }\textbf {\bibinfo
  {volume} {17}},\ \bibinfo {pages} {113020} (\bibinfo {year}
  {2015}{\natexlab{a}})}\BibitemShut {NoStop}%
\bibitem [{\citenamefont {Wallman}\ \emph
  {et~al.}(2015{\natexlab{b}})\citenamefont {Wallman}, \citenamefont
  {Barnhill},\ and\ \citenamefont {Emerson}}]{Wallman2015b}%
  \BibitemOpen
  \bibfield  {author} {\bibinfo {author} {\bibfnamefont {Joel~J.}\ \bibnamefont
  {Wallman}}, \bibinfo {author} {\bibfnamefont {Marie}\ \bibnamefont
  {Barnhill}}, \ and\ \bibinfo {author} {\bibfnamefont {Joseph}\ \bibnamefont
  {Emerson}},\ }\bibfield  {title} {\enquote {\bibinfo {title} {{Robust
  Characterization of Loss Rates}},}\ }\href {\doibase
  10.1103/PhysRevLett.115.060501} {\bibfield  {journal} {\bibinfo  {journal}
  {Phys. Rev. Lett.}\ }\textbf {\bibinfo {volume} {115}},\ \bibinfo {pages}
  {060501} (\bibinfo {year} {2015}{\natexlab{b}})}\BibitemShut {NoStop}%
\bibitem [{\citenamefont {Wallman}\ \emph {et~al.}(2016)\citenamefont
  {Wallman}, \citenamefont {Barnhill},\ and\ \citenamefont
  {Emerson}}]{Wallman2016}%
  \BibitemOpen
  \bibfield  {author} {\bibinfo {author} {\bibfnamefont {Joel~J.}\ \bibnamefont
  {Wallman}}, \bibinfo {author} {\bibfnamefont {Marie}\ \bibnamefont
  {Barnhill}}, \ and\ \bibinfo {author} {\bibfnamefont {Joseph}\ \bibnamefont
  {Emerson}},\ }\bibfield  {title} {\enquote {\bibinfo {title} {{Robust
  characterization of leakage errors}},}\ }\href {\doibase
  10.1088/1367-2630/18/4/043021} {\bibfield  {journal} {\bibinfo  {journal}
  {New J. Phys.}\ }\textbf {\bibinfo {volume} {18}},\ \bibinfo {pages} {043021}
  (\bibinfo {year} {2016})}\BibitemShut {NoStop}%
\bibitem [{\citenamefont {Combes}\ \emph {et~al.}(2017)\citenamefont {Combes},
  \citenamefont {Granade}, \citenamefont {Ferrie},\ and\ \citenamefont
  {Flammia}}]{Combes2017}%
  \BibitemOpen
  \bibfield  {author} {\bibinfo {author} {\bibfnamefont {Joshua}\ \bibnamefont
  {Combes}}, \bibinfo {author} {\bibfnamefont {Christopher}\ \bibnamefont
  {Granade}}, \bibinfo {author} {\bibfnamefont {Christopher}\ \bibnamefont
  {Ferrie}}, \ and\ \bibinfo {author} {\bibfnamefont {Steven~T.}\ \bibnamefont
  {Flammia}},\ }\bibfield  {title} {\enquote {\bibinfo {title} {{Logical
  Randomized Benchmarking}},}\ }\href {http://arxiv.org/abs/1702.03688}
  {\bibfield  {journal} {\bibinfo  {journal} {arXiv:1702.03688}\ } (\bibinfo
  {year} {2017})}\BibitemShut {NoStop}%
\bibitem [{\citenamefont {Feng}\ \emph {et~al.}(2016)\citenamefont {Feng},
  \citenamefont {Wallman}, \citenamefont {Buonacorsi}, \citenamefont {Cho},
  \citenamefont {Park}, \citenamefont {Xin}, \citenamefont {Lu}, \citenamefont
  {Baugh},\ and\ \citenamefont {Laflamme}}]{Feng2016}%
  \BibitemOpen
  \bibfield  {author} {\bibinfo {author} {\bibfnamefont {Guanru}\ \bibnamefont
  {Feng}}, \bibinfo {author} {\bibfnamefont {Joel~J.}\ \bibnamefont {Wallman}},
  \bibinfo {author} {\bibfnamefont {Brandon}\ \bibnamefont {Buonacorsi}},
  \bibinfo {author} {\bibfnamefont {Franklin~H.}\ \bibnamefont {Cho}}, \bibinfo
  {author} {\bibfnamefont {Daniel~K.}\ \bibnamefont {Park}}, \bibinfo {author}
  {\bibfnamefont {Tao}\ \bibnamefont {Xin}}, \bibinfo {author} {\bibfnamefont
  {Dawei}\ \bibnamefont {Lu}}, \bibinfo {author} {\bibfnamefont {Jonathan}\
  \bibnamefont {Baugh}}, \ and\ \bibinfo {author} {\bibfnamefont {Raymond}\
  \bibnamefont {Laflamme}},\ }\bibfield  {title} {\enquote {\bibinfo {title}
  {{Estimating the Coherence of Noise in Quantum Control of a Solid-State
  Qubit}},}\ }\href {\doibase 10.1103/PhysRevLett.117.260501} {\bibfield
  {journal} {\bibinfo  {journal} {Phys. Rev. Lett.}\ }\textbf {\bibinfo
  {volume} {117}},\ \bibinfo {pages} {260501} (\bibinfo {year}
  {2016})}\BibitemShut {NoStop}%
\bibitem [{\citenamefont {Sheldon}\ \emph {et~al.}(2016)\citenamefont
  {Sheldon}, \citenamefont {Bishop}, \citenamefont {Magesan}, \citenamefont
  {Filipp}, \citenamefont {Chow},\ and\ \citenamefont
  {Gambetta}}]{Sheldon2016}%
  \BibitemOpen
  \bibfield  {author} {\bibinfo {author} {\bibfnamefont {Sarah}\ \bibnamefont
  {Sheldon}}, \bibinfo {author} {\bibfnamefont {Lev~S.}\ \bibnamefont
  {Bishop}}, \bibinfo {author} {\bibfnamefont {Easwar}\ \bibnamefont
  {Magesan}}, \bibinfo {author} {\bibfnamefont {Stefan}\ \bibnamefont
  {Filipp}}, \bibinfo {author} {\bibfnamefont {Jerry~M.}\ \bibnamefont {Chow}},
  \ and\ \bibinfo {author} {\bibfnamefont {Jay~M.}\ \bibnamefont {Gambetta}},\
  }\bibfield  {title} {\enquote {\bibinfo {title} {{Characterizing errors on
  qubit operations via iterative randomized benchmarking}},}\ }\href {\doibase
  10.1103/PhysRevA.93.012301} {\bibfield  {journal} {\bibinfo  {journal} {Phys.
  Rev. A}\ }\textbf {\bibinfo {volume} {93}},\ \bibinfo {pages} {012301}
  (\bibinfo {year} {2016})}\BibitemShut {NoStop}%
\bibitem [{\citenamefont {Carignan-Dugas}\ \emph {et~al.}(2016)\citenamefont
  {Carignan-Dugas}, \citenamefont {Wallman},\ and\ \citenamefont
  {Emerson}}]{Carignan-Dugas2016}%
  \BibitemOpen
  \bibfield  {author} {\bibinfo {author} {\bibfnamefont {Arnaud}\ \bibnamefont
  {Carignan-Dugas}}, \bibinfo {author} {\bibfnamefont {Joel~J.}\ \bibnamefont
  {Wallman}}, \ and\ \bibinfo {author} {\bibfnamefont {Joseph}\ \bibnamefont
  {Emerson}},\ }\bibfield  {title} {\enquote {\bibinfo {title} {{Efficiently
  characterizing the total error in quantum circuits}},}\ }\href
  {http://arxiv.org/abs/1610.05296} {\bibfield  {journal} {\bibinfo  {journal}
  {arXiv:1610.05296}\ } (\bibinfo {year} {2016})}\BibitemShut {NoStop}%
\bibitem [{\citenamefont {Sanders}\ \emph {et~al.}(2016)\citenamefont
  {Sanders}, \citenamefont {Wallman},\ and\ \citenamefont
  {Sanders}}]{Sanders2016}%
  \BibitemOpen
  \bibfield  {author} {\bibinfo {author} {\bibfnamefont {Yuval~R.}\
  \bibnamefont {Sanders}}, \bibinfo {author} {\bibfnamefont {Joel~J.}\
  \bibnamefont {Wallman}}, \ and\ \bibinfo {author} {\bibfnamefont {Barry~C.}\
  \bibnamefont {Sanders}},\ }\bibfield  {title} {\enquote {\bibinfo {title}
  {{Bounding quantum gate error rate based on reported average fidelity}},}\
  }\href {\doibase 10.1088/1367-2630/18/1/012002} {\bibfield  {journal}
  {\bibinfo  {journal} {New J. Phys.}\ }\textbf {\bibinfo {volume} {18}},\
  \bibinfo {pages} {012002} (\bibinfo {year} {2016})}\BibitemShut {NoStop}%
\bibitem [{\citenamefont {Kueng}\ \emph {et~al.}(2016)\citenamefont {Kueng},
  \citenamefont {Long}, \citenamefont {Doherty},\ and\ \citenamefont
  {Flammia}}]{Kueng2016}%
  \BibitemOpen
  \bibfield  {author} {\bibinfo {author} {\bibfnamefont {Richard}\ \bibnamefont
  {Kueng}}, \bibinfo {author} {\bibfnamefont {David~M.}\ \bibnamefont {Long}},
  \bibinfo {author} {\bibfnamefont {Andrew~C.}\ \bibnamefont {Doherty}}, \ and\
  \bibinfo {author} {\bibfnamefont {Steven~T.}\ \bibnamefont {Flammia}},\
  }\bibfield  {title} {\enquote {\bibinfo {title} {{Comparing Experiments to
  the Fault-Tolerance Threshold}},}\ }\href {\doibase
  10.1103/PhysRevLett.117.170502} {\bibfield  {journal} {\bibinfo  {journal}
  {Phys. Rev. Lett.}\ }\textbf {\bibinfo {volume} {117}},\ \bibinfo {pages}
  {170502} (\bibinfo {year} {2016})}\BibitemShut {NoStop}%
\bibitem [{\citenamefont {Wallman}(2015)}]{Wallman2015}%
  \BibitemOpen
  \bibfield  {author} {\bibinfo {author} {\bibfnamefont {Joel~J.}\ \bibnamefont
  {Wallman}},\ }\bibfield  {title} {\enquote {\bibinfo {title} {{Bounding
  experimental quantum error rates relative to fault-tolerant thresholds}},}\
  }\href {http://arxiv.org/abs/1511.00727} {\bibfield  {journal} {\bibinfo
  {journal} {arXiv:1511.00727}\ } (\bibinfo {year} {2015})}\BibitemShut
  {NoStop}%
\bibitem [{\citenamefont {Thinh}\ \emph {et~al.}(2018)\citenamefont {Thinh},
  \citenamefont {Faist}, \citenamefont {Helsen}, \citenamefont {Elkouss},\ and\
  \citenamefont {Wehner}}]{Thinh2018}%
  \BibitemOpen
  \bibfield  {author} {\bibinfo {author} {\bibfnamefont {Le~Phuc}\ \bibnamefont
  {Thinh}}, \bibinfo {author} {\bibfnamefont {Philippe}\ \bibnamefont {Faist}},
  \bibinfo {author} {\bibfnamefont {Jonas}\ \bibnamefont {Helsen}}, \bibinfo
  {author} {\bibfnamefont {David}\ \bibnamefont {Elkouss}}, \ and\ \bibinfo
  {author} {\bibfnamefont {Stephanie}\ \bibnamefont {Wehner}},\ }\bibfield
  {title} {\enquote {\bibinfo {title} {{Practical and reliable error bars for
  quantum process tomography}},}\ }\href {http://arxiv.org/abs/1808.00358}
  {\bibfield  {journal} {\bibinfo  {journal} {arXiv:1808.00358}\ } (\bibinfo
  {year} {2018})}\BibitemShut {NoStop}%
\bibitem [{\citenamefont {Epstein}\ \emph {et~al.}(2014)\citenamefont
  {Epstein}, \citenamefont {Cross}, \citenamefont {Magesan},\ and\
  \citenamefont {Gambetta}}]{Epstein2014}%
  \BibitemOpen
  \bibfield  {author} {\bibinfo {author} {\bibfnamefont {Jeffrey~M.}\
  \bibnamefont {Epstein}}, \bibinfo {author} {\bibfnamefont {Andrew~W.}\
  \bibnamefont {Cross}}, \bibinfo {author} {\bibfnamefont {Easwar}\
  \bibnamefont {Magesan}}, \ and\ \bibinfo {author} {\bibfnamefont {Jay~M.}\
  \bibnamefont {Gambetta}},\ }\bibfield  {title} {\enquote {\bibinfo {title}
  {{Investigating the limits of randomized benchmarking protocols}},}\ }\href
  {\doibase 10.1103/PhysRevA.89.062321} {\bibfield  {journal} {\bibinfo
  {journal} {Phys. Rev. A}\ }\textbf {\bibinfo {volume} {89}},\ \bibinfo
  {pages} {062321} (\bibinfo {year} {2014})}\BibitemShut {NoStop}%
\bibitem [{\citenamefont {Granade}\ \emph {et~al.}(2015)\citenamefont
  {Granade}, \citenamefont {Ferrie},\ and\ \citenamefont {Cory}}]{Granade2015}%
  \BibitemOpen
  \bibfield  {author} {\bibinfo {author} {\bibfnamefont {Christopher}\
  \bibnamefont {Granade}}, \bibinfo {author} {\bibfnamefont {Christopher}\
  \bibnamefont {Ferrie}}, \ and\ \bibinfo {author} {\bibfnamefont {D.~G.}\
  \bibnamefont {Cory}},\ }\bibfield  {title} {\enquote {\bibinfo {title}
  {{Accelerated randomized benchmarking}},}\ }\href {\doibase
  10.1088/1367-2630/17/1/013042} {\bibfield  {journal} {\bibinfo  {journal}
  {New J. Phys.}\ }\textbf {\bibinfo {volume} {17}},\ \bibinfo {pages} {013042}
  (\bibinfo {year} {2015})}\BibitemShut {NoStop}%
\bibitem [{\citenamefont {Hincks}\ \emph {et~al.}(2018)\citenamefont {Hincks},
  \citenamefont {Wallman}, \citenamefont {Ferrie}, \citenamefont {Granade},\
  and\ \citenamefont {Cory}}]{Hincks2018}%
  \BibitemOpen
  \bibfield  {author} {\bibinfo {author} {\bibfnamefont {Ian}\ \bibnamefont
  {Hincks}}, \bibinfo {author} {\bibfnamefont {Joel~J}\ \bibnamefont
  {Wallman}}, \bibinfo {author} {\bibfnamefont {Chris}\ \bibnamefont {Ferrie}},
  \bibinfo {author} {\bibfnamefont {Chris}\ \bibnamefont {Granade}}, \ and\
  \bibinfo {author} {\bibfnamefont {David~G}\ \bibnamefont {Cory}},\ }\bibfield
   {title} {\enquote {\bibinfo {title} {{Bayesian Inference for Randomized
  Benchmarking Protocols}},}\ }\href {http://arxiv.org/abs/1802.00401}
  {\bibfield  {journal} {\bibinfo  {journal} {arXiv:1802.00401}\ } (\bibinfo
  {year} {2018})}\BibitemShut {NoStop}%
\bibitem [{\citenamefont {Wallman}\ and\ \citenamefont
  {Flammia}(2014)}]{Wallman2014}%
  \BibitemOpen
  \bibfield  {author} {\bibinfo {author} {\bibfnamefont {Joel~J.}\ \bibnamefont
  {Wallman}}\ and\ \bibinfo {author} {\bibfnamefont {Steven~T.}\ \bibnamefont
  {Flammia}},\ }\bibfield  {title} {\enquote {\bibinfo {title} {{Randomized
  benchmarking with confidence}},}\ }\href {\doibase
  10.1088/1367-2630/16/10/103032} {\bibfield  {journal} {\bibinfo  {journal}
  {New J. Phys.}\ }\textbf {\bibinfo {volume} {16}},\ \bibinfo {pages} {103032}
  (\bibinfo {year} {2014})}\BibitemShut {NoStop}%
\bibitem [{\citenamefont {Helsen}\ \emph {et~al.}(2017)\citenamefont {Helsen},
  \citenamefont {Wallman}, \citenamefont {Flammia},\ and\ \citenamefont
  {Wehner}}]{Helsen2017}%
  \BibitemOpen
  \bibfield  {author} {\bibinfo {author} {\bibfnamefont {Jonas}\ \bibnamefont
  {Helsen}}, \bibinfo {author} {\bibfnamefont {Joel~J.}\ \bibnamefont
  {Wallman}}, \bibinfo {author} {\bibfnamefont {Steven~T.}\ \bibnamefont
  {Flammia}}, \ and\ \bibinfo {author} {\bibfnamefont {Stephanie}\ \bibnamefont
  {Wehner}},\ }\bibfield  {title} {\enquote {\bibinfo {title} {{Multi-qubit
  Randomized Benchmarking Using Few Samples}},}\ }\href
  {http://arxiv.org/abs/1701.04299} {\bibfield  {journal} {\bibinfo  {journal}
  {arXiv:1701.04299}\ } (\bibinfo {year} {2017})}\BibitemShut {NoStop}%
\bibitem [{\citenamefont {Gross}\ \emph {et~al.}(2007)\citenamefont {Gross},
  \citenamefont {Audenaert},\ and\ \citenamefont {Eisert}}]{Gross2007}%
  \BibitemOpen
  \bibfield  {author} {\bibinfo {author} {\bibfnamefont {D.}~\bibnamefont
  {Gross}}, \bibinfo {author} {\bibfnamefont {K.}~\bibnamefont {Audenaert}}, \
  and\ \bibinfo {author} {\bibfnamefont {J.}~\bibnamefont {Eisert}},\
  }\bibfield  {title} {\enquote {\bibinfo {title} {{Evenly distributed
  unitaries: On the structure of unitary designs}},}\ }\href {\doibase
  10.1063/1.2716992} {\bibfield  {journal} {\bibinfo  {journal} {J. Math.
  Phys.}\ }\textbf {\bibinfo {volume} {48}},\ \bibinfo {pages} {052104}
  (\bibinfo {year} {2007})}\BibitemShut {NoStop}%
\bibitem [{\citenamefont {Hoeffding}(1963)}]{Hoeffding1963}%
  \BibitemOpen
  \bibfield  {author} {\bibinfo {author} {\bibfnamefont {Wassily}\ \bibnamefont
  {Hoeffding}},\ }\bibfield  {title} {\enquote {\bibinfo {title} {{Probability
  Inequalities for Sums of Bounded Random Variables}},}\ }\href {\doibase
  10.1080/01621459.1963.10500830} {\bibfield  {journal} {\bibinfo  {journal}
  {J. Am. Stat. Assoc.}\ }\textbf {\bibinfo {volume} {58}},\ \bibinfo {pages}
  {301} (\bibinfo {year} {1963})}\BibitemShut {NoStop}%
\bibitem [{\citenamefont {Johnston}(2016)}]{QETLAB}%
  \BibitemOpen
  \bibfield  {author} {\bibinfo {author} {\bibfnamefont {Nathaniel}\
  \bibnamefont {Johnston}},\ }\href {http://qetlab.com} {\enquote {\bibinfo
  {title} {{QETLAB: A MATLAB toolbox for quantum entanglement, version 0.9}},}\
  } (\bibinfo {year} {2016})\BibitemShut {NoStop}%
\bibitem [{\citenamefont {Fuchs}\ and\ \citenamefont {van~de
  Graaf}(1999)}]{Fuchs1999}%
  \BibitemOpen
  \bibfield  {author} {\bibinfo {author} {\bibfnamefont {Christopher~A.}\
  \bibnamefont {Fuchs}}\ and\ \bibinfo {author} {\bibfnamefont {Jeroen}\
  \bibnamefont {van~de Graaf}},\ }\bibfield  {title} {\enquote {\bibinfo
  {title} {{Cryptographic Distinguishability Measures for Quantum Mechanical
  States}},}\ }\href {\doibase 10.1109/18.761271} {\bibfield  {journal}
  {\bibinfo  {journal} {IEEE Trans. Inf. Theory}\ }\textbf {\bibinfo {volume}
  {45}},\ \bibinfo {pages} {1216} (\bibinfo {year} {1999})}\BibitemShut
  {NoStop}%
\bibitem [{\citenamefont {Weiss}(2005)}]{Weiss2005}%
  \BibitemOpen
  \bibfield  {author} {\bibinfo {author} {\bibfnamefont {Neil~A.}\ \bibnamefont
  {Weiss}},\ }\href@noop {} {\emph {\bibinfo {title} {{A Course in
  Probability}}}}\ (\bibinfo  {publisher} {Addison-Wesley},\ \bibinfo {address}
  {Boston, MA},\ \bibinfo {year} {2005})\ pp.\ \bibinfo {pages}
  {380--383}\BibitemShut {NoStop}%
\bibitem [{\citenamefont {Watrous}(2018)}]{Watrous2017}%
  \BibitemOpen
  \bibfield  {author} {\bibinfo {author} {\bibfnamefont {John}\ \bibnamefont
  {Watrous}},\ }\href {\doibase 10.1017/9781316848142} {\emph {\bibinfo {title}
  {{The Theory of Quantum Information}}}}\ (\bibinfo  {publisher} {Cambridge
  University Press},\ \bibinfo {address} {Cambridge, MA},\ \bibinfo {year}
  {2018})\BibitemShut {NoStop}%
\bibitem [{\citenamefont {Zhu}\ \emph {et~al.}(2016)\citenamefont {Zhu},
  \citenamefont {Kueng}, \citenamefont {Grassl},\ and\ \citenamefont
  {Gross}}]{Zhu2016}%
  \BibitemOpen
  \bibfield  {author} {\bibinfo {author} {\bibfnamefont {Huangjun}\
  \bibnamefont {Zhu}}, \bibinfo {author} {\bibfnamefont {Richard}\ \bibnamefont
  {Kueng}}, \bibinfo {author} {\bibfnamefont {Markus}\ \bibnamefont {Grassl}},
  \ and\ \bibinfo {author} {\bibfnamefont {David}\ \bibnamefont {Gross}},\
  }\bibfield  {title} {\enquote {\bibinfo {title} {{The Clifford group fails
  gracefully to be a unitary 4-design}},}\ }\href
  {http://arxiv.org/abs/1609.08172} {\bibfield  {journal} {\bibinfo  {journal}
  {arXiv:1609.08172}\ } (\bibinfo {year} {2016})}\BibitemShut {NoStop}%
\bibitem [{\citenamefont {Zhu}(2017)}]{Zhu2017}%
  \BibitemOpen
  \bibfield  {author} {\bibinfo {author} {\bibfnamefont {Huangjun}\
  \bibnamefont {Zhu}},\ }\bibfield  {title} {\enquote {\bibinfo {title}
  {{Multiqubit Clifford groups are unitary 3-designs}},}\ }\href {\doibase
  10.1103/PhysRevA.96.062336} {\bibfield  {journal} {\bibinfo  {journal} {Phys.
  Rev. A}\ }\textbf {\bibinfo {volume} {96}},\ \bibinfo {pages} {062336}
  (\bibinfo {year} {2017})}\BibitemShut {NoStop}%
\bibitem [{\citenamefont {Helsen}\ \emph
  {et~al.}(2018{\natexlab{a}})\citenamefont {Helsen}, \citenamefont {Wallman},\
  and\ \citenamefont {Wehner}}]{Helsen2018}%
  \BibitemOpen
  \bibfield  {author} {\bibinfo {author} {\bibfnamefont {Jonas}\ \bibnamefont
  {Helsen}}, \bibinfo {author} {\bibfnamefont {Joel~J.}\ \bibnamefont
  {Wallman}}, \ and\ \bibinfo {author} {\bibfnamefont {Stephanie}\ \bibnamefont
  {Wehner}},\ }\bibfield  {title} {\enquote {\bibinfo {title} {{Representations
  of the multi-qubit Clifford group}},}\ }\href {\doibase 10.1063/1.4997688}
  {\bibfield  {journal} {\bibinfo  {journal} {J. Math. Phys.}\ }\textbf
  {\bibinfo {volume} {59}},\ \bibinfo {pages} {072201} (\bibinfo {year}
  {2018}{\natexlab{a}})}\BibitemShut {NoStop}%
\bibitem [{\citenamefont {Fulton}\ and\ \citenamefont
  {Harris}(2004)}]{Fulton2004}%
  \BibitemOpen
  \bibfield  {author} {\bibinfo {author} {\bibfnamefont {William}\ \bibnamefont
  {Fulton}}\ and\ \bibinfo {author} {\bibfnamefont {Joe}\ \bibnamefont
  {Harris}},\ }\href@noop {} {\emph {\bibinfo {title} {{Representation
  Theory}}}}\ (\bibinfo  {publisher} {Springer},\ \bibinfo {address} {New York
  City, NY},\ \bibinfo {year} {2004})\BibitemShut {NoStop}%
\bibitem [{\citenamefont {P{\'{e}}rez-Garc{\'{i}}a}\ \emph
  {et~al.}(2006)\citenamefont {P{\'{e}}rez-Garc{\'{i}}a}, \citenamefont {Wolf},
  \citenamefont {Petz},\ and\ \citenamefont {Ruskai}}]{Perez-Garcia2006}%
  \BibitemOpen
  \bibfield  {author} {\bibinfo {author} {\bibfnamefont {David}\ \bibnamefont
  {P{\'{e}}rez-Garc{\'{i}}a}}, \bibinfo {author} {\bibfnamefont {Michael~M.}\
  \bibnamefont {Wolf}}, \bibinfo {author} {\bibfnamefont {Denes}\ \bibnamefont
  {Petz}}, \ and\ \bibinfo {author} {\bibfnamefont {Mary~Beth}\ \bibnamefont
  {Ruskai}},\ }\bibfield  {title} {\enquote {\bibinfo {title} {{Contractivity
  of positive and trace-preserving maps under LPnorms}},}\ }\href {\doibase
  10.1063/1.2218675} {\bibfield  {journal} {\bibinfo  {journal} {J. Math.
  Phys.}\ }\textbf {\bibinfo {volume} {47}},\ \bibinfo {pages} {083506}
  (\bibinfo {year} {2006})}\BibitemShut {NoStop}%
\bibitem [{\citenamefont {Chasseur}\ and\ \citenamefont
  {Wilhelm}(2015)}]{Chasseur2015}%
  \BibitemOpen
  \bibfield  {author} {\bibinfo {author} {\bibfnamefont {T.}~\bibnamefont
  {Chasseur}}\ and\ \bibinfo {author} {\bibfnamefont {F.~K.}\ \bibnamefont
  {Wilhelm}},\ }\bibfield  {title} {\enquote {\bibinfo {title} {{Complete
  randomized benchmarking protocol accounting for leakage errors}},}\ }\href
  {\doibase 10.1103/PhysRevA.92.042333} {\bibfield  {journal} {\bibinfo
  {journal} {Phys. Rev. A}\ }\textbf {\bibinfo {volume} {92}},\ \bibinfo
  {pages} {042333} (\bibinfo {year} {2015})}\BibitemShut {NoStop}%
\bibitem [{\citenamefont {Proctor}\ \emph {et~al.}(2017)\citenamefont
  {Proctor}, \citenamefont {Rudinger}, \citenamefont {Young}, \citenamefont
  {Sarovar},\ and\ \citenamefont {Blume-Kohout}}]{Proctor2017}%
  \BibitemOpen
  \bibfield  {author} {\bibinfo {author} {\bibfnamefont {Timothy}\ \bibnamefont
  {Proctor}}, \bibinfo {author} {\bibfnamefont {Kenneth}\ \bibnamefont
  {Rudinger}}, \bibinfo {author} {\bibfnamefont {Kevin}\ \bibnamefont {Young}},
  \bibinfo {author} {\bibfnamefont {Mohan}\ \bibnamefont {Sarovar}}, \ and\
  \bibinfo {author} {\bibfnamefont {Robin}\ \bibnamefont {Blume-Kohout}},\
  }\bibfield  {title} {\enquote {\bibinfo {title} {{What Randomized
  Benchmarking Actually Measures}},}\ }\href {\doibase
  10.1103/PhysRevLett.119.130502} {\bibfield  {journal} {\bibinfo  {journal}
  {Phys. Rev. Lett.}\ }\textbf {\bibinfo {volume} {119}},\ \bibinfo {pages}
  {130502} (\bibinfo {year} {2017})}\BibitemShut {NoStop}%
\bibitem [{\citenamefont {Wallman}(2018)}]{Wallman2018}%
  \BibitemOpen
  \bibfield  {author} {\bibinfo {author} {\bibfnamefont {Joel~J.}\ \bibnamefont
  {Wallman}},\ }\bibfield  {title} {\enquote {\bibinfo {title} {{Randomized
  benchmarking with gate-dependent noise}},}\ }\href {\doibase
  10.22331/q-2018-01-29-47} {\bibfield  {journal} {\bibinfo  {journal}
  {Quantum}\ }\textbf {\bibinfo {volume} {2}},\ \bibinfo {pages} {47} (\bibinfo
  {year} {2018})}\BibitemShut {NoStop}%
\bibitem [{\citenamefont {Dankert}\ \emph {et~al.}(2009)\citenamefont
  {Dankert}, \citenamefont {Cleve}, \citenamefont {Emerson},\ and\
  \citenamefont {Livine}}]{Dankert2006}%
  \BibitemOpen
  \bibfield  {author} {\bibinfo {author} {\bibfnamefont {Christoph}\
  \bibnamefont {Dankert}}, \bibinfo {author} {\bibfnamefont {Richard}\
  \bibnamefont {Cleve}}, \bibinfo {author} {\bibfnamefont {Joseph}\
  \bibnamefont {Emerson}}, \ and\ \bibinfo {author} {\bibfnamefont {Etera}\
  \bibnamefont {Livine}},\ }\bibfield  {title} {\enquote {\bibinfo {title}
  {{Exact and Approximate Unitary 2-Designs: Constructions and
  Applications}},}\ }\href {\doibase 10.1103/PhysRevA.80.012304} {\bibfield
  {journal} {\bibinfo  {journal} {Phys. Rev. A}\ }\textbf {\bibinfo {volume}
  {80}},\ \bibinfo {pages} {012304} (\bibinfo {year} {2009})}\BibitemShut
  {NoStop}%
\bibitem [{\citenamefont {Carignan-Dugas}\ \emph {et~al.}(2015)\citenamefont
  {Carignan-Dugas}, \citenamefont {Wallman},\ and\ \citenamefont
  {Emerson}}]{Carignan-Dugas2015}%
  \BibitemOpen
  \bibfield  {author} {\bibinfo {author} {\bibfnamefont {Arnaud}\ \bibnamefont
  {Carignan-Dugas}}, \bibinfo {author} {\bibfnamefont {Joel~J.}\ \bibnamefont
  {Wallman}}, \ and\ \bibinfo {author} {\bibfnamefont {Joseph}\ \bibnamefont
  {Emerson}},\ }\bibfield  {title} {\enquote {\bibinfo {title} {{Characterizing
  universal gate sets via dihedral benchmarking}},}\ }\href {\doibase
  10.1103/PhysRevA.92.060302} {\bibfield  {journal} {\bibinfo  {journal} {Phys.
  Rev. A}\ }\textbf {\bibinfo {volume} {92}},\ \bibinfo {pages} {060302}
  (\bibinfo {year} {2015})}\BibitemShut {NoStop}%
\bibitem [{\citenamefont {Cross}\ \emph {et~al.}(2016)\citenamefont {Cross},
  \citenamefont {Magesan}, \citenamefont {Bishop}, \citenamefont {Smolin},\
  and\ \citenamefont {Gambetta}}]{Cross2015}%
  \BibitemOpen
  \bibfield  {author} {\bibinfo {author} {\bibfnamefont {Andrew~W.}\
  \bibnamefont {Cross}}, \bibinfo {author} {\bibfnamefont {Easwar}\
  \bibnamefont {Magesan}}, \bibinfo {author} {\bibfnamefont {Lev~S.}\
  \bibnamefont {Bishop}}, \bibinfo {author} {\bibfnamefont {John~A.}\
  \bibnamefont {Smolin}}, \ and\ \bibinfo {author} {\bibfnamefont {Jay~M.}\
  \bibnamefont {Gambetta}},\ }\bibfield  {title} {\enquote {\bibinfo {title}
  {{Scalable randomized benchmarking of non-Clifford gates}},}\ }\href
  {\doibase 10.1038/npjqi.2016.12} {\bibfield  {journal} {\bibinfo  {journal}
  {npj Quantum Inf.}\ }\textbf {\bibinfo {volume} {2}},\ \bibinfo {pages}
  {16012} (\bibinfo {year} {2016})}\BibitemShut {NoStop}%
\bibitem [{\citenamefont {Fran{\c{c}}a}\ and\ \citenamefont
  {Hashagen}(2018)}]{Franca2018}%
  \BibitemOpen
  \bibfield  {author} {\bibinfo {author} {\bibfnamefont {Daniel~Stilck}\
  \bibnamefont {Fran{\c{c}}a}}\ and\ \bibinfo {author} {\bibfnamefont
  {Anna-Lena}\ \bibnamefont {Hashagen}},\ }\bibfield  {title} {\enquote
  {\bibinfo {title} {{Approximate Randomized Benchmarking for Finite
  Groups}},}\ }\href {\doibase 10.1088/1751-8121/aad6fa} {\bibfield  {journal}
  {\bibinfo  {journal} {J. Phys. A}\ }\textbf {\bibinfo {volume} {51}},\
  \bibinfo {pages} {395302} (\bibinfo {year} {2018})}\BibitemShut {NoStop}%
\bibitem [{\citenamefont {Hashagen}\ \emph {et~al.}(2018)\citenamefont
  {Hashagen}, \citenamefont {Flammia}, \citenamefont {Gross},\ and\
  \citenamefont {Wallman}}]{Hashagen2018}%
  \BibitemOpen
  \bibfield  {author} {\bibinfo {author} {\bibfnamefont {A.~K.}\ \bibnamefont
  {Hashagen}}, \bibinfo {author} {\bibfnamefont {S.~T.}\ \bibnamefont
  {Flammia}}, \bibinfo {author} {\bibfnamefont {D.}~\bibnamefont {Gross}}, \
  and\ \bibinfo {author} {\bibfnamefont {J.~J.}\ \bibnamefont {Wallman}},\
  }\bibfield  {title} {\enquote {\bibinfo {title} {{Real Randomized
  Benchmarking}},}\ }\href {\doibase 10.22331/q-2018-08-22-85} {\bibfield
  {journal} {\bibinfo  {journal} {Quantum}\ }\textbf {\bibinfo {volume} {2}},\
  \bibinfo {pages} {85} (\bibinfo {year} {2018})}\BibitemShut {NoStop}%
\bibitem [{\citenamefont {Brown}\ and\ \citenamefont
  {Eastin}(2018)}]{Brown2018}%
  \BibitemOpen
  \bibfield  {author} {\bibinfo {author} {\bibfnamefont {Winton~G.}\
  \bibnamefont {Brown}}\ and\ \bibinfo {author} {\bibfnamefont {Bryan}\
  \bibnamefont {Eastin}},\ }\bibfield  {title} {\enquote {\bibinfo {title}
  {{Randomized benchmarking with restricted gate sets}},}\ }\href {\doibase
  10.1103/PhysRevA.97.062323} {\bibfield  {journal} {\bibinfo  {journal} {Phys.
  Rev. A}\ }\textbf {\bibinfo {volume} {97}},\ \bibinfo {pages} {062323}
  (\bibinfo {year} {2018})}\BibitemShut {NoStop}%
\bibitem [{\citenamefont {Helsen}\ \emph
  {et~al.}(2018{\natexlab{b}})\citenamefont {Helsen}, \citenamefont {Xue},
  \citenamefont {Vandersypen},\ and\ \citenamefont {Wehner}}]{Helsen2018a}%
  \BibitemOpen
  \bibfield  {author} {\bibinfo {author} {\bibfnamefont {Jonas}\ \bibnamefont
  {Helsen}}, \bibinfo {author} {\bibfnamefont {Xiao}\ \bibnamefont {Xue}},
  \bibinfo {author} {\bibfnamefont {Lieven M.~K.}\ \bibnamefont {Vandersypen}},
  \ and\ \bibinfo {author} {\bibfnamefont {Stephanie}\ \bibnamefont {Wehner}},\
  }\bibfield  {title} {\enquote {\bibinfo {title} {{A new class of efficient
  randomized benchmarking protocols}},}\ }\href
  {http://arxiv.org/abs/1806.02048} {\bibfield  {journal} {\bibinfo  {journal}
  {arXiv:1806.02048}\ } (\bibinfo {year} {2018}{\natexlab{b}})}\BibitemShut
  {NoStop}%
\bibitem [{\citenamefont {Farinholt}(2014)}]{Farinholt2014}%
  \BibitemOpen
  \bibfield  {author} {\bibinfo {author} {\bibfnamefont {J.~M.}\ \bibnamefont
  {Farinholt}},\ }\bibfield  {title} {\enquote {\bibinfo {title} {{An ideal
  characterization of the Clifford operators}},}\ }\href {\doibase
  10.1088/1751-8113/47/30/305303} {\bibfield  {journal} {\bibinfo  {journal}
  {J. Phys. A}\ }\textbf {\bibinfo {volume} {47}},\ \bibinfo {pages} {305303}
  (\bibinfo {year} {2014})}\BibitemShut {NoStop}%
\bibitem [{\citenamefont {Ozols}(2008)}]{Ozols2008}%
  \BibitemOpen
  \bibfield  {author} {\bibinfo {author} {\bibfnamefont {Maris}\ \bibnamefont
  {Ozols}},\ }\href
  {http://home.lu.lv/~sd20008/papers/essays/Clifford%20group%20[paper].pdf}
  {\enquote {\bibinfo {title} {{Clifford group}},}\ } (\bibinfo {year}
  {2008})\BibitemShut {NoStop}%
\bibitem [{\citenamefont {Nielsen}(2002)}]{Nielsen2002}%
  \BibitemOpen
  \bibfield  {author} {\bibinfo {author} {\bibfnamefont {Michael~A.}\
  \bibnamefont {Nielsen}},\ }\bibfield  {title} {\enquote {\bibinfo {title} {{A
  simple formula for the average gate fidelity of a quantum dynamical
  operation}},}\ }\href {\doibase 10.1016/S0375-9601(02)01272-0} {\bibfield
  {journal} {\bibinfo  {journal} {Phys. Lett. A}\ }\textbf {\bibinfo {volume}
  {303}},\ \bibinfo {pages} {249--252} (\bibinfo {year} {2002})}\BibitemShut
  {NoStop}%
\bibitem [{\citenamefont {Etingof}\ \emph {et~al.}(2009)\citenamefont
  {Etingof}, \citenamefont {Golberg}, \citenamefont {Hensel}, \citenamefont
  {Liu}, \citenamefont {Schwendner}, \citenamefont {Vaintrob},\ and\
  \citenamefont {Yudovina}}]{Etingof2009}%
  \BibitemOpen
  \bibfield  {author} {\bibinfo {author} {\bibfnamefont {Pavel}\ \bibnamefont
  {Etingof}}, \bibinfo {author} {\bibfnamefont {Oleg}\ \bibnamefont {Golberg}},
  \bibinfo {author} {\bibfnamefont {Sebastian}\ \bibnamefont {Hensel}},
  \bibinfo {author} {\bibfnamefont {Tiankai}\ \bibnamefont {Liu}}, \bibinfo
  {author} {\bibfnamefont {Alex}\ \bibnamefont {Schwendner}}, \bibinfo {author}
  {\bibfnamefont {Dmitry}\ \bibnamefont {Vaintrob}}, \ and\ \bibinfo {author}
  {\bibfnamefont {Elena}\ \bibnamefont {Yudovina}},\ }\href {\doibase
  10.1090/stml/059} {\emph {\bibinfo {title} {{Introduction to representation
  theory}}}}\ (\bibinfo  {publisher} {American Mathematical Society},\ \bibinfo
  {address} {Providence, RI},\ \bibinfo {year} {2009})\BibitemShut {NoStop}%
\bibitem [{\citenamefont {Braun}\ \emph {et~al.}(2014)\citenamefont {Braun},
  \citenamefont {Giraud}, \citenamefont {Nechita}, \citenamefont {Pellegrini},\
  and\ \citenamefont {{\v{Z}}nidari{\v{c}}}}]{Braun2014}%
  \BibitemOpen
  \bibfield  {author} {\bibinfo {author} {\bibfnamefont {Daniel}\ \bibnamefont
  {Braun}}, \bibinfo {author} {\bibfnamefont {Olivier}\ \bibnamefont {Giraud}},
  \bibinfo {author} {\bibfnamefont {Ion}\ \bibnamefont {Nechita}}, \bibinfo
  {author} {\bibfnamefont {Cl{\'{e}}ment}\ \bibnamefont {Pellegrini}}, \ and\
  \bibinfo {author} {\bibfnamefont {Marko}\ \bibnamefont
  {{\v{Z}}nidari{\v{c}}}},\ }\bibfield  {title} {\enquote {\bibinfo {title} {{A
  universal set of qubit quantum channels}},}\ }\href {\doibase
  10.1088/1751-8113/47/13/135302} {\bibfield  {journal} {\bibinfo  {journal}
  {J. Phys. A}\ }\textbf {\bibinfo {volume} {47}},\ \bibinfo {pages} {135302}
  (\bibinfo {year} {2014})}\BibitemShut {NoStop}%
\bibitem [{\citenamefont {Horn}\ and\ \citenamefont
  {Johnson}(2013)}]{Horn2013}%
  \BibitemOpen
  \bibfield  {author} {\bibinfo {author} {\bibfnamefont {Roger~A.}\
  \bibnamefont {Horn}}\ and\ \bibinfo {author} {\bibfnamefont {Charles~R.}\
  \bibnamefont {Johnson}},\ }\href@noop {} {\emph {\bibinfo {title} {{Matrix
  analysis}}}},\ \bibinfo {edition} {2nd}\ ed.\ (\bibinfo  {publisher}
  {Cambridge University Press},\ \bibinfo {address} {Cambridge, MA},\ \bibinfo
  {year} {2013})\BibitemShut {NoStop}%
\bibitem [{\citenamefont {Aliprantis}\ and\ \citenamefont
  {Burkinshaw}(1998)}]{Aliprantis1998}%
  \BibitemOpen
  \bibfield  {author} {\bibinfo {author} {\bibfnamefont {Charalambos~D.}\
  \bibnamefont {Aliprantis}}\ and\ \bibinfo {author} {\bibfnamefont {Owen}\
  \bibnamefont {Burkinshaw}},\ }\href@noop {} {\emph {\bibinfo {title}
  {{Principles of Real Analysis}}}},\ \bibinfo {edition} {3rd}\ ed.\ (\bibinfo
  {publisher} {Academic Press},\ \bibinfo {address} {San Diego, CA},\ \bibinfo
  {year} {1998})\BibitemShut {NoStop}%
\end{thebibliography}%

\end{document}